\documentclass[12pt]{article}
\pdfoutput=1
\usepackage{array}
\usepackage{longtable}%
\usepackage{amsthm,amsmath,amsfonts,bbm,bm,enumitem,tikz,verbatim,multirow,multicol,arydshln,amssymb,makecell,float}
\usepackage[most]{tcolorbox}
\usepackage[pass]{geometry}
\usepackage{graphicx}
\usepackage{enumerate}
\usepackage{natbib}
\usepackage{url}
\usepackage[flushleft]{threeparttable}
\usepackage{caption}
\usepackage{subcaption}

\tikzstyle{process} = [rectangle, 
text centered, 
minimum width=2cm, 
text width=1.5cm, 
draw=white, 
fill=none]

\usepackage{booktabs}
\usepackage{pgfplots}
\pgfplotsset{compat=1.16,every axis/.style={width=5cm}}
\usepgfplotslibrary{fillbetween}
\usetikzlibrary{pgfplots.groupplots} 

\newtheoremstyle{normal}%
  {3pt}
  {3pt}
  {\normalfont}
  {}
  {\bfseries}
  {.}
  {.4em}
  {}

\theoremstyle{normal}

\newtheorem{theorem}{Theorem}[section]
\newtheorem{lemma}[theorem]{Lemma}
\newtheorem{corollary}[theorem]{Corollary}

\newtheorem{definition}{Definition}
\newtheorem{prop}[theorem]{Proposition}
\newtheorem{example}{Example}
\newtheorem*{examplenonumber}{Example}

\newtheorem{remark}{Remark}
\newtheorem{condition}{Condition}
\newtheorem{assumption}{Assumption}
\usepackage[colorlinks,citecolor=blue,urlcolor=blue,hypertexnames=false]{hyperref}%

\newcommand{\vZero}{\bm{0}}

\newcommand{\vB}{\bm{B}}
\newcommand{\vD}{\bm{D}}
\newcommand{\vG}{\bm{G}}\newcommand{\vH}{\bm{H}}

\newcommand{\vL}{\bm{L}}\newcommand{\vM}{\bm{M}}
\newcommand{\vQ}{\bm{Q}}
\newcommand{\vT}{\bm{T}}
\newcommand{\vV}{\bm{V}}\newcommand{\vW}{\bm{W}}
\newcommand{\vX}{\bm{X}}\newcommand{\vY}{\bm{Y}}\newcommand{\vZ}{\bm{Z}}
\newcommand{\va}{\bm{a}}\newcommand{\vb}{\bm{b}}\newcommand{\vc}{\bm{c}}
\newcommand{\vd}{\bm{d}}\newcommand{\ve}{\bm{e}}

\newcommand{\vq}{\bm{q}}
\newcommand{\vu}{\bm{u}}
\newcommand{\vv}{\bm{v}}\newcommand{\vw}{\bm{w}}
\newcommand{\vz}{\bm{z}} 

\newcommand{\bR}{\mathbb{R}}

\newcommand{\cS}{\mathcal{S}}

\newcommand{\cA}{\mathcal{A}}
\newcommand{\cF}{\mathcal{F}}

\newcommand{\cM}{\mathcal{M}}
\newcommand{\cB}{\mathcal{B}}

\newcommand{\cD}{\mathcal{D}}

\newcommand{\vbeta}{\bm{\beta}}
\newcommand{\vOmega}{\bm{\Omega}}

\newcommand{\vGamma}{\bm{\gamma}}
\newcommand{\veta}{\bm{\eta}}

\newcommand{\var}{\operatorname{var}}

\newcommand{\cov}{\text{cov}}

\newcommand{\tCR}{\text{CR}}
\newcommand{\tDC}{\text{DCAR}}
\newcommand{\tCC}{\text{CCAR}}
\newcommand{\tMC}{\text{MixCAR}}

\newcommand{\bP}{\mathbf{P}}
\newcommand{\bI}{\mathbf{I}}
\newcommand{\bPi}{\boldsymbol{\Pi}}
\newcommand{\bell}{\boldsymbol{\ell}}
\newcommand{\bpi}{\boldsymbol{\pi}}
\newcommand{\be}{\mathbf{e}}

\newcommand{\bone}{\bm{1}}
\newcommand{\bigtimes}{\mathop{\vcenter{\hbox{\large$\times$}}}}
\newcommand{\vZt}{\widetilde{\bm{Z}}}
\newcommand{\SigZd}[1]{\bm{\Sigma}_{\bm{Z},#1}}
\newcommand{\SigZtd}[1]{\bm{\Sigma}_{\widetilde{Z},#1}}
\newcommand{\bin}{\mathsf{b}}
\newcommand{\stind}[1]{\chi_{#1}}

\newcommand{\bT}{\bm{T}}
\newcommand{\bW}{\bm{W}}
\newcommand{\bgamt}{\widetilde{\bm{\gamma}}}
\newcommand{\are}{\operatorname{ARE}}
\newcommand{\aZ}{\mathrm{Z}}
\newcommand{\aS}{\mathrm{S}}
\newcommand{\aCM}{\mathrm{CM}}
\newcommand{\bbetaW}{\bm{\beta}_{W}}

\newcommand{\PL}{P_{\Lambda}}
\newcommand{\piL}{\pi_{\Lambda}}

\begin{document}

\def\spacingset#1{\renewcommand{\baselinestretch}%
{#1}\small\normalsize} \spacingset{1}
\addtocontents{toc}{\protect\setcounter{tocdepth}{-1}}

\title{\bf Discretization in covariate-adaptive randomization: gains and losses}
\author{Zixuan Zhao and
  Feifang Hu\thanks{Corresponding author (feifang@gwu.edu)} \\
  Department of Statistics, George Washington University}
\maketitle

\bigskip
\begin{abstract}
Covariate-adaptive randomization (CAR) is widely implemented in clinical trials to balance prognostic covariates across treatment arms. Continuous covariates are often discretized into strata in practice, yet their consequences are not clearly understood. This paper provides a comprehensive study of the impact of discretization on both the CAR design process and the inferential results thereafter. We establish the asymptotic properties of both imbalance measures and treatment effect estimators under discretized and non-discretized settings. Practical recommendations are given on when and how discretization should be employed. We show that discretization in design is generally recommended, as it enhances robustness against model misspecification. However, if the true model is known, the most efficient strategy is to balance covariates according to that model in the design. The theoretical results are corroborated by extensive simulation studies and an empirical application to a diabetes trial dataset. Together, the results clarify the gains and losses of discretization in CAR and pave the way for learning impact of discretization to other designs and beyond.
\end{abstract}

\noindent%

\vfill

\spacingset{1.9}
\section{Introduction}\label{sec:intro}

Randomization in clinical trials serves to balance both observed and unobserved confounders, with covariate-adaptive randomization (CAR) being a widely adopted method to improve comparability across treatment arms \cite{taves2010use,hu2014adaptive}. Existing CAR methods can be broadly categorized into three classes: discrete CAR (DCAR), continuous CAR (CCAR), and mixed CAR (MixCAR). DCAR approaches apply adaptive randomization to categorical or discretized versions of continuous covariates. In fact, a systematic review shows that $82\%$ of the randomized controlled trials reported in PubMed $2014$ used DCAR \cite[Table 2]{lin2015pursuit}. More recently, CCAR and MixCAR methods are developed. These methods offer more flexible approaches for complex trial designs. Table \ref{tab:CAR summary} summarizes representative CAR methods. For more details on the design of CAR procedures, see \cite{hu2014adaptive,rosenberger2015randomization} and discussion therein.

\begin{table}[H]
\centering
\caption{Summary of related literature in CAR design and analysis.}
{%
    \fontsize{9pt}{11pt}\selectfont
\begin{tabular}{cccc}
\toprule
\textbf{Design Name}& \textbf{Acronym}&\textbf{Reference}& \textbf{Design Type} \\
\midrule
Stratified permuted block design&STR-PB&\cite{zelen1974randomization}&DCAR\\
Stratified biased coin design&STR-BC&\cite{wei1978adaptive}&DCAR\\
Pocock and Simon's method&PS&\cite{pocock}&DCAR\\
Hu and Hu's design&HH&\cite{Hu2012}&DCAR\\
Adaptive Randomization via Mahalanobis distance & ARM&\cite{qin2022adaptive}&CCAR\\
Covariate means and covariance balancing&COV &\cite{ma2024new}&CCAR\\
Kernel balancing & KER &\cite{ma2024new}&CCAR\\
MIX procedure & MIX &\cite{liu2025} & MixCAR\\
\bottomrule
\end{tabular}}
\label{tab:CAR summary}
\end{table}

Baseline covariates, which are often collected prior to randomization, are encouraged for use in both design and analysis by regulatory agencies. The recent FDA (US Food and Drug Administration) and EMA (European Medicines Agency) guidelines \cite{guidance2018adaptive,guidelinefda,guidelineema} emphasize that covariates adjustment in both randomization design and analysis are beneficial for ``enhancing the credibility of the trial'' and ``result in a more efficient use of data''. Covariates are encouraged to be adjusted in randomization designs because CAR can be utilized to balance over baseline covariates \cite{lin2015pursuit}. In statistical analyses, covariates can be adjusted via statistical models, e.g., analysis of covariance (ANCOVA) \cite{ma2015testing,ma2020inference,ma2022regression,ye2023toward}, strata fixed effects (SFE) model \cite{bugni2018inference}, fully saturated models \cite{gu2023regression,bugni2019inference}, etc. 

Discretization of continuous covariates is an important and long standing topic, originally studied as an information-loss problem \citep{Cox1957NoteOG}. It has been discussed in various fields, including causal inference \citep{rosenbaum1984reducing}, data mining \citep{garcia2012survey}, econometrics \citep{cattaneo2024binscatter}, etc. The dominant pattern is also clear in the field of randomized experiments applications: \cite{sullivan2024categorisation} reports that $79 (75\%)$ trials between April $1,2022$ and July $31, 2022$ in leading medical journals use stratified randomization with discrete covariates. Furthermore, $32 (30\%)$ trials used discretization in design. Out of the $35$ continuous covariates discretized in design, $31$ are categorised into two levels and used in randomization. However, none of the articles adjusted for continuous covariates after stratification because of two main reasons: (i) a lack of theoretical tools for handling continuous covariates; and more importantly, (ii) substantive clinical considerations, where discretization better reflects clinically meaningful groupings. For example, in \cite{jeong2018association}, the total cholesterol level variable is discretized into tertiles. Since both too low and too high levels are not indications for better prognosis of coronary heart disease, mixing clinically abnormal patients with normal patients by a continuous variable is inappropriate. 

Nevertheless, with recent advances in CAR methodology for continuous covariates \cite{qin2022adaptive,ma2024new}, concerns have been brought up regarding the potential loss of efficiency due to discretization. Despite the prevalence of such discretization practice, its properties have been examined only empirically and many results advocate against na\"ive discretization
\cite{sullivan2024categorisation,royston2006dichotomize}. The theory for the broader class of CAR procedures remains largely unexplored.

Our contribution is multifold. This paper provides comprehensive theoretical results on discretization under CAR and offers principled guidelines on when discretization should, or should not, be used. It is shown that if the response-covariate model is well understood, strongly balancing functions that span this structure in the design achieves optimal efficiency. If not known, discretization in design with model \eqref{eq:combined working model} that fits discrete and continuous covariates together is recommended. To correct for possible conservativeness or inflation, we propose a bootstrap adjustment in Section \ref{sec:bootstrap} that is broadly valid across different CAR procedures.

The paper is organized as follows. In Section \ref{sec:CAR}, we define the problem setup and provide necessary assumptions. Section \ref{sec:model based} derives asymptotic properties related to covariate imbalance and treatment effect estimation. Section \ref{sec:impact discuss} gives an in-depth discussion of discretization. In Section \ref{sec:bootstrap}, a consistent bootstrap adjusted test is proposed. Section \ref{sec:model robust} discusses another important, model-robust inferential paradigm. Numerical studies, real data analysis are presented in Sections \ref{sec:numerical}, \ref{sec: real data}. Proofs and several important extension discussions are deferred to the Appendix. 

\textbf{Notation and conventions.} For a positive integer $n$, let $[n]=\{1,\ldots,n\}$, and let $\bigtimes_{j=1}^{p}[m_j]$ denote the Cartesian product of the sets $[m_j]$. We write $\rightsquigarrow$ and $\overset{P}{\to}$ for convergence in distribution and in probability, and $O_P(1)$ and $o_P(1)$ for boundedness and convergence to zero in probability. Throughout, $\xi$ denotes a standard normal random variable with distribution function $\Phi$ and density $\phi$; $\mathcal{N}_k(\bm{a},\bm{B})$ is the $k$-dimensional normal distribution with mean $\bm{a}$ and covariance $\bm{B}$; $\mathbbm{1}\{\cdot\}$ is the indicator function, $\bone$ (resp. $\vZero$) the one (resp. zero) vector; $\|\cdot\|$ is the Euclidean norm; $\operatorname{vec}$, $\operatorname{diag}$ and $\operatorname{sgn}$ denote the matrix vectorization, the diagonal matrix, and the sign function. For a vector $\bm{a}$ and an integer $d$, $\bm{a}_{1:d}=(a_1,\ldots,a_d)^{\top}$ denotes its leading $d$ components.

\section{Preliminaries}\label{sec:CAR}

\subsection{Randomization and discretization}\label{subsec:CAR and discretization}

Consider a randomization procedure for assigning $n$ experimental units to two treatment arms. Let $I_i$ denote the treatment indicator for the $i$-th unit, i.e., $I_i = 1$ if assigned to active treatment and $I_i = 0$ if assigned to control, where $i \in [n]$. Let $\vZ = ({Z_1,\ldots, Z_p})^\top \in \mathbb{R}^p$ be the $p$-dimensional vector of \textit{continuous} covariates and let $\vZ_i=({Z_{i,1},\ldots,Z_{i,p}})^\top$ be the realization of the $i$th unit. CAR methods aim to sequentially assign units while balancing a pre-specified function of the prognostic covariates across treatment groups. Let $\psi = (\psi_1,\ldots,\psi_{p'}): \mathbb{R}^p \to \mathbb{R}^{p'}$ be a measurable covariate map. Define a measure of imbalance as
$$
    \text{Imb}_n := \|\Lambda_n\|^2,\quad \Lambda_n := \sum_{i=1}^{n}(2I_i-1)\psi(\vZ_i)
$$
where $\|\cdot\|$ is the Euclidean norm. In this study, we consider a class of CAR procedures that follows the assignment process \citep{ma2024new}:

\noindent(Step 1) Assign the first subject to the treatment group with probability $1/2$.

\noindent(Step 2) For subsequent subjects, indexed by $i\geq 2$, compute the potential imbalance measures Imb$_i^{(1)}$ and Imb$_i^{(0)}$, assuming $I_i = 1$ and $I_i = 0$, respectively.
    
\noindent(Step 3) Assign the $i$th unit to the treatment group according to 
    $$P(I_i = 1\mid \vZ_1,\ldots,\vZ_i,I_1,\ldots, I_{i-1})=\begin{cases}
    \rho & \text{if }\text{Imb}_i^{(1)}<\text{Imb}_i^{(0)},\\
    0.5 & \text{if }\text{Imb}_i^{(1)}=\text{Imb}_i^{(0)},\\
    1-\rho & \text{if }\text{Imb}_i^{(1)}>\text{Imb}_i^{(0)},
    \end{cases}$$
    where $0.5<\rho\leq 1$ is some pre-specified biased-coin probability. 
    
\noindent(Step 4) Repeat \textbf{Steps 2,3} until all subjects are allocated.

\begin{remark}\label{remark:rho}
    Typically, one chooses $\rho\in[0.75,0.95]$ as suggested in \cite[Remark 2.1]{hu2020theory}. The procedure can potentially be applied to high-dimensional $\vZ$ or other data types like images by modifying the covariate mapping $\psi$. Further, the assignment rule accommodates two-arm randomization only. The extension to multi-arm randomization is relegated to Appendix \ref{app:multiarm}.
\end{remark}

The framework encompasses a wide range of randomization procedures. To further categorize these procedures, we call that one assignment procedure as DCAR (resp. CCAR) if $(\psi_2,\ldots,\psi_{p'})$ is a discrete (resp. continuous) random vector. If $(\psi_2,\ldots,\psi_{p'})$ is neither discrete nor continuous, the design is called MixCAR, which is short for mixed CAR. The exclusion of $\psi_1$ in above definition is used to account for overall imbalance, which will be a constant function, as will be illustrated in Examples \ref{example: COV}, \ref{example: HuHu}, or \ref{example: mixed} in Appendix \ref{app:discussions on MixCAR}. We now examine several representative examples for each type.

\begin{example}[COV procedure \cite{ma2024new}]\label{example: COV}
    The COV procedure is a CCAR procedure with feature map
    $\psi(\vZ_i) = (\sqrt{\omega_0},\sqrt{\omega_1} \vZ_i^\top, \sqrt{\omega_2}{\operatorname{vec}}(\vZ_i\vZ_i^\top)^\top)^\top,\omega_0,\omega_1,\omega_2\geq 0, i\in[n]$
    where ${\operatorname{vec}}(\cdot)$ stacks the columns of a matrix into a column vector.
\end{example}

It is common in practice to discretize continuous covariates and balance across the categories. The discretization is built from marginal binning maps $\bin_j:\mathbb{R}\to[m_j]$, which partition each margin $Z_{j}$ into $m_j$ bins, collected into $\bin=(\bin_1,\ldots,\bin_p):\mathbb{R}^p\to\bigtimes_{j=1}^{p}[m_j]$, where $\bigtimes$ denotes the Cartesian product. Composing $\bin$ with a fixed bijective enumeration of $\bigtimes_{j=1}^{p}[m_j]$ into $[m]$, where $m=\prod_{j=1}^pm_j$ is the total number of strata, yields the stratum label $S:\mathbb{R}^p\to[m]$. To illustrate, suppose $\vZ_i$ follows $\mathcal{N}_2(\va,\vB)$. Then,  the dichotomization of two Gaussian marginals at $0$ will give four strata.

\begin{example}[Hu and Hu's method \cite{Hu2012}]\label{example: HuHu}
    With non-negative weights $\omega_0$, $\{\omega_{j\ell}: j\in[p],\ell\in[m_j]\}$ and $\{\omega_s: s\in[m]\}$, consider the following feature map 
    \begin{equation}\label{eq:HuHu feature map}
        \psi(\vZ_i) = (\sqrt{\omega_0},\ldots,{\sqrt{\omega_{j\ell}}\mathbbm{1}\{\bin_{j}(Z_{i,j})=\ell\}},\ldots, \ldots,\sqrt{\omega_s}\mathbbm{1}\{S_i = s\},\ldots)^\top,\quad i\in[n],
    \end{equation}
    where $\mathbbm{1}\{\bin_{j}(Z_{i,j})=\ell\}$ is the margin indicator for the $i$th unit's $j$th margin equaling $\ell$, and $\mathbbm{1}\{S_i=s\}$ is the stratum indicator for the $i$th unit's stratum equaling $s$. 
\end{example}
If $\omega_0=\omega_s = 0$ for all $s$, the procedure reduces to PS, moreover, if $\rho=1$, then the procedure reduces to Tave's minimization \cite{taves}. If $\omega_0=\omega_{j\ell}=0$ for all $j$ and $\ell$, the procedure reduces to the stratified biased coin method \cite{efron1971forcing,baldi}. In practice, $\vZ_i$ may not be fully continuous as assumed. See Appendix \ref{app:discussions on MixCAR} for discussion on MixCAR.

\subsection{Assumptions}

The following assumptions regarding $\vZ$ and $\psi$ are assumed throughout the paper.

\begin{assumption}\label{assume:independent copy}
The covariates $\vZ_i$, $i\in[n]$, are independent copies of a generic vector
$\vZ=(Z_1,\ldots,Z_p)^{\top}$ with $E(\vZ)=\vZero$, and (1) the number of strata $m$ is fixed, and $q_s=P(S=s)>0$ for every $s\in[m]$; (2) $\bm{\Sigma}_{\vZt}=\var(\vZt)=E\{\var(\vZ\mid S)\}$ is positive definite, where $\vZt:=\vZ-E(\vZ\mid S)$ denotes the within-stratum centered covariate vector.
\end{assumption}

\begin{assumption}\label{assume:bounded moments}
    The feature map $\psi$ satisfies $E[\|\psi(\vZ)\|^{2+\upsilon}]<\infty$ for some $\upsilon>0$.
\end{assumption}
Without loss of generality, the variables for consideration are centered. Assumption \ref{assume:independent copy} (2) holds if $\vZ$ is non-constant in every strata and (2) implies $\vZ$ admits positive definite covariance. Assumption \ref{assume:bounded moments} is a mild assumption on the moment of $\psi$.

Covariate imbalances are defined to assess how balanced a design is. For a fixed function $f:\bR^{p}\to \bR$ and a randomization method $(r)\in\{\tCR,\tCC,\tDC,\tMC\}$, we define the (general) covariate imbalance as 
\begin{equation}\label{eq:general imbalance}
    D_n^{(r)}(f) = \sum_{i=1}^{n}(2I_i^{(r)}-1) f(\vZ_i).
\end{equation}
This type of imbalance is quite general and includes, e.g., the overall ($f\equiv 1$), within-strata ($f=\stind{s}= \mathbbm{1}\{{S}=s\}$), and mean imbalances ($f(\vz)=z_j$) as special cases. The following definition is used throughout the paper:
\begin{definition}\label{def:balanced by design}
If $D_n^{(r)}(f) = o_P(\sqrt{n})$ for some fixed $f$ and design $r$, then we say $f$ is strongly balanced by $r$.
\end{definition}
The notion of strong balance originates from \cite{bugni2018inference,bugni2019inference}. In theory, $o_P(\sqrt{n})$ is not the fastest rate. If a further condition on the distribution of $\psi$ can be assumed, then the design may strongly balance some feature at rate $O_P(1)$ \citep[Assumption 4 \& Theorem 3.4]{ma2024new}. For inferential considerations, $o_P(\sqrt{n})$ is fast enough. The following condition is a result of Assumptions \ref{assume:independent copy}, \ref{assume:bounded moments} and the assignment mechanism. 
\begin{condition}\label{assume: dgp}
    Given a design $r$ and a feature map $\psi:\bR^p\to \bR^{p'}$, $\psi_j$ will be strongly balanced by $r$ for any $j\in [p']$ and the overall imbalance $D_n=\sum_{i=1}^{n}(2I_i-1)=o_P(\sqrt{n})$. Additionally, all DCAR will balance $\stind{s}(\vz) = \mathbbm{1}\{S=s\}, s\in[m]$ strongly and all CCAR considered will balance $f(\vz)=z_j$ strongly for any $j\in[p]$.
\end{condition}
Condition \ref{assume: dgp} is required for the technical proofs and is satisfied by a broad collection of CAR procedures including but not limited to: DCARs like STR-PB \cite{zelen1974randomization}, STR-BC \cite{baldi}, HH with $\omega_s>0$ in Example \ref{example: HuHu}; CCARs like the COV procedure with $\omega_0>0$ procedure in Example \ref{example: COV}, the KER procedure in \cite[Section 2.3]{ma2024new}. CR does not satisfy Condition \ref{assume: dgp}, it is therefore considered separately as the benchmark method. Strictly speaking, only Assumption \ref{assume:independent copy} and Condition \ref{assume: dgp} are needed for technical results under DCAR and CCAR. Some other common DCAR variants satisfying Condition \ref{assume: dgp} (and hence applicable to our results) are catalogued in Appendix \ref{app:discuss on dgp}.

\subsection{Convergence rate for covariate imbalances}\label{sec:instability}
The impact of discretization in the design stage is examined by the convergence property of the general imbalance \eqref{eq:general imbalance}. The convergence rate of $D_n^{(r)}(f)$ is shown to be very crucial for the subsequent statistical inference, see, e.g., Theorem \ref{theorem:combined}.

\begin{theorem}\label{theorem: design property}
    For a fixed function $f:\bR^p\to \bR$ such that $E[f^2(\vZ)]<\infty$, and any design $r\in\{\tCR,\tDC,\tCC\}$, we have $n^{-1/2}D_n^{(r)}(f)\rightsquigarrow \sigma_f^{(r)}\xi$
    where $(\sigma_f^{(\tCR)})^2 = E[f^2(\vZ)]$. If $f$ is strongly balanced, then $\sigma^{(\tDC)}_f=\sigma^{(\tCC)}_f=0$. Otherwise, 
    \begin{align*}
        &(\sigma_f^{(\tDC)})^2 = E\{\var[f(\vZ)\mid {S}]\},\\
        &(\sigma_f^{(\tCC)})^2=
        E[f^2(\vZ)] + 2{\pi_{\Lambda}(\zeta_f)},
    \end{align*}
    where ${\pi_{\Lambda}(\zeta_f)}$ is defined in Equation \eqref{eq:pi_Lambda f} in Appendix \ref{app:notation}.
\end{theorem}
\begin{proof}
    See Appendix \ref{app:design prop}.
\end{proof}

The terms ${\pi_{\Lambda}(\zeta_f)}$ are associated with the invariant distributions of the Markov chains $(\Lambda_i)_{i\geq 1}$ under CCAR. In general, it may not admit closed form expression and can be either negative or positive. Result from Theorem \ref{theorem: design property} is important. Even when $f$ is not strongly balanced under DCAR, the asymptotic variance is guaranteed to reduce compared with CR. The result could be viewed as a generalization and extension of existing theoretical results in literature on DCAR \cite{liu2022unobserved,liu2024unobserved}. On the contrary, CCAR cannot outperform CR uniformly. Such adverse phenomenon is called the `variance inflation' issue. Similar issues are also observed for MixCAR as well. To quantify this phenomenon theoretically, we provide a specific example.

\begin{prop}[Less is more]\label{prop:scalar-ccar}
Let the CAR be run with a continuous covariate map $\psi:\mathbb{R}^p\to\mathbb{R}$ such that $P\{\psi(\vZ)=0\}=0$. Then the worst-case inflation is
$$
\sup_{f:0<E[f^2(\vZ)]<\infty}
\frac{\bigl(\sigma_f^{(\tCC)}\bigr)^2}{\bigl(\sigma_f^{(\tCR)}\bigr)^2}
= 1+\frac{\var\bigl(|\psi(\vZ)|\bigr)}{\{E|\psi(\vZ)|\}^{2}},
$$
and the supremum can be attained, uniquely up to nonzero constant multiples. Such $f$ with exists if and only if $|\psi(\vZ)|$ is not almost surely constant.
\end{prop}
\begin{proof}
See Appendix \ref{app:design prop}.
\end{proof}
Proposition \ref{prop:scalar-ccar} bounds the largest possible variance inflation of CCAR relative to CR for a special case. The source of the variance inflation originates because the assignment probability depends only on $\operatorname{sgn}(\Lambda\psi(\vZ))$. However, the updated covariate imbalance, $\Lambda$, is determined by the magnitude of $\psi(\vZ)$. We therefore conjecture that a smooth assignment function may mitigate this problem, but it is left for future research. Proposition \ref{prop:scalar-ccar} also says that, if and only if the scalar $\psi$ is almost surely constant (e.g., Efron's biased coin design \citep{efron1971forcing}), then there would never be an inflation problem. The reason is as follows: $f$ is uncentered, thus
$$D_n^{(r)}(f) = \sum_{i=1}^{n}\underbrace{(2I_i^{(r)}-1)\{f(\vZ_i)-E[f(\vZ)]\}}_{\text{Centered, $f$ direction}}+\underbrace{D_n^{(r)}E[f(\vZ)]}_{\text{Constant direction}}.$$
The constant direction is bounded in probability by the biased-coin method. In this view, balancing a non-constant direction $\psi$ is worse than balancing a `null' direction in terms of imbalance defined by $f$. Hence, `balancing less' is `balancing more'. We now provide a working example.

\begin{example}[Explicit variance inflation with a normal
covariate]\label{ex:gaussian-inflation}
Let $p=1$, $Z\sim N(0,1)$, $\psi(z)=z$, and let $\rho\in(1/2,1]$ be
arbitrary. Take $f^{\ast}(z)=\sqrt{\pi/2}\operatorname{sgn}(z)-z.$ Theorem~\ref{theorem: design property} together with the variance identity~\eqref{eq:scalar-var} give $\bigl(\sigma_{f^{\ast}}^{(\tCR)}\bigr)^2=\pi/2-1, \bigl(\sigma_{f^{\ast}}^{(\tCC)}\bigr)^2
=\pi(\pi/2-1)/2,$ showing a $57\%$ inflation of the limiting variance. Indeed in this example, $f^{\ast}(z)$ is, $(\pi/2-1)\sqrt{2/\pi}$ times
the unique maximizer of the inflation ratio defined in the proof of Proposition \ref{prop:scalar-ccar}. 
\end{example}

\begin{remark}\label{remark:inflation conjecture}
Whether the inflation persists with $p'\geq 2$ is not discussed by
Proposition~\ref{prop:scalar-ccar}. In higher dimensions, the associated vector Poisson equation is hard to solve. We conjecture that similar variance inflation still persists, and report supporting numerical evidence in Example \ref{example:instability} in Appendix \ref{app:additional sim}.
\end{remark}

\section{Model-based inference} \label{sec:model based}

\subsection{A general model}\label{subsec:setup}
For statistical inference, we assume the following general response model
\begin{equation}\label{eq:true model}
    Y_i = \theta_1 I_i +\theta_0(1-I_i) + g(\vZ_i) +\varepsilon_i,
\end{equation}
the error terms $\varepsilon_i'$s are independent and identically distributed, independent of all covariates and randomization variables, with $E(\varepsilon_i) = 0$, and $\var(\varepsilon_i)=\sigma_{\varepsilon}^2$. The function $g:\bR^p\to \bR$ is unknown and possibly non-linear in $\vZ_i$. Note that aside from $\vZ_i$, which is used in the randomization, it is very likely that $Y_i$ also depends on additional covariates $\vX_i$. The extension is straightforward, a brief discussion can be found in Appendix \ref{sec:additional covariates}.

We are interested in the estimation and hypothesis testing of the treatment effect $\theta_1-\theta_0$ based on the following two (possibly misspecified) linear working models
\begin{equation}\label{eq:working model 1}
    E(Y_i\mid I_i,\vZ_i)= \vG_{i,d_a}^{(r,a)}\vbeta_{d_a}^{(a)}, 
\end{equation}
whose fitted covariates and parameters are given by
\begin{align*}
    &\vG_{i,d_Z}^{(r,\aZ)} = (I_i^{(r)},1-I_i^{(r)},{Z_{i,1}},\ldots,{Z_{i,d_Z}})^\top, \quad\vbeta_{d_Z}^{(\aZ)} = (\theta_1,\theta_0,\vGamma^\top)^\top\\
    &\vG_{i,d_S}^{(r,\aS)} = (I_i^{(r)},1-I_i^{(r)},\stind{1}(\vZ_i),\ldots,\stind{d_S}(\vZ_i))^\top, \quad\vbeta_{d_S}^{(\aS)} = (\theta_1,\theta_0,\veta^\top)^\top
\end{align*}
for some $d_Z\leq p, d_S\leq m-1$, that is, not all covariates are fitted. The function $\stind{j}(\vZ_i) =\mathbbm{1}\{S(\vZ_i)=j\}$ is the strata indicator of some pre-specified discretization scheme $S$ and the superscript $a\in \{\aZ,\aS\}$ denotes analysis strategy, fitting either ANCOVA with $\vZ_i$ or SFE with $S_i$. Subscripts $d_Z,d_S$ are used to denote the number of covariates fitted in the working models. Without loss of generality, the fitted covariates are taken to be the first $d_Z, d_S$ coordinates.

After assignment, a Wald-type test is usually used to test $H_0:\theta_1-\theta_0=0$ versus $H_1:\theta_1-\theta_0\neq 0$ with the test statistic,
\begin{equation}\label{eq:test statistic}
    T^{(r,a)}_{n,d_a} = \left\{(\hat{\sigma}_{\varepsilon}^{(r,a)})^2{\vL_{d_a}}[(\vG_{d_a}^{(r,a)})^\top\vG_{d_a}^{(r,a)}]^{-1}(\vL_{d_a})^\top\right\}^{-1/2}{\vL_{d_a}}\hat{\vbeta}_{n,d_a}^{(r,a)},
\end{equation}
where $\vL_{d_a} = (1,-1,0,\ldots,0)$ is a row vector with $2+d_a$ entries, $\vG_{d_a}^{(r,a)} = (\vG^{(r,a)}_{1,d_a},\ldots, \vG^{(r,a)}_{n,d_a})^\top$ is the design matrix, $\hat{\vbeta}_{n,d_a}^{(r,a)}$ is the least squares (LS) estimator and $(\hat{\sigma}_{\varepsilon}^{(r,a)})^2 = (\vY-\vG_{d_a}^{(r,a)}\hat{\vbeta}_{n,d_a}^{(r,a)})^\top(\vY-\vG_{d_a}^{(r,a)}\hat{\vbeta}_{n,d_a}^{(r,a)})/(n-d_a-2)$ is the homoscedastic variance estimator for $\sigma^2_{\varepsilon}$. The test rejects the null hypothesis at level $\alpha$ if $|T^{(r,a)}_{n,d_a}|>\Phi^{-1}(1-\alpha/2)$, with $\Phi(\cdot)$ being the cumulative distribution function of $\xi$. In later presentation, we may suppress the subscript $d_a$ for simplicity unless strictly necessarily. 

There are mainly two frameworks for statistical inference after CAR. The model-based approach \cite{shao2010test,ma2015testing,ma2020inference,qin2022adaptive,ma2024new} and model-robust approach \cite{bugni2018inference,bugni2019inference,ye2023toward,gu2023regression,tu2024unified}. Customary model-based approaches in the literature assume a true model linear in both $I_i$ and $\vZ_i$ or $\psi(\vZ_i)$. Our approach is more general and sees them as a special case. The two working models \eqref{eq:working model 1} in this work are adopted as a general method for covariate adjustment \cite[Section III,B]{guidelinefda} and have been previously considered in \cite{ma2015testing,ma2020inference,bugni2018inference,ye2023toward}. 

\begin{remark}
    The response model \eqref{eq:true model} is a partially linear model, with linear components in $I_i$. There are extensive discussion on estimation of linear components effect under partially linear model, see, e.g., \cite{robinson1988root, liang2006estimation, chernozhukov2018double}. In this paper, we are investigating the impact of discretization and misspecification under different designs. 
\end{remark}

\subsection{Treatment effect estimation}\label{subsec:effect estimation}
This subsection examines the impact of covariate discretization on the estimation of treatment effects. The next theorem establishes the asymptotic behavior of LS estimators. 

\begin{theorem}\label{theorem:combined}
    Assume Assumption \ref{assume:independent copy} and Condition \ref{assume: dgp} holds. For fixed $r\in\{\tCR,\tDC,\tCC\}$ and $a\in\{\aZ,\aS\},$ $$\sqrt{n}\vL(\hat{\vbeta}_n^{(r,a)}-\vbeta^{(a)}) \rightsquigarrow 2\tau^{(r,a)}\xi$$
    where the asymptotic variances $(\tau^{(r,a)})^2=\sigma_{\varepsilon}^2+(\sigma^{(r,a)})^2$, $(\sigma^{(r,a)})^2$ are defined in Equations \eqref{eq:varCR}-\eqref{eq:varsACC} in Appendix \ref{app:proof inference}. Additionally, if $r=\tCC$ and $g$ is strongly balanced, then $\tau^{(\tCC,\aZ)} \equiv \sigma_{\varepsilon}$.
\end{theorem}
\begin{proof}
    See Appendix \ref{app:proof inference}.
\end{proof} 

The limiting variance $\sigma^{(r,a)}$ is the variability introduced by fitting a misspecified working model to $Y_i'= g(\vZ_i)+\varepsilon_i$ where $Y_i' = Y_i-\theta_1I_i-\theta_0(1-I_i)$. Different $(r,a)$ combinations result in different asymptotic variances. The DCAR procedure always achieves higher precision than CR in terms of lower asymptotic variance regardless of the choice of $a$ (see Equations \eqref{eq:varCR}-\eqref{eq:varsACC} in Appendix \ref{app:proof inference}). Moreover, when $a=\aS$ the limiting variance under DCAR is invariant to the number of strata $d_S$ included in the analysis. Due to the association with the invariant distribution $\pi_{\Lambda}$ defined in Theorem \ref{theorem: design property}, $\sigma^{(\tCC,\aZ)}$, $\sigma^{(\tCC,\aS)}$ may not admit closed forms (see Equations \eqref{eq:varACC}-\eqref{eq:varsACC} in Appendix \ref{app:proof inference}) but may also be $0$ when $g$ is strongly balanced by the design. Previous discussion motivate the definition of optimal precision of treatment effect estimators.
\begin{definition}[Optimal precision setup]\label{def:optimal}
    The treatment effect estimation achieves optimal precision if $\sigma^{(r,a)}=0$ for some $(r,a)$.    
\end{definition}
In this ideal case, the design and analysis strategy does not introduce any extra variance to the estimator, leaving only irreducible error variance $\sigma_{\varepsilon}^2$ and is therefore called \textbf{optimal precision setup}. The optimal precision setup can be attained in two ways: either fitting the correct model or controlling the covariate imbalances introduced by misspecification. Specifically, if a correct model $E(Y_i\mid I_i,\vZ_i) = \theta_1I_i+\theta_0(1-I_i)+g(\vZ_i)$ is fitted, then $\sigma^{(r,a)}\equiv 0$ and the optimal precision will always be attained. Alternatively, if the model misspecification terms can be controlled by the randomization and analysis strategy, in the sense that $g$ is strongly balanced by the design, then the optimal precision can be reached. Such approach allows the model to be arbitrarily misspecified, i.e., even an intercept-only model $E(Y_i\mid I_i) = \theta_1 I_i+\theta_0(1-I_i)$ will attain the optimal precision setup. However, both approaches are difficult to accomplish in practice since $g$ is usually unknown.

We now provide three different sufficient conditions for the attainment of optimal precision setup. Consider three different design/analysis strategies: (I) DCAR is used with SFE model and the function $g$ is strongly balanced. (II) CCAR is used with ANCOVA model and the function $g$ is strongly balanced. (III) Function $g$ is known and $E(Y_i\mid I_i,\vZ_i) = \theta_1I_i-\theta_0(1-I_i)-g(\vZ_i)$ is fitted as the working model.
\begin{prop}[Attainment of optimal precision setup]\label{prop:opt precision}
    Under Assumptions of Theorem \ref{theorem:combined}, if one of (I), (II), or (III) holds, then the proposed strategy achieves the optimal precision setup in the sense of Definition \ref{def:optimal}.
\end{prop}
\begin{proof}
    See Appendix \ref{app:proof inference}.
\end{proof}
The sufficient conditions (I)--(II) advocate designs that balance over a large class of functions, so that the true association $g$ has a larger chance to be in the span of $\psi$ induced by design $r$. Despite the heterogeneity introduced by different randomization designs, we show in Lemma \ref{lemma:inconsistency} that, owing to model misspecification, $\hat{\theta}_1$ and $\hat{\theta}_0$ are in general inconsistent estimators for $\theta_1$ and $\theta_0$. Nevertheless, the treatment effect estimator $\hat{\theta}_1-\hat{\theta}_0$ is still consistent. The covariates effect estimator $\hat{\vGamma}$ (resp. $\hat{\veta}$) is consistent for estimating $\vGamma$ (resp. $\veta$), which could be viewed as the best linear projection of $g(\vZ_i)+\varepsilon_i$ onto the column span of the respective design matrices. Since the covariate effects are not of primary interest, we do not discuss their estimation properties in further depth in this paper.

\begin{remark}\label{remark:plr error}
    In the partially linear model literature \cite{hardle2000partially}, incorrectly estimating the non-parametric component $g$ may lead to asymptotic bias. In our setting, however, due to the balanced nature between treatment and control groups ($D_n^{(r)}(g)= O_P(\sqrt{n})$ by Theorem \ref{theorem: design property}), the asymptotic bias cancels out. To eliminate the asymptotic variance introduced by model misspecification, it must hold that $g$ is strongly balanced, i.e., $D_n^{(r)}(g)= o_P(\sqrt{n})$.
\end{remark}

\subsection{Testing hypothesis}\label{subsec:hypothesis test}
The distribution of the test statistic depends on the variance estimator, which is studied in the next result.
\begin{theorem}\label{theorem:var limit}
    For any fixed design and modeling strategy $(r,a)$ under Assumption \ref{assume:independent copy} and Condition \ref{assume: dgp}, we have $$n(\hat{\sigma}_{\varepsilon}^{(r,a)})^2 \vL[(\vG^{(r,a)})^\top\vG^{(r,a)}]^{-1}\vL^\top = 4(\tau^{(\tCR,a)})^2+o_P(1),$$
    where $\tau^{(\tCR,a)}$ is defined in Theorem \ref{theorem:combined}.
\end{theorem}
\begin{proof}
    See Appendix \ref{app:proof inference}.
\end{proof}

The probability limit of $\hat{\sigma}_{\varepsilon}^{(r,a)}$ is does not depend on $r$. In fact, the variance estimator is consistent for estimating the limiting variance under CR and tends to overestimate the limiting variance under DCAR. For most cases under CCAR, the variance is overestimated as well. However, if the covariates are as introduced in Example~\ref{ex:gaussian-inflation} and the true model \eqref{eq:true model} admits $g=f^{\ast}$ with $r=\tCC$,
$a=\aZ$, then combining Equations \eqref{eq:varCR},\eqref{eq:varACC} yields
$(\sigma^{(\tCC,\aZ)})^2 - (\sigma^{(\tCR,\aZ)})^2
\ge 2{\pi_{\Lambda}(\zeta_{f^{\ast}})}= (\pi/2-1)^{2}>0,$ where the first equality is Proposition~\ref{prop:scalar-ccar} evaluated at $g=f^{\ast}$. Hence, the variance will be underestimated in this case. The asymptotic distributions of the tests under different randomization schemes is established in the next corollary.

\begin{corollary}\label{cor: test under Rand}
     Suppose Assumption \ref{assume:independent copy} and Condition \ref{assume: dgp} hold. Under $H_0:\theta_1-\theta_0=0$, we have 
    $T_n^{(r,a)}\rightsquigarrow (\tau^{(r,a)}/\tau^{(\tCR,a)})\xi.$
    Furthermore, under $H_1:\theta_1-\theta_0=\delta/\sqrt{n}$ for some $\delta\neq 0$, we have
    $T_n^{(r,a)}\rightsquigarrow \delta/(2\tau^{(\tCR,a)})+ (\tau^{(r,a)}/\tau^{(\tCR,a)})\xi,$
    where $\tau^{(r,a)}$ are defined in Theorem \ref{theorem:combined} for $a\in \{\aZ,\aS\}$, $r\in\{\tCR,\tDC,\tCC\}$.
\end{corollary}
\begin{proof}
    The proof follows from results of Theorem \ref{theorem:combined} and Theorem \ref{theorem:var limit}.
\end{proof}

The asymptotic variance of $T_n^{(r,a)}$ under $H_0$ comprises two components. The denominators differ only based on the analysis model employed. The numerators are influenced by both randomization design and model. To evaluate type I error, the test $T_n$ is said to be (asymptotically) conservative at level $\alpha$ if under $H_0$
\begin{equation}\label{eq:test validity}
    \lim_{n\to\infty}P[|T_n|>\Phi^{-1}(1-\alpha/2)]\leq \alpha_0<\alpha.
\end{equation}
If $\alpha_0=\alpha$, then the test is said to be (asymptotically) valid. Due to Corollary \ref{cor: test under Rand}, tests under CR remain valid with controlled type I error rates. Tests under DCAR are conservative due to the overestimation of asymptotic variances. This guaranteed conservativeness is attributed to the variance reduction from balanced covariate distributions across treatment groups. Tests under CCAR are most conservative under the optimal precision setup (II), see Proposition \ref{prop:opt precision}. However, they may also introduce type I error inflation due to possible underestimation of the asymptotic variance, see, e.g., Examples \ref{ex:gaussian-inflation}, \ref{example:instability} and discussion after Theorem \ref{theorem:var limit}.  We show such inflation in Appendix \ref{app:additional sim} both theoretically and by numerical simulation. The forthcoming theorem establishes the asymptotic relative efficiency (ARE) under different scenarios.

\begin{theorem}\label{theorem:compare tests}
    Suppose Assumption \ref{assume:independent copy} and Condition \ref{assume: dgp} hold. The Pitman's ARE between two tests is given by
    ${\are}(T_n^{(r,a)},T_n^{(r',a')}) =(\tau^{(r',a')})^2/(\tau^{(r,a)})^2,$ where $\tau^{(r,a)},\tau^{(r',a')}$ are defined as in Theorem \ref{theorem:combined} for $a\in\{\aZ,\aS\},r\in\{\tCR,\tDC,\tCC\}$. Furthermore,
    ${\are}(T_n^{(\tCR,\aS)},T_n^{(\tDC,\aS)})\leq 1$
    with equality holding if, in particular, $d_S=m-1$, i.e., all strata are fitted.
\end{theorem}
\begin{proof}
    See Appendix \ref{app:proof inference}.
\end{proof}

\begin{remark}\label{remark:CR}
   Although, via Theorem \ref{theorem:compare tests}, covariate imbalance under CR can be addressed through full post-stratification, relying on such adjustment becomes increasingly unattractive when the number of strata is large relative to the sample size \citep{maratrix2013post}. In fact, empirical reviews have documented analyses that adjust for only a subset of the balancing covariates \citep{kahan2012reporting} and `no more than a few covariates should be included in the primary analysis' \citep{guidelineema}. Furthermore, some commonly used stratification factors, e.g., randomization cites, are seldomly adjusted in practice. 
\end{remark}

If ${\are}(T_n,T'_n)\in(1,\infty)$, then $T_n$ is more efficient. If ${\are}(T_n,T'_n)\in(0,1)$, then $T'_n$ is more efficient. In our case, the ARE is determined solely by the limiting variances ${\tau^{(r,a)}}$ introduced by model misspecification and achieving the highest statistical efficiency is equivalent to achieving optimal precision setup.

\section{Impact of discretization}\label{sec:impact discuss}
 In this section, the impact of discretization is investigated assuming different types of $g$ functions. Linear, quadratic and exponential functions are of most interest, as discussed in \cite{calabrese2001u}, \cite[Section 3.1]{sullivan2024categorisation}.  In practice, the function $g$ may exhibit complicated, non-linear relationships. We now present examples with more detailed discussion.
 \subsection{Known $g$ and additive}\label{subsec: additive functions}
 Suppose $g(\vz)=\sum_{j=1}^{{J}}a_jg_j(\vz), {J}>0$ is a known additive functional with coefficients $\va = (a_1,\ldots,a_{{J}})$ and components $g_j:\bR^p\to \bR$. Picking $\psi = (1,g_1,\ldots, g_{{J}})$ results in strongly balanced $g$ and thus the optimal precision, see Proposition \ref{prop:opt precision}. 
 
A linear $g$ is the case where $g(\vz) = \sum_{j=1}^{{J}}a_jz_j, {J}\leq p$ with $g_j(\vz)=z_j$. An optimal choice of feature mapping is $\psi(\vz) = (1,\vz^{\top})^{\top}$, which reduces to the CCAR procedure balancing covariate means \cite{li2019testing}. For quadratic functions $g(\vz) = \sum_{j=1}^{{J}}a_jz_j^2, {J}\leq p$, an optimal feature map is given in Example \ref{example: COV} which reduces to the COV procedure. Since the covariates are assumed continuous, balancing known $g$ would be considered a CCAR. Another special case different from above is the class of step functions $g(\vz) = \sum_{j=1}^{{J}}a_j\stind{j}(\vz), {J}\leq m-1$. An optimal feature map would be Equation \eqref{eq:HuHu feature map} and the best design is DCAR.

In summary, when $g$ is known and strongly balanced in the design with the correct strategy (I) or (II), the power will be maximized. This observation favors designs that can balance over a large family of functions, some of which are developed recently, such as the COV procedure for balancing higher moments in Example \ref{example: COV} or the kernel procedure \cite[Section 2.3]{ma2024new}. 

 \subsection{Unknown $g$}\label{subsec: other functions}
     In this example, we assume that $g$ is unknown and not strongly balanced, otherwise the discussion would be trivial and a model with intercept only achieves optimal precision \ref{prop:opt precision}.

    Under CR, due to Lemma \ref{lemma: compare TZ, TS under CR,D,C quadratic}, ANCOVA beats SFE if and only if the linear signal in $\vZ$ explains as much variance in $g(\vZ)$ as the between stratum means of $g$.  As a toy example, we compute the limiting variances under CR for a special case. Consider $Z\overset{d} = \xi$, $g(Z)= bZ^3$, the randomization is CR and $Z$ is discretized at $\{-c,c\}$ in the analysis. Then, $\gamma = E(b Z^4)/E(Z^2)=3b$, so for ANCOVA $(\sigma^{(\tCR,\aZ)}_1)^2= 6b^2.$ On the other hand, for SFE, some calculations yield
    $(\sigma^{(\tCR,\aS)}_{2})^2 = 15b^2-2\{1-{\Phi}(c)]\left[(2+c^2)\phi(c)/[1-\Phi(c)]\right\}^2b^2.$ 
    For $c=1$, we have $(\sigma^{(\tCR,\aS)}_{2})^2\approx 8.35b^2>6b^2$, so ANCOVA is preferable. A more thorough discussion on the interplay of model uncertainty introduced by misspecification can be found in, e.g.,  \cite{buja2019models}.

Under DCAR, the preference between ANCOVA and SFE is also not one-sided and it is governed by Lemma \ref{lemma: compare TZ, TS under CR,D,C quadratic}, in fact,
$$(\sigma_{d_Z}^{(\tDC,\aZ)})^2 - (\sigma_{d_S}^{(\tDC,\aS)})^2 =E\left[\var\left(\sum_{j=1}^{d_Z}{\gamma_j}{Z_j}\mid {S}\right)-2\cov\left(g(\vZ),\sum_{j=1}^{d_Z}{\gamma_j}{Z_j}\mid {S}\right)\right].$$
Consider three instructive cases: (1) If $g$ is $\sigma({S})$-measurable (which reduces to discussion of stepwise function in Subsection \ref{subsec: additive functions}), then $(\sigma_{d_Z}^{(\tDC,\aZ)})^2 - (\sigma_{d_S}^{(\tDC,\aS)})^2=E\{\var[\sum_j{\gamma_j}{Z_j}\mid {S}]\}\geq 0$ and thus SFE is preferred. (2) If $g$ is linear in $\vZ$ (which reduces to discussion of linear function in Subsection \ref{subsec: additive functions}), then $(\sigma_{d_Z}^{(\tDC,\aZ)})^2 - (\sigma_{d_S}^{(\tDC,\aS)})^2<0$ and thus ANCOVA is preferred. (3) If $E[\cov(g(\vZ),\sum_{j=1}^{d_Z}{\gamma_j}{Z_j}\mid {S})]$ is small or negative, for instance, the unconditional linear trend is positive while $g$ tends to decrease within many strata. Then, SFE still beats ANCOVA. This can be viewed as a within-stratum Simpson-Paradox pattern, we refer to this as subpopulation effect. For more numerical results on this, see Case 4 in Section \ref{sec:numerical}. 

In practice, $g$ is unknown and hard to balance strongly, also, the working models usually tend to be misspecified. In this case, CCAR's performance can be unstable and may lead to empirically inflated type I errors, see Subsection \ref{subsec:hypothesis test}. For this reason, we recommend using DCAR as the design strategy.

\subsection{A combined model}\label{subsec:combined}
From previous discussion, DCAR is preferred for unknown $g$. However, the working model selection depends heavily on our prior knowledge of the shape of $g$. Consider the following working model
\begin{equation}\label{eq:combined working model}
    E(Y_i\mid I_i,\vZ_i) = \vG_{i,d_Z,d_S}^{(r,\aCM)}\vbeta_{d_Z,d_S}^{(\aCM)},
\end{equation}
where superscript `CM' stands for combined and $\vZt=\vZ-E(\vZ\mid S)$ is as in Assumption~\ref{assume:independent copy}, 
\begin{align*}
    &\vG_{i,d_Z,d_S}^{(r,\aCM)} = (I_i^{(r)},1-I_i^{(r)},\stind{1}(\vZ_i),\ldots, \stind{d_S}(\vZ_i),{\widetilde{Z}_{i,1},\ldots, \widetilde{Z}_{i,d_Z}})^\top,\\
    &\vbeta_{d_Z,d_S}^{(\aCM)} = (\theta_1,\theta_0,\eta_1,\ldots,\eta_{d_S},\tilde{\gamma}_1,\ldots, \tilde{\gamma}_{d_Z})^\top.
\end{align*}
Then, similar to subsection \ref{subsec:setup}, $\hat{\vbeta}^{(r,\aCM)}_n$ is the usual OLS estimator and $T_{n,d_Z,d_S}^{(r,\aCM)}$ is the test statistic for testing treatment effect. Asymptotic results Theorem \ref{theorem:combined}-\ref{theorem:compare tests} hold with $\sigma^{(\tCR,\aCM)},\sigma^{(\tDC,\aCM)}$ defined in Equations \eqref{eq:varCR},\eqref{eq:varDC} in Appendix \ref{app:proof inference}. The asymptotic results will be omitted in the main paper for simplicity but proofs are given in the Appendix. 

Intuitively, the combined model captures between-stratum differences through the
stratum indicators and within-stratum smooth effects through ${\vZt}$. The next
proposition shows that this model attains the smallest asymptotic variance achievable by any linear
adjustment of this form.

\begin{prop}\label{prop:better model}
Under Assumptions of Theorem \ref{theorem:combined}. Fix $d_Z \le p$ and $d_S \le m-1$. Then, $\min_{\mathbf{b}\in\mathbb{R}^{d_Z}}E\bigl\{\var\bigl[g(\vZ)-\mathbf{b}^{\top}{\vZt}_{1:d_Z}\big| S\bigr]\bigr\},$ is attained uniquely at $\mathbf{b}=\bgamt$. Further,
\begin{align}
\bigl(\sigma^{(\tDC,\aS)}\bigr)^2-\bigl(\sigma^{(\tDC,{\aCM})}\bigr)^2
   &= \bgamt^{\top}{\SigZtd{d_Z}}\bgamt\geq 0, \label{eq:gain-S}\\
\bigl(\sigma^{(\tDC,\aZ)}\bigr)^2-\bigl(\sigma^{(\tDC,{\aCM})}\bigr)^2
   &= (\vGamma-\bgamt)^{\top}{\SigZtd{d_Z}}(\vGamma-\bgamt)\geq 0. \label{eq:gain-Z}
\end{align}
\end{prop}
\begin{proof}
    See Appendix \ref{app:proof inference}.
\end{proof}

Proposition \ref{prop:better model} shows that CM model achieves the minimal model-based limiting variance under DCAR (see Equations \eqref{eq:varDC}). Additionally, Equations \eqref{eq:gain-S} (resp. \eqref{eq:gain-Z}) show that the variance gain is negligible if and only if ${\vZt}$ has no linear impact on $Y$ (resp. adjusting for strata does not affect the association between $\vZ$, $Y$). In addition, the ARE are the same for CR and DCAR if all strata are fitted in the working model. The proof is the same as Theorem \ref{theorem:compare tests}. A summary of our recommendations is given in Table \ref{tab:summary}.

\section{Bootstrap adjustment}\label{sec:bootstrap}
The usual Wald test is conservative under DCAR and can be either conservative or anti-conservative under CCAR. We therefore use the following bootstrap variance adjustment. For a fixed randomization rule $r$ and working model $a$, generate each replicate by resampling $n$ observed pairs $(\vZ_i,Y_i)$, re-randomizing the resampled covariates with rule $r$, and refitting model $a$. Discard a replicate if the smallest eigenvalue of the design matrix is too small, as specified in Appendix~\ref{app:bootstrap}. Let $\hat{v}_{*,n,B}^{(r,a)}$ be the sample variance of the retained treatment-effect estimates. Define $T_{n,B}^{(r,a)}=
\bigl[\hat{v}_{*,n,B}^{(r,a)}\bigr]^{-1/2}(\hat{\theta}_1-\hat{\theta}_0)
,$
with $T_{n,B}^{(r,a)}=0$ if the sample variance is zero. The adjusted level-$\alpha$ test rejects when $|T_{n,B}^{(r,a)}|>\Phi^{-1}(1-\alpha/2)$.

\begin{theorem}\label{theorem:bootstrap consistency}
Fix $a\in\{\aZ,\aS,\aCM\}$, $r\in\{\tCR,\tDC\}$, and $\delta\in\bR$. Under Assumptions of Theorem \ref{theorem:combined} and suppose Assumption~\ref{assume:bootstrap}(1) holds, together with Assumption~\ref{assume:bootstrap}(2) when $r=\tDC$. Set $\theta_1-\theta_0=\delta/\sqrt n$. Along every sequence for which $\min\{n,B\}\to\infty$,
$n\hat{v}_{*,n,B}^{(r,a)}\overset{P}\to 4(\tau^{(r,a)})^2.$ Further,  $T_{n,B}^{(r,a)}\rightsquigarrow
\mathcal{N}\left(\delta/(2\tau^{(r,a)}),1\right).$
\end{theorem}
\begin{proof}
See Appendix \ref{app:bootstrap}.
\end{proof}

The idea of bootstrap adjustment is similar in spirit to \cite{shao2010test}, in which they developed a bootstrap adjusted two sample $t$-test under STR-BC. In their work, bootstrap consistency follows from classical bootstrap theory \cite{Bicke1981bootstrap} under the null because the omitted covariate effect is linear in the stratification levels. Under the alternatives, the distribution of outcome shifts and the statements cannot be inherited. Such complications are nontrivial and lead to a modified bootstrap proof with some additional assumptions in Appendix \ref{app:bootstrap}. 

\begin{remark}
    Bootstrap consistency for CCAR works empirically in Section \ref{sec:numerical}. Although it is conjectured that bootstrap can be applied to CCAR \citep[Remark 4.4]{ma2024new}, the extension requires conditional control of the invariant-distribution $\pi_{\Lambda}$. In addition, other non-bootstrap type adjustment for CCAR has been proposed in, e.g., \cite{zhang2026asymptoticpropertiesmultitreatmentcovariate}. 
\end{remark}

\section{Model-robust inference}\label{sec:model robust}

The preceding sections work under the response model \eqref{eq:true model}. The model-robust approach does not assume such a model and targets the treatment effect directly. In the potential outcomes framework \citep{imbens2010rubin}, $Y_i(k)$ denotes the potential outcome of unit $i$ under treatment $k\in\{0,1\}$, so that the observed outcome is $Y_i=I_iY_i(1)+(1-I_i)Y_i(0)$. Model-robust inference targets the average treatment effect (ATE) $\vartheta=E\{Y(1)-Y(0)\}.$

We first define the model-robust estimator. Take $\vX_i\in \bR^{p_x}$ to be the additional prognostic covariates not used in randomization and let $\hat m_{s,k}:\bR^{p_x}\to\bR$, for $k=0,1$ and
$s\in[m]$, be the estimated adjustment functions. Write $N_s$ (resp. $N_{s,k}$) as the number of units in strata $s$ (resp. stratum $s$, treatment arm $k$). The model-robust estimator is defined as
\begin{equation}\label{eq:master}
\begin{aligned}
&\hat\vartheta(\hat m)=\sum_{s\in[m]}\frac{N_s}{n}
  \bigl\{\hat\vartheta_{s,1}(\hat m)-\hat\vartheta_{s,0}(\hat m)\bigr\},\\
&\hat\vartheta_{s,k}(\hat m)=\frac{1}{N_{s,k}}\sum_{i:S_i=s,I_i^{(r)}=k}
  \bigl\{Y_i-\hat m_{s,k}(\vX_i)\bigr\}
  +\frac{1}{N_s}\sum_{i:S_i=s}\hat m_{s,k}(\vX_i),
\end{aligned}
\end{equation}
It has been shown that $\sqrt{n}(\hat{\theta}(\hat{m})-\vartheta)$ converges to a normal distribution with variance depending on the in-probability limit of $\hat{m}$. Different choices of $\hat m$ produce the different estimators proposed in the literature
\citep{bannick2023general,gu2023regression,tu2024unified,liu2023lasso,ye2021inference}. 

We emphasize that at least two features of the existing theory limit the discussion about discretization. Firstly, theoretical results for model-robust approach relies on a critical assumption that the covariates and potential outcomes are independent from assignment conditioning on strata \citep[Assumption 2.2(a)]{bugni2019inference}. A CCAR or MixCAR design does not in general satisfy this condition. Secondly, model-robust estimators of the form \eqref{eq:master} are \textbf{post-stratified}. Model-robust inference relies heavily on the strata and collapses quickly when the number of strata is relatively large \citep{ye2023toward,gu2025assumption}. Discretization in the model-robust literature is therefore exogenous. Nevertheless, there are some connections to our work. The strata only adjustment $a=\aS$ is essentially the strata fixed effect estimator in \cite{bugni2018inference}. The proposed combined model is similar to the general ANHECOVA model in \cite{ye2023toward} with stratum dummies and centered covariates. Although the motivations and theory are completely different.

\section{Numerical Studies}\label{sec:numerical}

In this section, our theoretical results on test efficiency are corroborated with simulation studies. The simulated data sets are generated according to the true model \eqref{eq:true model} with $\varepsilon_i\overset{d}=\xi$ for $i\in [n], n\in\{400,500,600\}$. The simulation results are repeated $1200$ times. We also consider different randomization procedures. More specifically, Hu and Hu's method with $\omega_s>0$ (HH) are implemented as the DCAR procedures, see R package ``carat'' \citep{ma2023carat}; COV procedures with $\omega_0=\omega_2=1, \omega_1 = p$ is implemented as the CCAR procedure. The biased coin probabilities are taken as $\rho=0.8$ for all procedures. To introduce dependence, we use Gaussian Copula. Let $\vV_i\in \bR^p$ be random variables following $\mathcal{N}_p(\vZero,\vB_{\rho_Z})$ where $\vB_{\rho_Z}$ has unit diagonal and off diagonal $\rho_Z$. Set $F_j$ to be the specified marginal distribution of $Z_j$, define $U_{i,j} = \Phi(V_{i,j})$, $Z_{i,j}=F_j^{-1}(U_{i,j})$. Here, $\rho_Z\in\{0,0.3,0.6\}$ characterizes the strength of correlation. We do not fit CM after CCAR methods because there is no obvious reason to discretize after balancing for continuous variables. 

\noindent\textbf{Case 1.} (Threshold-effect model, $n=400$) Consider two exponential random variables ${Z_{i,1}, Z_{i,2}}$ each following $\text{Exp}(1)-1, \text{Exp}(2)-1/2$ distribution, where $\text{Exp}(j)$ is an exponential distribution with rate $j$. The variables are discretized at their medians, so that $S_i:\bR^2 \to \{1,\ldots, 4\}$. The response covariate relationship is a $\sigma(S)$-measurable function $g(\vz) = 1.5\mathbbm{1}\{z_1>\log(2)-1\} + \mathbbm{1}\{z_2>\log(2)/2-1/2\}+\mathbbm{1}\{z_1>\log(2)-1\} \mathbbm{1}\{z_2>\log(2)/2-1/2\}$. For ANCOVA model, $d_Z=2$; for SFE model, $d_S=3$; for CM model, $(d_Z,d_S) = (2,3)$.

\noindent\textbf{Case 2.} (Linear-quadratic model, $n=600$) Consider two continuous covariates $({Z_{i,1}},{Z_{i,2}})$ following a two dimensional multivariate normal distribution $\mathcal{N}_2(\vZero,\vB_{\rho_Z})$. Each ${Z_{i,j}}$ is discretized at $2/5, 3/4$ quantiles so that $S:\mathbb{R}^2\to\{1,\ldots,9\}$. The response-covariate relationship $g(\vz) = 1.8z_1+0.8(z_1^2-1)+0.8(z_2^2-1)$.  For ANCOVA model, $d_Z=2$; for SFE model, $d_S=8$; for CM model, $(d_Z,d_S) = (2,8)$.

\noindent \textbf{Case 3.} (Subpopulation effects, $n=500$) Consider two continuous covariates $({Z_{i,1}},{Z_{i,2}})$ following a two dimensional multivariate normal distribution $\mathcal{N}_2(\vZero,\vB_{\rho_Z})$. Both variables are discretized marginally at $1/3,2/3$ quantiles, so that $S:\mathbb{R}^2\to\{1,\ldots,9\}$. The response-covariate relationship $g(\vz) = 2.6\{\Phi[8(z_1-\Phi^{-1}(1/3))]+\Phi[8(z_1-\Phi^{-1}(2/3))]\}-1.5z_1+0.8z_2$. For ANCOVA model, $d_Z=2$; for SFE model, $d_S=8$; for CM model, $(d_Z,d_S) = (2,8)$.

\noindent \textbf{Case 4.} (Unstable model, $n=500$) Consider two continuous covariates $({Z_{i,1}},{Z_{i,2}})$ following a two dimensional multivariate normal distribution $\mathcal{N}_2(\vZero,\vB_{\rho_Z})$. The two variables are discretized marginally at two thresholds $0$, so that $S:\mathbb{R}^2\to\{1,\ldots,4\}$. The response-covariate relationship $g(\vz) = 1.5 [\sqrt{\pi/2}\operatorname{sgn}(z_1)-z_1]+0.8z_2$. The first term is marginally orthogonal to $z_1$, similar as in Example \ref{ex:gaussian-inflation}. For ANCOVA model, $d_Z=2$; for SFE model, $d_S=3$; for CM model, $(d_Z,d_S) = (2,3)$.

\begin{table}[t!]
\centering
\caption{Type I error under \textbf{Cases} 1--5 for $\rho_Z\in\{0,0.3,0.6\}$. The units are $10^{-2}$. Entries for HH and COV are unadjusted, with the bootstrap-adjusted size in parentheses below.}
\label{tab:size main}
{\fontsize{10pt}{10.5pt}\selectfont
\renewcommand{\arraystretch}{0.92}
\setlength{\tabcolsep}{1.8pt}
\begin{tabular}{lcccccccccccc}
\toprule
 & \multicolumn{3}{c}{Case 1} & \multicolumn{3}{c}{Case 2} & \multicolumn{3}{c}{Case 3} & \multicolumn{3}{c}{Case 4} \\
\cmidrule(lr){2-4} \cmidrule(lr){5-7} \cmidrule(lr){8-10} \cmidrule(lr){11-13} 
 & $\rho_Z = 0$ & 0.3 & 0.6 & 0 & 0.3 & 0.6 & 0 & 0.3 & 0.6 & 0 & 0.3 & 0.6 \\
\midrule
    $T^{(\tCR,\aZ)}_{n,d_Z}$ & 4.2 & 5.1 & 6.0 & 4.9 & 5.8 & 5.3 & 4.8 & 4.8 & 4.8 & 5.2 & 5.1 & 5.8\\
    $T^{(\tCR,\aS)}_{n,d_S}$ & 4.5 & 4.9 & 5.2 & 4.7 & 6.2 & 4.5 & 5.2 & 5.3 & 5.3 & 4.5 & 5.2 & 6.0\\
    $T^{(\tCR,\aCM)}_{n,d_Z,d_S}$ & 4.5 & 5.1 & 5.4 & 5.1 & 5.7 & 4.5  & 4.8 & 5.2 & 5.0 & 4.6 & 4.9 & 5.0\\\hline
    $T^{(\text{HH},\aZ)}_{n,d_Z}$ & \makecell{1.7\\(3.7)} & \makecell{2.3\\(4.2)} & \makecell{2.8\\(5.5)} & \makecell{2.7\\(4.2)} & \makecell{2.7\\(4.8)} & \makecell{2.0\\(5.2)}  & \makecell{5.6\\(6.3)} & \makecell{4.1\\(4.9)} & \makecell{3.9\\(5.2)} & \makecell{2.6\\(4.5)} & \makecell{2.5\\(5.2)} & \makecell{1.6\\(4.3)}\\
    $T^{(\text{HH},\aS)}_{n,d_S}$ & \makecell{4.7\\(4.8)} & \makecell{4.7\\(4.8)} & \makecell{5.0\\(4.9)} & \makecell{4.9\\(4.8)} & \makecell{5.1\\(5.3)} & \makecell{4.6\\(4.8)} & \makecell{5.7\\(6.0)} & \makecell{4.6\\(5.2)} & \makecell{5.0\\(5.4)} & \makecell{5.1\\(5.0)} & \makecell{5.2\\(5.4)} & \makecell{4.2\\(4.6)}\\
    $T^{(\text{HH},\aCM)}_{n,d_Z,d_S}$ & \makecell{4.5\\(4.5)} & \makecell{4.7\\(4.8)} & \makecell{5.2\\(5.1)} & \makecell{4.5\\(5.1)} & \makecell{4.4\\(4.7)} & \makecell{4.6\\(5.2)} & \makecell{4.8\\(5.0)} & \makecell{4.9\\(5.0)} & \makecell{5.2\\(5.5)} & \makecell{4.7\\(4.6)} & \makecell{4.3\\(4.4)} & \makecell{4.7\\(4.8)}\\\hline
    $T^{(\text{COV},\aZ)}_{n,d_Z}$ & \makecell{5.0\\(5.2)} & \makecell{5.9\\(6.0)} & \makecell{5.1\\(5.8)} & \makecell{0.0\\(4.6)} & \makecell{0.1\\(5.1)} & \makecell{0.0\\(5.8)} & \makecell{\underline{\textbf{6.4}}\\(5.2)} & \makecell{\underline{\textbf{6.8}}\\(5.2)} & \makecell{\underline{\textbf{6.0}}\\(4.5)}& \makecell{\underline{\textbf{6.7}}\\(5.2)} & \makecell{\underline{\textbf{6.3}}\\(5.2)} & \makecell{\underline{\textbf{7.0}}\\(4.6)}\\
    $T^{(\text{COV},\aS)}_{n,d_S}$ & \makecell{5.2\\(5.2)} & \makecell{5.9\\(5.9)} & \makecell{5.3\\(5.4)} & \makecell{1.3\\(4.2)} & \makecell{2.4\\(5.4)} & \makecell{2.8\\(4.8)} & \makecell{5.4\\(5.4)} & \makecell{5.4\\(5.7)} & \makecell{5.2\\(5.6)} & \makecell{4.2\\(5.8)} & \makecell{3.8\\(5.0)} & \makecell{2.9\\(4.8)}\\
\bottomrule
\end{tabular}}
\end{table}

\begin{figure}[t!]
    \centering
    
    \begin{subfigure}[b]{\textwidth}
        \centering
        \includegraphics[width=\linewidth]{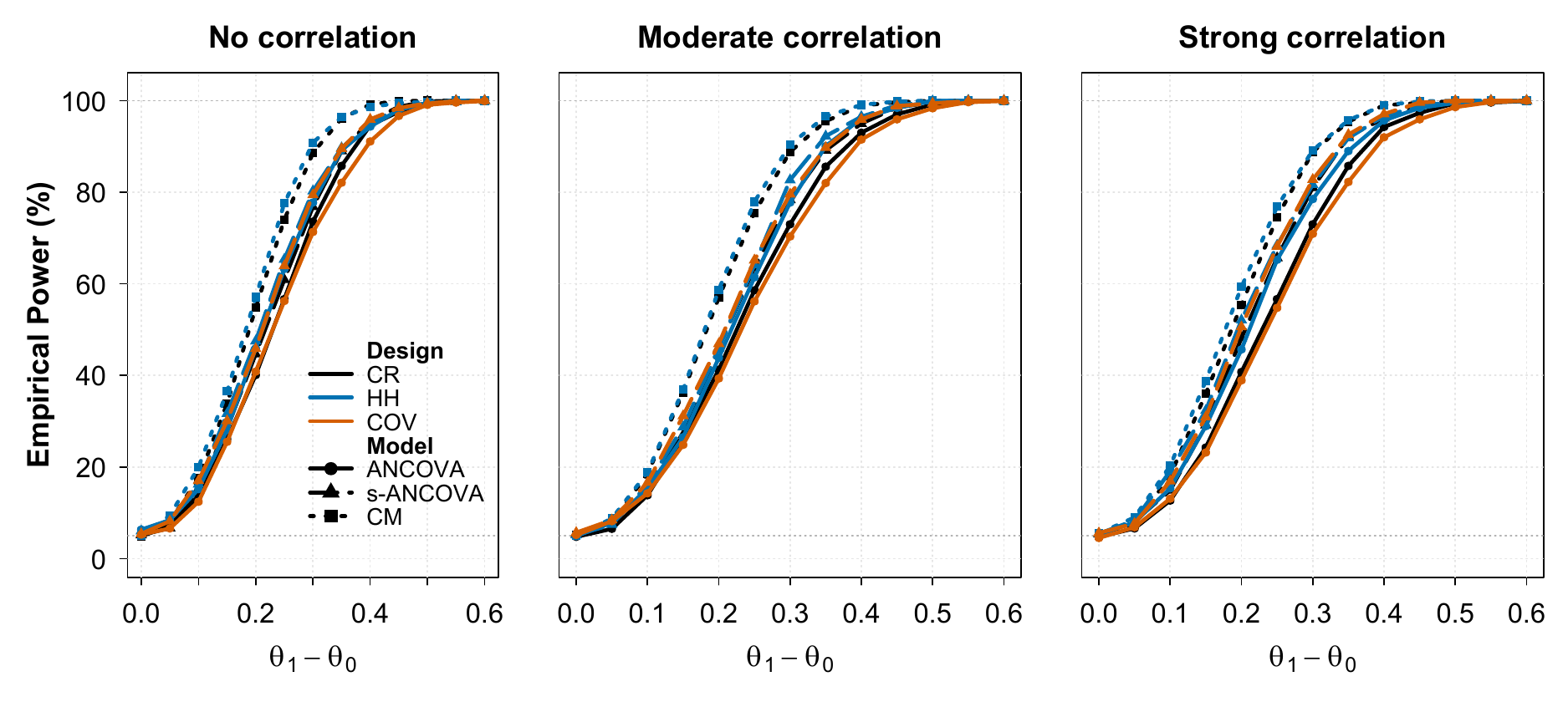}
        
    \end{subfigure}
    
    \bigskip
    
    \begin{subfigure}[b]{\textwidth}
        \centering
        \includegraphics[width=\linewidth]{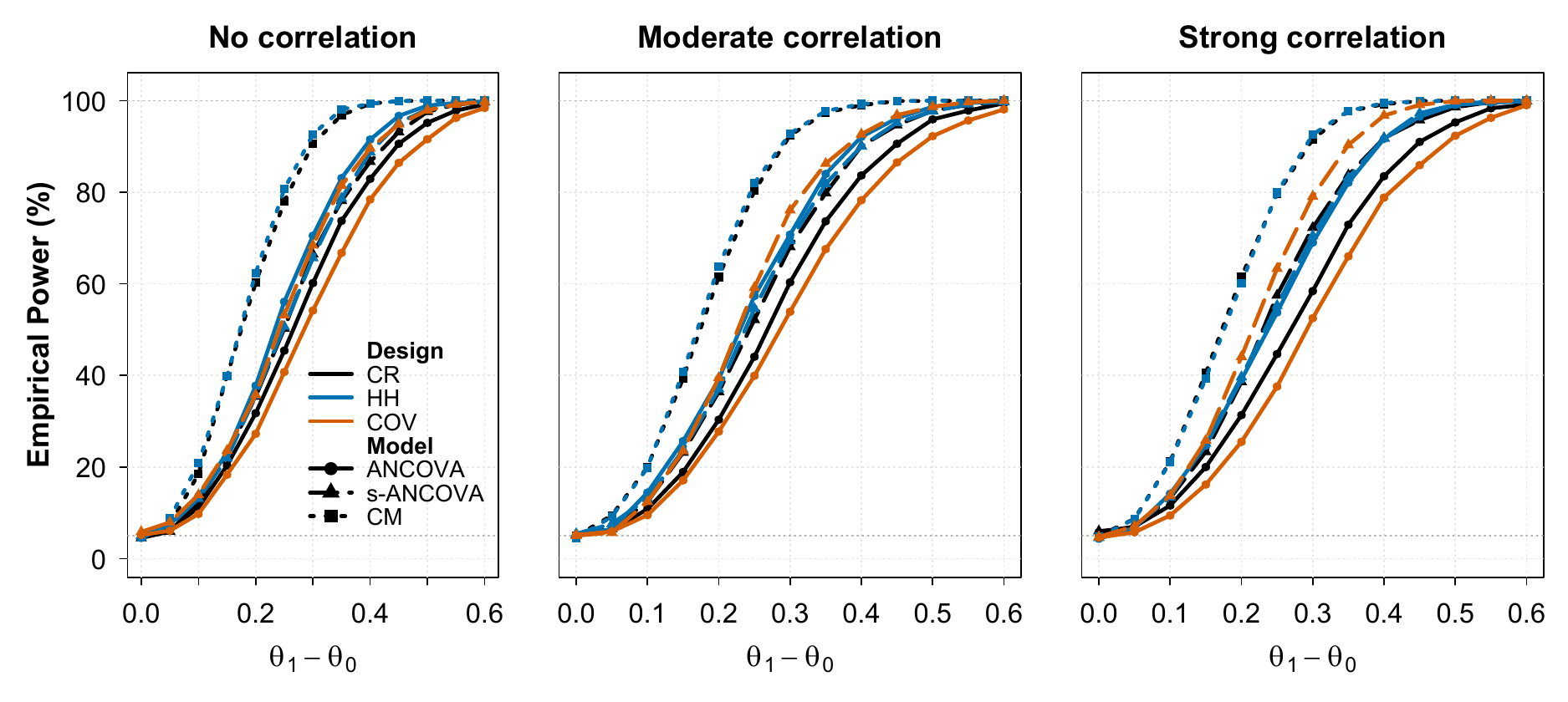}
        
    \end{subfigure}

    \caption{Bootstrap adjusted empirical power curves for \textbf{Case} 3 (Top) and \textbf{Case} 4 (Bottom). }
    \label{fig:case 45}
\end{figure}
Cases $1$--$4$ represent different forms of $g$. Case $3$ illustrates a scenario of possible subpopulation effect across different margins/stratum of $\vZ_i$ (see Section \ref{sec:impact discuss}). Cases $1,3,4$, especially $4$, depict scenarios under which CCAR is not stable, and DCAR may be more robust. The type I error rate of the unadjusted and bootstrap adjusted test (see Appendix \ref{app:bootstrap} for a detailed construction) under Cases $1$--$4$ are reported in Table \ref{tab:size main}. Further, we report the empirical power curves of Cases $3,4$ in Figure \ref{fig:case 45}. The perfect linear case is not studied in this simulation since it is studied extensively in the literature \cite{ma2015testing,ma2024new,li2019testing}. Note that in above simulations, all available covariates are fitted. In Appendix \ref{app:additional sim}, we also consider models with omitted variables, and additiona Cases $5,6$ studying approximately linear and non-additive non-linear effects.

Table \ref{tab:size main} shows that type I error is always controlled under CR, no matter the true/analysis model or covariates' correlation. DCAR methods are in general conservative due to the overestimation of limiting variance. In contrast, CCAR may be empirically inflated because of the variance inflation issue (see Cases $3,4$). The bootstrap adjustment restores the validity of conservative or inflated tests.

Figure \ref{fig:case 45} depicts the empirical power curves for Cases $3,4$. The result provides rich interpretations. For Case 4, subpopulation effects discussed in Subsection \ref{subsec: other functions} makes SFE and CM models outperform ANCOVA models under DCAR, even though the covariates are continuous and the relationship is smooth. In both cases, the better method is DCAR with CM model, which consistently outperforms CCAR methods but cannot distinguish from CR. This is because of the ARE comparison in Theorem \ref{theorem:compare tests}, there is no efficiency difference between CR and DCAR if all strata are fitted. We emphasize that such similarity is, in fact, too ideal for CR, because it requires the correct identification of all prognostic strata. In practice, omission of strata are usually common. Numerical studies on the impact of such omission can be located in Section \ref{sec: real data} and Appendix \ref{app:additional sim}.

\section{Real data example}\label{sec: real data}
Empirical evidence is provided based on a randomized clinical trial to compare sitagliptin with conventional treatment on antherosclerosis in Type $2$ diabetes patients \citep{oyama2016effect}. The data is publicly available in \citep{shao2026should}. One of the secondary endpoints is the 24 month high-molecular-weight adiponectin concentration ($\text{HMW}$). HMW helps muscle and liver tissues use glucose and burn fat and a lower HMW adiponectin concentration predicts metabolic syndrome. The covariates used in
randomization are the baseline HMW concentration, systolic blood pressure and body-weights. All three variables are continuous and prognostic for the response. We include these variables because the inclusion mimics the true data generating process. The data is analyzed via a spline model that includes the baseline covariates
\begin{equation}\label{eq:dgp model real data}
    \log(\text{HMW}_i) = \beta_0+(\theta_1-\theta_0)I_i+\beta_1\text{weight}^0_i+s(\text{log HMW}^0_i)+s(\text{SBP}^0_i)+\varepsilon_i
\end{equation}
where $\varepsilon_i \overset{d} = 0.5\xi$ and the superscript $0$ denotes baseline values. The estimated smooth effects are plotted in Figure \ref{fig:smooth effects}. All variable and smoothed effects are highly significant. The treatment effect of sitagliptin over conventional treatment is $\theta_1-\theta_0=-0.184$, the covariate effects are $(\beta_0,\beta_1) = (1.66,-0.01)$. Synthetic trials of size $n=400$ are created by resampling the covariates from the real data, assigning via different randomization methods and generating responses via \eqref{eq:dgp model real data}. In the following discussion, whenever strata are needed, we note that baseline HMW is discretized at $(2.37,4.44)\mu\text{g/mL}$, blood pressure is discretized at $135\text{mmHg}$ and weight is discretized at $64.3$kg, resulting in $3\times2\times 2=12$ strata.

\begin{figure}
    \centering
    \includegraphics[width=\linewidth]{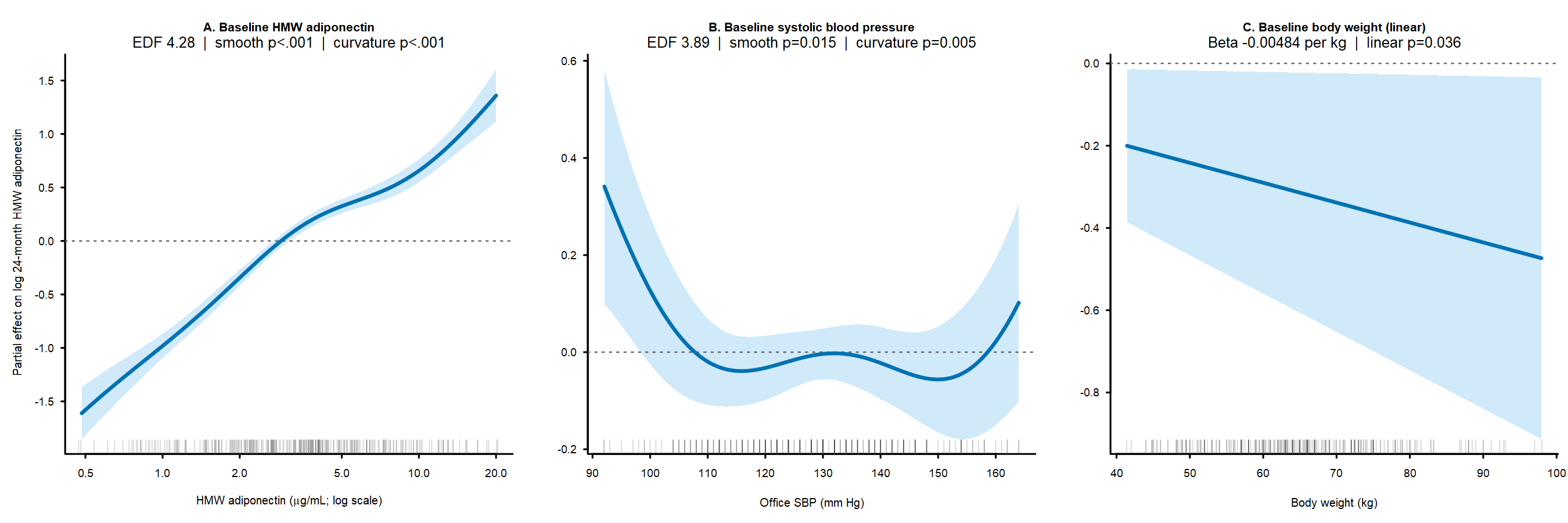}
    \caption{Estimated smooth effects of the continuous baseline covariates, HMW, blood pressure and weight, under model \eqref{eq:dgp model real data} for the Type-$2$ diabetes data. }
    \label{fig:smooth effects}
\end{figure}

We consider detecting the treatment effect with the proposed tests under
$a=\aZ,\aS,\aCM$. Table \ref{tab:realdata_trial6} summarizes the estimation and
adjusted test power. All design--model pairs estimate the treatment effect
without visible bias. If all strata are fitted, then CR and HH achieve similar effects, but as emphasized in Remark \ref{remark:CR}, the full adjustment under CR is unrealistic and could lead to too few observations in each strata. However, if important strata are omitted, the statistical power collapses under CR. The COV procedure, which balances the continuous covariates themselves, is the most efficient on average. At the analysis stage, replacing the ANCOVA by SFE results in a drop of statistical power. This is anticipated because among the two dominant prognostic effects, weight has completely linear effects while baseline HMW adiponectin is approximately linear. The gain of CM over ANCOVA is negligible because the within-stratum slopes nearly coincide with the global ones, so the quadratic form in \eqref{eq:gain-Z} is small.

We also observe interesting behaviors within each randomized method. For HH, SFE model does not lose efficiency from omitting strata ($87.2\%$ vs $86.3\%$). This is the full strata is readily balanced by HH. For COV, the inclusion of more strata leads to a decrease of power ($88.6\%$ to $94.7\%$). This is because the residual produced by the full SFE model is less aligned with the features balanced by COV directly, thereby introducing additional variability. The observation is somehow counter-intuitive compared with the consensus that fitting more variables will not make the model worse. We discuss this issue in details in  Appendix \ref{app:counter}.

\begin{table}[htbp]
\centering
\caption{Estimation and bootstrap-adjusted test power for the treatment effect, reps$=1200$, $B=500$. If there is omission, only $3$ out of $11$ strata dummies (blood pressure$\times$weight effects), otherwise all variables are fitted. \textbf{Boldface}: best model under a fixed design.}
{\fontsize{10pt}{10.5pt}\selectfont
\renewcommand{\arraystretch}{0.92}
\setlength{\tabcolsep}{1.8pt}
\begin{tabular}{llcccccccc}
\toprule
\multirow{2}{*}{Design}
& All covariates
& & \multicolumn{3}{c}{Estimation}
& & \multicolumn{3}{c}{Adjusted test power (\%)} \\
\cline{4-6}\cline{8-10}
& fitted & & ANCOVA & SFE & CM
& & ANCOVA & SFE & CM \\
\midrule
\multirow{2}{*}{CR}
& Yes
& & -0.187 & -0.186 & -0.187
& & 97.3& 87.2 & \textbf{97.5} \\
& No
& &  -0.187 &  -0.187 &  -0.187
& & 66.9 & 65.8 & \textbf{67.3} \\
\addlinespace
\multirow{2}{*}{HH}
& Yes
& &  -0.184 & -0.184 & -0.184
& & 96.3 & 87.2 & \textbf{96.5} \\
& No
& & -0.184 & -0.184 & -0.183
& & 86.2 & 86.3 & \textbf{86.9} \\
\addlinespace
\multirow{2}{*}{COV}
& Yes
& & -0.186 & -0.185 & --
& & \textbf{96.7} & 88.6 & -- \\
& No
& & -0.186 & -0.186 & --
& & \textbf{96.3} & \underline{94.7} & -- \\
\bottomrule
\end{tabular}}
\label{tab:realdata_trial6}
\end{table}

\section{Conclusion}\label{sec:conclusion}

This paper provides a theoretical investigation of discretization under CAR and its implications for subsequent statistical inference. Our main findings can be summarized as follows. If the response-covariate relationship is known, then strongly balancing the relationship in the design yields the highest efficiency. When $g$ is unknown and/or potentially non-linear, discretization in the design is preferable due to its robustness. A combined model \eqref{eq:combined working model} should be used with DCAR in analysis. For partially known structures, mixed designs that balance the known components and discretize the unknown ones represent a practical choice (Appendix \ref{app:discussions on MixCAR}). In any case, bootstrap adjustment is recommended to address the conservativeness of DCAR and the possible type I error issues under CCAR or MixCAR. The recommendations are given in Table \ref{tab:summary}.

\begin{table}[t!]
\centering
\caption{A summary of recommendations. \checkmark: the preferable design and the preferable model provided that the preferable design has already been used. \dag: CR is not recommended because of worse covariate balance and the requirement of full post-stratification, which is unrealistic in many cases (See SubSection \ref{subsec:hypothesis test}, Remark \ref{remark:CR}).}
{%
    \fontsize{9pt}{11pt}\selectfont
\begin{tabular}{lcccc:ccc|c}
\toprule
$g$ \textbf{Function} & $g_j$& CR$^\dag$ & DCAR & CCAR  & ANCOVA & SFE& CM&\textbf{Where to discretize}\\
\midrule
Known, additive& Not stepwise &&&\checkmark&\checkmark &&& Do not discretize\\
Known, additive&Stepwise&&\checkmark&&&\checkmark&&In design and analysis\\
Unknown&-&&\checkmark&&&&\checkmark&In design and analysis\\
\bottomrule
\end{tabular}}
\label{tab:summary}
\end{table}

Discretization is conventionally related to loss of information \citep{Cox1957NoteOG,royston2006dichotomize}. We show in this paper that discretization is not only loss of information. It is a form of `regularization' in randomized experiments. It determines at what resolution should the covariates enter treatment balancing. The most aggressive type of discretization has only $2$ units in each strata \citep{bai2022optimality}. Thus, many practical randomized experiment rely on discretization, but lacks theoretical justification. From a methodological standpoint, the framework developed in this paper provides a perspective for analyzing a wide class of designs. 

Several directions remain open for future research. First, this work focused on equal allocation; extending the framework to unequal allocation ratios would increase its applicability. Second, our asymptotic results assume fixed covariate dimension and fixed strata counts, but in practice these may grow with sample size. Understanding non-asymptotic performance in high-dimensional or diverging-strata regimes is an important next step. Third, as discussed in Section \ref{sec:model robust}, the extension of model-robust inference after CAR to incorporate continuous variables merit further study. Finally, extending the framework to other outcome types, e.g., binary/count with generalized linear models \citep{zhai2026validtestmultiarmtrials} or survival data could broaden its practical relevance. 

\section*{Gen AI declaration}
The original proofs and code were written by the authors, generative AI (ChatGPT 5.6-soL and Claude Fable 5) were used for checking the proofs, revising codes and language polishing.

\newgeometry{left=1in,right=1in,top=0.75in,bottom=0.75in}
\bibliographystyle{plain}
\bibliography{r1}
\restoregeometry

\clearpage

\appendix
\addtocontents{toc}{\protect\setcounter{tocdepth}{2}}

\spacingset{1}

\begin{center}
{\LARGE\bf Supplementary materials for\\[3pt]
Discretization in covariate-adaptive randomization: gains and losses}
\end{center}
\medskip

This Appendix provides supplementary materials and technical proofs for the paper ``Discretization in covariate-adaptive randomization: gains and losses''. In Appendix \ref{app:technical}, \ref{app:bootstrap} we present the technical details for the main results and the bootstrap adjustment. In Appendix \ref{app:discussions on MixCAR}, \ref{app:discuss on dgp}, we discuss some additional topics not addressed in the main paper. Appendix \ref{sec:additional covariates} considers the extension of the topic to include additional covariates. In Appendix \ref{app:additional sim}, we include some additional simulation results on the empirical power of the unadjusted Wald-test under difference cases studied in Section \ref{sec:numerical} in the main paper. Finally, Appendix \ref{app:monte carlo} justifies the consistency of the Monte-Carlo variance estimator used in Example \ref{example:instability}.

\tableofcontents

\section{Technical proofs}\label{app:technical}

\subsection{Previous results}\label{app:notation}
In this section, we present the definitions of $\pi_\Lambda$ in Theorems \ref{theorem: design property}, \ref{theorem:combined}. First, we recall some theoretical results in \cite{ma2024new} that are needed for the proof for CCAR. Recall 
$$\Lambda_n^{(r)} = \sum_{i=1}^{n}(2I_i^{(r)}-1)\psi(\vZ_i),$$
where $\psi:\mathbb{R}^p\to\mathbb{R}^{p'}$ is the pre-specified covariate map and we omit the dependence on the design type $r$ in the following proofs if the context is clear. Let $\mathcal{S}_{\Lambda}\subseteq \mathbb{R}^{p'}$ be a general state space. For any function $w:\mathcal{S}_{\Lambda}\to \mathbb{R}$, symmetric invariant distribution $\pi_{\Lambda}$ (see \cite{ma2024new} for the proof of symmetricity under CAR) on $\mathcal{S}_{\Lambda}$ and transition probability kernel ${\PL}(x,{\cB})$ with $\cB$ a Borel subset of $\mathcal{S}_{\Lambda}$, denote by $\pi_{\Lambda}(w)=\int_{x\in \mathcal{S}_{\Lambda}}w(x)\pi_{\Lambda}(dx), \PL(w(\lambda)) = \int w(y)\PL(\lambda,dy)$ whenever the integrations are well-defined. The chain is a Markov chain on $\bR^{p'}$ and $\Lambda_n^{(r)}=o_P(\sqrt{n})$ \cite[Theorem 3.1]{ma2024new}. Throughout the proofs, we use $a\wedge b$ (resp. $a \vee b$) to denote $\min\{a,b\}$ (resp. $\max\{a,b\}$).

Let $\cF_0$ be the trivial sigma-algebra and for $i\geq 1$, $\cF_i=\sigma(\cF_{i-1},I_i,\vZ_i)$ the sigma-algebra generated by all the information upon the arrival of the $i$th subject. Note that $\cF_i$ may depend on the design type, we omit the extra notation if the context is made clear. For a function $h:\bR^p\to \bR$, define
$$\kappa_h(\Lambda_{i-1}) := E\left\{(2{I_i}-1)[h(\vZ_i)+\varepsilon_i]\mid \mathcal{F}_{i-1}\right\} = (1-2{\rho})E[h(\vZ){\operatorname{sgn}}({\langle \Lambda_{i-1},\psi(\vZ)\rangle})],$$
where $\langle\cdot,\cdot\rangle$ is the inner product on $\mathbb{R}^{p'}$, ${\operatorname{sgn}}(a) = \mathbbm{1}(a>0)-\mathbbm{1}(a<0)$ and the second equality follows directly from \cite[Proof of Theorem 3.6]{ma2024new} and the independence of $\varepsilon_i$ from all other variables.

\begin{lemma}\label{lemma:previous result1}
    Consider the general CAR framework in Section \ref{subsec:CAR and discretization}, for any type of design, assume $h:\bR^p\to \bR$ is some square-integrable function. Then, the following holds

    \noindent (i) $|\kappa_h(\Lambda)|\leq E(|h(\vZ)|)<\infty$ and $\kappa_h(-\Lambda) =-\kappa_h(\Lambda)$.

    \noindent (ii) One has $\pi_{\Lambda}(\kappa_h)=0$, and there is a unique function $\hat{\kappa}_h:\mathcal{S}_{\Lambda}\to\mathbb{R}$ with $\hat{\kappa}_h({\vZero})=0$ solving the Poisson equation
    $${\hat{\kappa}_h-\PL\hat{\kappa}_h=\kappa_h-\pi_{\Lambda}(\kappa_h)}.$$

    \noindent (iii) The solution $\hat{\kappa}_h$ satisfies the bound $|\hat{\kappa}_h(\Lambda)|\leq C(\|\Lambda\|^2+1)$ for some constant $C>0$. Also, $\sup_n E|\hat{\kappa}_h(\Lambda_n)|^{{k}}<\infty$ for each $k$$>0$.

    \noindent (iv) For each $i\geq 1$, define $\Delta M_{h,i,1}=(2I_i-1)[h(\vZ_i)+\varepsilon_i]-E\left\{(2{I_i}-1)[h(\vZ_i)+\varepsilon_i]\mid \mathcal{F}_{i-1}\right\}$, $\Delta M_{h,i,2} = \hat{\kappa}_h(\Lambda_i)-{\PL}\hat{\kappa}_h(\Lambda_{i-1})$. Then, the sequence
    $$\Delta M_{h,i} = \Delta M_{h,i,1}+\Delta M_{h,i,2} = (2I_i-1)[h(\vZ_i)+\varepsilon_i]+\hat{\kappa}_h(\Lambda_i)-\hat{\kappa}_h(\Lambda_{i-1})$$
    is a martingale difference sequence with respect to the filtration $\{\mathcal{F}_i:i\geq 1\}$.
\end{lemma}
\begin{proof}
    Claim (i) follows directly by definition. Claim (ii), (iii) utilizes Poisson equation arguments from \cite[Theorem 14.3.7, 17.4.2]{meyn2012markov} which follows from \cite[proof of Theorem 3.6]{ma2024new}. Claim (iv) follows from the Poisson equation and by noting $\varepsilon_i$ is independent from other variables with mean $0$.
\end{proof}

To explicitly define ${\pi_{\Lambda}(\zeta_{f})},{\pi_{\Lambda}(\zeta_{g})},{\pi_{\Lambda}(\zeta_{g_{d_S}})}$ in Theorems \ref{theorem: design property} and \ref{theorem:combined}, we require the following notation. Denote by ${g_{d_S}}(\vZ_i) = g(\vZ_i) - \sum_{j=1}^{{d_S}}\eta_j\stind{j}(\vZ_i)$. For any $\lambda\in \mathcal{S}_{\Lambda}$ and some function $h:\bR^p\to\bR$, define
$$\zeta_{h}(\lambda) = E[h(\vZ)\hat{\kappa}_h(\lambda+\psi(\vZ))\{1+(1-2\rho){\operatorname{sgn}}(\langle\lambda,\psi(\vZ)\rangle)\}].$$
Hence, in Theorem \ref{theorem: design property}, \begin{equation}
    {\pi_{\Lambda}(\zeta_f)} = \int_{\lambda\in \mathcal{S}_{\Lambda}} \zeta_f(\lambda)d\pi_{\Lambda}.\label{eq:pi_Lambda f}
\end{equation} 
In Theorem \ref{theorem:combined},
\begin{align}
    &{\pi_{\Lambda}(\zeta_g)} = \int_{\lambda\in \mathcal{S}_{\Lambda}} \zeta_g(\lambda)d\pi_{\Lambda},\label{eq:pi_Lambda g}\\
    &{\pi_{\Lambda}(\zeta_{g_{d_S}})} = \int_{\lambda\in \mathcal{S}_{\Lambda}} \zeta_{{g_{d_S}}}(\lambda)d\pi_{\Lambda}.\label{eq:pi_Lambda gstar}
\end{align}

\subsection{Proof of imbalance properties}\label{app:design prop}
In this Appendix, Theorem \ref{theorem: design property} and Proposition \ref{prop:scalar-ccar} are proved in order.

\begin{proof}[Proof of Theorem \ref{theorem: design property}]
    The results for CR follows directly by applying the CLT. Under DCAR for any function $f$ that is not balanced, note the following decomposition
    \begin{align*}
        \frac{1}{\sqrt{n}}\sum_{i=1}^{n}(2I_i-1)f(\vZ_i)& = \frac{1}{\sqrt{n}}\sum_{i=1}^{n}(2I_i-1)\left\{f(\vZ_i)-E[f(\vZ_i)\mid S_i]\right\} + \frac{1}{\sqrt{n}}\sum_{i=1}^{n}(2I_i-1)E[f(\vZ_i)\mid S_i]\\&=
        \frac{1}{\sqrt{n}}\sum_{i=1}^{n}(2I_i-1)\left\{f(\vZ_i)-E[f(\vZ_i)\mid S_i]\right\}\\&\hspace{7em}+\frac{1}{\sqrt{n}}\sum_{s\in[\prod_j m_j]}D_n({\stind{s}})E[f(\vZ)\mid {S}=s],
    \end{align*}
    where $\stind{s}(\vZ_i)=\mathbbm{1}\{S_i=s\}$ is the stratum indicator. The second term is $o_P(1)$ due to Condition \ref{assume: dgp}. Moreover, conditioning on $\{S_1,\ldots,S_i\}$, the difference $f(\vZ_i)-E[f(\vZ_i)\mid S_i]$ is mean $0$ and independent of $I_i$. Therefore, by CLT
    $$\frac{1}{\sqrt{n}}\sum_{i=1}^{n}(2I_i-1)f(\vZ_i)\rightsquigarrow \left(E\{\var[f(\vZ)\mid {S}]\}\right)^{1/2}\xi.$$
    Under CCAR, for $f$ that is not strongly balanced, the proof is established by directly applying \cite[Theorem 3.6]{ma2024new}. The case for directly blanaced $f$ is trivial by Condition \ref{assume: dgp}. The proof is finished.
\end{proof}

Recall the set up of CAR with scalar feature map discussed in Section \ref{sec:instability}. We now explain and quantify the variance inflation phenomenon theoretically. For
$f:\mathbb{R}^p\to\mathbb{R}$ with $E[f^2(\vZ)]<\infty$, set
$$b_f:=\frac{E[f(\vZ)\operatorname{sgn}\psi(\vZ)]}{E|\psi(\vZ)|},
\quad
b_f^{\ast}:=\frac{E[f(\vZ)\psi(\vZ)]}{E[\psi^2(\vZ)]},$$
which are well defined because $P\{\psi(\vZ)=0\}=0$, and $E[\psi^2(\vZ)]<\infty$ by Assumption~\ref{assume:bounded moments}. Note that $E\{f(\vZ)[\operatorname{sgn}\psi(\vZ)]\}=b_fE|\psi(\vZ)|$ by the definition of $b_f$. Since $\psi(z)$ is a scalar, we have for every $\lambda\in\mathbb{R}$ and $\vz\in\mathbb{R}^p$,
\begin{equation}\label{eq:sign-id}
\operatorname{sgn}\{\lambda\psi(\vz)\}=\operatorname{sgn}(\lambda)\operatorname{sgn}\psi(\vz),\qquad \psi(\vz)\operatorname{sgn}\{\lambda\psi(\vz)\}=\operatorname{sgn}(\lambda)|\psi(\vz)|,
\end{equation}
with both sides equaling $0$ whenever $\lambda=0$ or $\psi(\vz)=0$. Recall in Appendix~\ref{app:notation}, $E[(2I_i-1)\mid\mathcal{F}_{i-1},\vZ_i]=(1-2\rho)\operatorname{sgn}\{\Lambda_{i-1}\psi(\vZ_i)\}$. Since $\vZ_i$ is independent of $\mathcal{F}_{i-1}$, Equations \eqref{eq:sign-id} yield, for every $\lambda\in\cS_{\Lambda}$,
\begin{align}
E[\Lambda_i\mid\Lambda_{i-1}=\lambda]&=\lambda+(1-2\rho)\operatorname{sgn}(\lambda)E|\psi(\vZ)|,\label{eq:onestep1}\\
E[\Lambda_i^{2}\mid\Lambda_{i-1}=\lambda]&=\lambda^{2}-2(2\rho-1)E|\psi(\vZ)||\lambda|+E[\psi^2(\vZ)].\label{eq:onestep2}
\end{align}
The following lemma derives the expectation of $|\Lambda|$, which is useful for establishing Proposition \ref{prop:scalar-ccar}.
\begin{lemma}\label{lem:piL-abs}
Under the conditions of Proposition~\ref{prop:scalar-ccar}, the distribution of $\Lambda$ satisfies
\[
\piL(|\Lambda|)=\frac{E[\psi^2(\vZ)]}{2(2\rho-1)E|\psi(\vZ)|}<\infty .
\]
\end{lemma}
\begin{proof}
Let $\Lambda_0$ follow $\piL$ and $\Lambda_1=\Lambda_0+(2I_1-1)\psi(\vZ_1)$ be one stationary transition, so that $\Lambda_1$ follows $\piL$ as well. Define $$c(\lambda)=2(2\rho-1)E|\psi(\vZ)||\lambda|-E[\psi^2(\vZ)],$$ so that \eqref{eq:onestep2} implies $E[\Lambda_1^{2}\mid\Lambda_0=\lambda]=\lambda^{2}-c(\lambda)$. By definition, $\lambda^{2}-c(\lambda)\ge 0$ and $c(\lambda)\ge -E[\psi^2(\vZ)]$ for every $\lambda$.

The first step is to show boundedness. Fix $M>0$ and let $\varphi_M(x)=x\wedge M$ for $x\ge 0$. For any nonnegative random variable $X$ and any $\sigma$-field $\mathcal{G}$,
\[
E[\varphi_M(X)\mid\mathcal{G}]\le\min\bigl\{E[X\mid\mathcal{G}],M\bigr\}=\varphi_M\bigl(E[X\mid\mathcal{G}]\bigr).
\]
Applying this with $X=\Lambda_1^{2}$ and $\mathcal{G}=\sigma(\Lambda_0)$, taking expectations, and using the stationarity of $\pi$, we obtain
$$
0\le E\bigl[\varphi_M\bigl(\Lambda_0^{2}-c(\Lambda_0)\bigr)-\varphi_M(\Lambda_0^{2})\bigr].
$$
We bound the integrand pointwise, writing $a=\Lambda_0^{2}$ and $c=c(\Lambda_0)$. When $c\ge 0$ and $a\le M$, both $a-c\le a\le M$, so the difference in above display equals $-c$. When $c\ge 0$ and $a>M$, monotonicity of $\varphi_M$ makes the difference nonpositive. When $c<0$, the $1$-Lipschitz property of $\varphi_M$ bounds the difference by $|c|\le E[\psi^2(\vZ)]$. Combining the three cases,
$$
E\bigl[c^{+}(\Lambda_0)\mathbbm{1}\{\Lambda_0^{2}\le M\}\bigr]\le -c^+ \mathbbm{1}\{\Lambda_0^{2}\le M\} + E[\psi^2(\vZ)]\mathbbm{1}\{c<0\} \le E[\psi^2(\vZ)],
$$
with $c^+= c\vee 0$. Monotone convergence theorem as $M\to\infty$ gives $E[c^{+}(\Lambda_0)]\le E[\psi^2(\vZ)]$. Since $2(2\rho-1)E|\psi(\vZ)||\lambda|=c(\lambda)+E[\psi^2(\vZ)]\le c^{+}(\lambda)+E[\psi^2(\vZ)]$ for each $\lambda$, taking expectations on both sides with respect to $\piL$ yields $$\piL(|\Lambda|)\le \frac{E[\psi^2(\vZ)]}{\{(2\rho-1)E|\psi(\vZ)|\}}<\infty.$$

Now we calculate the exact value of $\piL(|\Lambda|)$. For fixed $M>0$, stationarity gives
\begin{equation}\label{eq:trunc}
E\bigl[\varphi_M(\Lambda_1^{2})-\varphi_M(\Lambda_0^{2})\bigr]=0.
\end{equation}
Moreover, the $1$-Lipschitz property and $\Lambda_1=\Lambda_0+(2I_1-1)\psi(\vZ_1)$ gives
$$
\bigl|\varphi_M(\Lambda_1^{2})-\varphi_M(\Lambda_0^{2})\bigr|\le|\Lambda_1^{2}-\Lambda_0^{2}|\le 2|\Lambda_0||\psi(\vZ_1)|+\psi^{2}(\vZ_1),
$$
whose expectation is at most $2\piL(|\Lambda|)E|\psi(\vZ)|+E[\psi^2(\vZ)]<\infty$ by above derivations. Since $\varphi_M(\Lambda_1^{2})-\varphi_M(\Lambda_0^{2})\to\Lambda_1^{2}-\Lambda_0^{2}$ as $M\to\infty$, dominated convergence applied to \eqref{eq:trunc} gives $E(\Lambda_1^{2}-\Lambda_0^{2})=0$. On the other hand, $E(\Lambda_1^{2}-\Lambda_0^{2}\mid \Lambda_0)=-c(\Lambda_0)$, so applying the law of total expectation, we get
$$
0=E(\Lambda_1^{2}-\Lambda_0^{2})=E[-c(\Lambda_0)]=-2(2\rho-1)E|\psi(\vZ)|\piL(|\Lambda|)+E[\psi^2(\vZ)],
$$
where $E[c(\Lambda_0)]$ is finite by the boundedness step. Solving for $\piL(|\Lambda|)$ proves the lemma.
\end{proof}

The next proposition contains some detailed, intermediate results that reveals the mechanism under variance inflation, which provides intuition on why this inflation may happen. 

\begin{prop}\label{prop:scalar-ccar-detail}
Under the conditions of Proposition~\ref{prop:scalar-ccar}, the following hold.

\noindent(i) $n^{-1/2}D_n^{(\tCC)}(f)\rightsquigarrow\sigma_f^{(\tCC)}\xi$ with
\begin{equation}\label{eq:scalar-var}
\bigl(\sigma_f^{(\tCC)}\bigr)^2= E\bigl[\{f(\vZ)-b_f\psi(\vZ)\}^2\bigr].
\end{equation}

\noindent(ii) The solution of the Poisson equation in Lemma~\ref{lemma:previous result1}(ii) is has a closed form $\widehat{\kappa}_f(\lambda)=-b_f\lambda$, for all $\lambda\in\cS_{\Lambda}$. Consequently, $2{\piL(\zeta_f)}=-2b_fE[f(\vZ)\psi(\vZ)]+b_f^2E[\psi^2(\vZ)]$.

\noindent(iii) Furthermore, for scalar feature maps, $f$ is strongly balanced if and only if $f\in\operatorname{span}\{\psi\}$ and 
\begin{equation}\label{eq:gain-loss}
\bigl(\sigma_f^{(\tCC)}\bigr)^2=\underbrace{\bigl(\sigma_f^{(\tCR)}\bigr)^2-\frac{\{E[f(\vZ)\psi(\vZ)]\}^{2}}{E[\psi^{2}(\vZ)]}}_{\text{best possible after projecting out }\operatorname{span}\{\psi\}}+\underbrace{(b_f-b_f^{\ast})^{2}E[\psi^{2}(\vZ)]}_{\text{misprojection penalty}}.
\end{equation}
\end{prop}

\begin{proof}[Proof of Propositions~\ref{prop:scalar-ccar} and \ref{prop:scalar-ccar-detail}]

\noindent\textbf{Part (i).} Set $h:=f-b_f\psi$, so that $E[h^2(\vZ)]<\infty$ and by the definition of $b_f$,
\[
E[h(\vZ)\operatorname{sgn}\psi(\vZ)]=E[f(\vZ)\operatorname{sgn}\psi(\vZ)]-b_fE[\psi(\vZ)\operatorname{sgn}\psi(\vZ)]=0.
\]
For $i\ge 1$ let $\Delta\cM_i:=(2I_i-1)h(\vZ_i)$ and $\cM_n:=\sum_{i=1}^{n}\Delta\cM_i$. Conditioning on $(\mathcal{F}_{i-1},\vZ_i)$ first and then using the independence of $\vZ_i$ from $\mathcal{F}_{i-1}$, the function $\kappa_h$ of Appendix~\ref{app:notation} gives $E[\Delta\cM_i\mid\mathcal{F}_{i-1}]=\kappa_h(\Lambda_{i-1})$, and \eqref{eq:sign-id} gives $\kappa_h(\lambda)=(1-2\rho)\operatorname{sgn}(\lambda)E[h(\vZ)\operatorname{sgn}\psi(\vZ)]=0$ for every $\lambda$. Hence $(\cM_i)_{i\ge1}$ is a martingale with respect to $\mathcal{F}_i$. Further, since $(2I_i-1)^2=1$ and $\vZ_i$ is independent of $\mathcal{F}_{i-1}$, we get $E[\Delta\cM_i^{2}\mid\mathcal{F}_{i-1}]$$=E[h^{2}(\vZ_i)\mid\mathcal{F}_{i-1}]=E[h^{2}(\vZ)]$, a constant. Thus, $n^{-1}\sum_{i=1}^{n}E[\Delta\cM_i^{2}\mid\mathcal{F}_{i-1}]$$=E[h^{2}(\vZ)]$ for every $n$. For the conditional Lindeberg condition, for every $\epsilon>0$,
\[
\frac{1}{n}\sum_{i=1}^{n}E\bigl[{\Delta\cM_i^{2}}\mathbbm{1}\{|{\Delta\cM_i}|>\epsilon\sqrt{n}\}\mid\mathcal{F}_{i-1}\bigr]=E\bigl[h^{2}(\vZ)\mathbbm{1}\{|h(\vZ)|>\epsilon\sqrt{n}\}\bigr]\to 0
\]
as $n\to\infty$, by dominated convergence. The martingale central limit theorem \cite[Corollary~3.1]{hall2014martingale}, applied to the array $n^{-1/2}\Delta\cM_i$ with filtration $\mathcal{F}_i$, gives $n^{-1/2}\cM_n$$\rightsquigarrow \mathcal{N}\bigl(0,E[h^{2}(\vZ)]\bigr)$. By the definition of $h$,
$$
D_n^{(\tCC)}(f)=\sum_{i=1}^{n}(2I_i-1)\bigl\{h(\vZ_i)+b_f\psi(\vZ_i)\bigr\}={\cM_n}+b_f\Lambda_n .
$$
Since $\psi$ is strongly balanced by the design, we have $n^{-1/2}\Lambda_n=o_P(1)$. Slutsky's theorem now yields part (i).

\noindent\textbf{Part (ii).} Recall $\kappa_f$ from Appendix~\ref{app:notation} and by \eqref{eq:sign-id},
\[
\kappa_f(\lambda)=(1-2\rho)E\bigl[f(\vZ)\operatorname{sgn}\{\lambda\psi(\vZ)\}\bigr]=(1-2\rho)b_fE|\psi(\vZ)|\operatorname{sgn}(\lambda),
\]
which is an odd step function. This shows that $\piL(\kappa_f)=0$ by the symmetry of $\piL$. Consider the candidate function $\widetilde{\kappa}(\lambda)=-b_f\lambda$. The mean $\PL\widetilde{\kappa}(\lambda)=-b_fE[\Lambda_1\mid\Lambda_0=\lambda]$ is finite because $E[|\Lambda_1|\mid\Lambda_0=\lambda]\le|\lambda|+E|\psi(\vZ)|$, and \eqref{eq:onestep1} gives
\[
\widetilde{\kappa}(\lambda)-\PL\widetilde{\kappa}(\lambda)=-b_f\lambda+b_f\bigl\{\lambda+(1-2\rho)\operatorname{sgn}(\lambda)E|\psi(\vZ)|\bigr\}=(1-2\rho)b_fE|\psi(\vZ)|\operatorname{sgn}(\lambda)=\kappa_f(\lambda)-\piL(\kappa_f)
\]
for every $\lambda\in\cS_{\Lambda}$ while $\widetilde{\kappa}(0)=0$. By the uniqueness in Lemma~\ref{lemma:previous result1}(ii), the solution to the Poisson equation is essentially $\widehat{\kappa}_f=\widetilde{\kappa}$. 

Substituting $\widehat{\kappa}_f(\lambda+\psi(z))=-b_f\{\lambda+\psi(z)\}$ into the definition of $\zeta_f$ in Appendix~\ref{app:notation} and using \eqref{eq:sign-id},
\begin{align*}
\zeta_f(\lambda)&=-b_fE\Bigl[f(\vZ)\{\lambda+\psi(\vZ)\}\bigl\{1+(1-2\rho)\operatorname{sgn}(\lambda)\operatorname{sgn}\psi(\vZ)\bigr\}\Bigr]\\
&=-b_f\Bigl[\lambda E[f(\vZ)]+E[f(\vZ)\psi(\vZ)]+(1-2\rho)\operatorname{sgn}(\lambda)\bigl\{\lambda b_fE|\psi(\vZ)|+E[f(\vZ)|\psi(\vZ)|]\bigr\}\Bigr].
\end{align*}
We integrate against $\piL$ term by term. The law $\piL$ is symmetric and $\piL(|\Lambda|)<\infty$ by Lemma~\ref{lem:piL-abs}, so $\piL(\lambda)=0$; because $\operatorname{sgn}$ is odd and bounded with $\operatorname{sgn}(0)=0$, also $\piL\{\operatorname{sgn}(\lambda)\}=0$; and finally $\piL\{\lambda\operatorname{sgn}(\lambda)\}=\piL(|\Lambda|)$. Thus
\[
\piL(\zeta_f)=-b_fE[f(\vZ)\psi(\vZ)]-b_f(1-2\rho)b_fE|\psi(\vZ)|\piL(|\Lambda|)=-b_fE[f(\vZ)\psi(\vZ)]+\tfrac{1}{2}b_f^{2}E[\psi^2(\vZ)],
\]
where the last equality uses Lemma~\ref{lem:piL-abs}, Part (ii) is proved.

\noindent\textbf{Part (iii).} Because $b_f^{\ast}$ minimizes $b\to E[\{f(\vZ)-b\psi(\vZ)\}^2]$ and the minimizer satisfies the normal equation $E[\{f-b_f^{\ast}\psi\}\psi]=0$, we have
\[
E[\{f-b_f\psi\}^{2}]=E[\{f-b_f^{\ast}\psi\}^{2}]+(b_f-b_f^{\ast})^{2}E[\psi^2(\vZ)]=E[f^{2}(\vZ)]-\frac{\{E[f(\vZ)\psi(\vZ)]\}^{2}}{E[\psi^2(\vZ)]}+(b_f-b_f^{\ast})^{2}E[\psi^2(\vZ)],
\]
which is \eqref{eq:gain-loss}. For the equivalence, suppose first that $f=c\psi$ almost surely for some constant $c$. Then $b_f=cE[\psi\operatorname{sgn}\psi]/E|\psi(\vZ)|=c$, so $h\equiv 0$ and $D_n^{(\tCC)}(f)=c\Lambda_n=O_P(1)=o_P(\sqrt{n})$ by the tightness established in Part (i); thus $f$ is strongly balanced. Conversely, if $f$ is strongly balanced, then $n^{-1/2}D_n^{(\tCC)}(f)=o_P(1)$, which together with Part (i) forces $E[\{f-b_f\psi\}^{2}]=0$, that is, $f=b_f\psi$ almost surely, so $f\in\operatorname{span}\{\psi\}$. The proof is finished.
\end{proof}

Proposition \ref{prop:scalar-ccar-detail} reveals important intuitions. It solves the Poisson equation under a specific case and quantifies the loss of using CCAR compared with CR. The variance inflation comes from the difference between $b_f$ and $b_f^*$, which can be viewed as projections of $f(\vZ)$ onto different spaces. This is, in fact, connected to the based-coin probability in \textbf{Step 3.} of the general CAR procedure in Subsection \ref{subsec:CAR and discretization}. The biased-coin assignment decides to assign patients via the sign of $\langle \Lambda, \psi\rangle$, while the update of imbalance is $\pm \psi$. Thus, the magnitude of $\psi$ is not fully absorbed into the adaptive, biased-coin decision. In fact, relevant discussion has already shown that the utilization of continuous allocation functions may mitigate this issue \citep[Equation (2.3)]{lixin2026covariateadaptiverandomizationclinicaltrials}. 

Finally, we prove Proposition \ref{prop:scalar-ccar}.
\begin{proof}[Proof of Proposition \ref{prop:scalar-ccar}]
Write $\tilde{f}:=f\operatorname{sgn}\psi$, so that $E[\tilde{f}^{2}]=E[f^{2}(\vZ)]$ because $P\{\psi(\vZ)=0\}=0$, $b_f=E(\tilde{f})/E|\psi(\vZ)|$, and $\tilde{f}-b_f|\psi|=\operatorname{sgn}(\psi)\{f-b_f\psi\}$. Because multiplication by $\operatorname{sgn}\psi$ is an isometric involution on $L^{2}$ up to null sets, maximizing the inflation ratio is essentially maximizing
\[
{R(\tilde{f}):=\frac{E\bigl[\{\tilde{f}-(E[\tilde{f}]/E|\psi(\vZ)|)|\psi|\}^{2}\bigr]}{E[\tilde{f}^{2}]}}.
\]
Set ${r_\psi}:=|\psi|/E|\psi(\vZ)|-1$ and $\nu_\psi^{2}:=E[r_\psi^{2}]=\var(|\psi(\vZ)|)/\{E|\psi(\vZ)|\}^{2}$, and note that $E[{r_\psi}]=0$. If $E[\tilde{f}]=0$ then $R(\tilde{f})=1$. If $E[\tilde{f}]\neq 0$, by homogeneity of $R$ we may normalize $E[\tilde{f}]=1$ and write $\tilde{f}=1+w$ with $E[w]=0$; then $\tilde{f}-(E[\tilde{f}]/E|\psi(\vZ)|)|\psi|=w-r_\psi$ and, with $t:=(E[w^{2}])^{1/2}$,
\[
{R(\tilde{f})=\frac{E[w^{2}]-2E[wr_\psi]+\nu_\psi^{2}}{1+E[w^{2}]}\le\frac{(t+\nu_\psi)^{2}}{1+t^{2}}},
\]
by the Cauchy--Schwarz inequality $E[wr_\psi]\ge -t\nu_\psi$, with equality exactly when $w$ is a nonpositive multiple of ${r_\psi}$. If $\nu_\psi=0$, the right-hand side equals $t^{2}/(1+t^{2})<1$, so $R\le 1$ for every $f$, with equality whenever $E[\tilde{f}]=0$; this proves the `only if' part of the final claim. If $\nu_\psi>0$, differentiating in $t\ge 0$ gives
\[
{\frac{d}{dt}\frac{(t+\nu_\psi)^{2}}{1+t^{2}}=\frac{2(t+\nu_\psi)(1-t\nu_\psi)}{(1+t^{2})^{2}}},
\]
which is positive for $t<1/\nu_\psi$ and negative for $t>1/\nu_\psi$, so the maximum over $t\ge0$ is attained at $t=1/\nu_\psi$ with value $1+\nu_\psi^{2}$, and is achieved exactly by $w=-r_\psi/\nu_\psi^{2}$, that is, by
\[
{{f^{\ast}}:=\operatorname{sgn}(\psi)\tilde{f} = \operatorname{sgn}(\psi)\Bigl\{1+\frac{1}{\nu_\psi^{2}}\Bigl(1-\frac{|\psi|}{E|\psi(\vZ)|}\Bigr)\Bigr\}}
\]
up to nonzero scalar multiples, where the second equality is due to $(\operatorname{sgn}\psi)^2\equiv 1$. Hence $\sup R=1+\nu_\psi^{2}$, attained uniquely up to scale, which is the worst-case inflation of Proposition~\ref{prop:scalar-ccar}.
\end{proof}

\subsection{Proof of inference properties}\label{app:proof inference}
Some additional notations are needed for the proof. Denote by $\vq:= \{q_s:s\in[\prod_j m_j]\}^\top$ the stratum probability $P({S}=s)$, by $\vq_{1:d_S}$ the first $d_S$ entries, and let $q_{\text{rem}}=1-\sum_{s\le d_S}q_s>0$ be the remainder. Further, for $a\in\{\aZ,\aS,\aCM\}$ let
\begin{align*}
    &\vW_i^{(\aZ)} = ({Z_{i,1}},\ldots, {Z_{i,d_Z}}),\\
    &\vW_i^{(\aS)} = (\stind{1}(\vZ_i),\ldots,\stind{d_S}(\vZ_i)),\\
    &\vW_i^{(\aCM)} = (\stind{1}(\vZ_i),\ldots,\stind{d_S}(\vZ_i),{\widetilde{Z}_{i,1},\ldots, \widetilde{Z}_{i,d_Z}}),
\end{align*}
be the design vectors with treatment removed and set $\bar{\vW}_1^{(r,a)} = n_1^{-1}\sum_{i=1}^{n}I_i^{(r)} \vW_i^{(a)}, \bar{\vW}_0^{(r,a)} = n_0^{-1}\sum_{i=1}^{n}(1-I_i^{(r)}) \vW_i^{(a)}$, $n_1,n_0$ are the number of subjects in treatment arm $1,0$, respectively.

\begin{remark}
    In practice, $\mu_s=E(\vZ\mid S=s)$ is not known and needs to be estimated before model fitting. We use the within-stratum sample mean $\bar{\vZ}_s:=N_s^{-1}\sum_{i:S_i=s}\vZ_i$, $N_s$ is the number of subjects in stratum $s$ to estimate. For fixed $m$, $\max_s\|\bar{\vZ}_s-\mu_s\|=O_p(n^{-1/2})$ by the Central Limit Theorem. Then, viewing $\hat{\theta}$ as a function of the population, or sample means, it follows by an application of the Mean-Value Theorem, that $\hat{\theta}(\bar{\vZ}_s:s\in [m]) - \hat{\theta}(\mu_s:s\in[m]) = o_p(n^{-1/2})$. Therefore the replacement induces no asymptotic difference. 
\end{remark}

Before proving Theorem \ref{theorem:combined}, we present two lemmas. The OLS estimator is decomposed into three components via the following lemma.
\begin{lemma}\label{lemma:OLS decomposition}
    For any design and analysis strategy $r\in \{\tCR,\tDC,\tCC\}, a\in\{\aZ,\aS\}$, the following decomposition for OLS estimator holds
    $$
        \sqrt{n}\vL(\hat{\vbeta}_{n}^{(r,a)}-\vbeta^{(a)}) =R_{n,1}^{(r,a)}+R_{n,2}^{(r,a)}+o_P(1)
    $$
    where $R_{n,1}^{(r,a)}, R_{n,2}^{(r,a)}$ are defined throughout the proof.
\end{lemma}
\begin{proof}
    Direct calculation yields
    \begin{align*}
    \sqrt{n}\vL(\hat{\vbeta}_{n}^{(r,a)}-\vbeta^{(a)})&=\sqrt{n}\vL\left(\frac{(\vG^{(r,a)})^\top\vG^{(r,a)}}{n}\right)^{-1}\frac{1}{n}(\vG^{(r,a)})^\top(\vY-\vG^{(r,a)}\vbeta^{(a)})\\&= \vL\left\{n[(\vG^{(r,a)})^\top\vG^{(r,a)}]^{-1}-(\vOmega^{(a)})^{-1}\right]\sqrt{n}\vv^{(a)}\\&\quad+\vL(\vOmega^{(a)})^{-1}{\vH_{n}^{(r,a)}}+\vL\left\{n[(\vG^{(r,a)})^\top\vG^{(r,a)}]^{-1}-(\vOmega^{(a)})^{-1}\right\}({\vH_{n}^{(r,a)}}-\sqrt{n}\vv^{(a)})\\&:=R_{n,1}^{(r,a)}+R_{n,2}^{(r,a)}+R_{n,3}^{(r,a)},
\end{align*}
where $\vOmega^{(a)}$ are constant matrices defined in Lemma \ref{lemma:block matrix},  ${\vH_n^{(r,a)}} := n^{-1/2}\sum_{i=1}^{n}\vG_i^{(r,a)}e_i^{(a)}$ and $\vv^{(a)} = E[\vG^{(r,a)}e^{(a)}]$ with $e^{(a)}:= Y-(\vG^{(r,a)})^\top \vbeta^{(a)}$ the generic working model residuals. Specifically,  
\begin{align*}
    \vv^{(\aZ)}&:=\tfrac12 E[g(\vZ)](1,1,0,\ldots,0)^\top, 
    \\\vv^{(\aS)} &:= {\tfrac12 E[g(\vZ)\mid {S}>d_S]\left(1,1, 2q_1,\ldots, 2q_{d_S}\right)^\top}
    \\\vv^{(\aCM)} &:= {\tfrac12 E[g(\vZ)\mid {S}>d_S]\left(1,1, 2q_1,\ldots, 2q_{d_S},0,\ldots,0\right)^\top}.
\end{align*}
We now claim that $n^{-1/2}\vH_n^{(r,a)}=\vv^{(a)}+o_P(1)$. To show this, decompose $e^{(a)} = h^{(a)} + \varepsilon$, where $h^{(a)}$ is the residual function (e.g., $h^{(\aZ)} = g-\vGamma^\top \vZ_{1:d_Z}$). Note that $\vG^{(r,a)}h^{(a)}$ is a function of $\vZ$ and thus, $n^{-1}\sum_{i=1}^{n}I_i^{(r)}\vG^{(r,a)}_ih^{(a)}(\vZ_i)\overset{P}\to \vv^{(a)}$ by Lemma \ref{lemma:wlln for car}. Moreover, $\{I_i\vG_i^{(r,a)}\varepsilon_i\}$ is a martingale difference sequence with respect to the filtration $\cF_i$, so $n^{-1}\sum_{i=1}^{n}I_i\vG_i^{(r,a)}\varepsilon_i\overset{P}\to {\vZero}$. This proves the claim and hence $R_{n,3}^{(r,a)} = o_P(1)$ by Lemma \ref{lemma:block matrix}. The calculation of $\vv^{(a)}$ is straightforward.
\end{proof}

\begin{lemma}\label{lemma:Rn1 expansion}
    For any fixed $(r,a)$, 
    $${R_{n,1}^{(r,a)} = -\frac{4v^{(a)}_1}{\sqrt{n}}D_n+o_P(1)},$$
    where $\vv^{(a)} = (v^{(a)}_1,v^{(a)}_2,\vw^\top)^\top$, $D_n$ is the overall imbalance.
\end{lemma}
\begin{proof}
    Recall that $R^{(r,a)}_{n,1}=\vL n[(\vG^{(r,a)})^{\top}\vG^{(r,a)}]^{-1}\sqrt{n}\vv^{(a)}$ with $\vL=(1,-1,{\vZero}^\top)$. Partition the regressors as $\vG^{(r,a)}=[\vT\mid\vW]$, where $\vT$ collects the two treatment indicators $(I^{(r)}_i,1-I^{(r)}_i)$ and $\vW$ collects the treatment-free regressors $\vW^{(a)}_i$. Write
\[
n\hat{\bm{\Sigma}}_{\vW}^{(r,a)}=\sum_{i=1}^{n}\vW^{(a)}_i(\vW^{(a)}_i)^{\top}
 -n_1\bar\vW^{(r,a)}_1(\bar\vW^{(r,a)}_1)^{\top}
 -n_0\bar\vW^{(r,a)}_0(\bar\vW^{(r,a)}_0)^{\top}
\]
for the residual cross-product of $\vW$ after projecting out $\vT$. Suppressing superscripts $r,a$, the matrix inversion formula gives
$$
(\vG^{\top}\vG)^{-1}
=\begin{pmatrix}
(\vT^{\top}\vT)^{-1}+(\vT^{\top}\vT)^{-1}\vT^{\top}\vW{(n\hat{\bm{\Sigma}}_{\vW})^{-1}}\vW^{\top}\vT(\vT^{\top}\vT)^{-1}
 & -(\vT^{\top}\vT)^{-1}\vT^{\top}\vW{(n\hat{\bm{\Sigma}}_{\vW})^{-1}}\\[2pt]
-{(n\hat{\bm{\Sigma}}_{\vW})^{-1}}\vW^{\top}\vT(\vT^{\top}\vT)^{-1} & {(n\hat{\bm{\Sigma}}_{\vW})^{-1}}
\end{pmatrix},
$$
where $\vT^{\top}\vT={\operatorname{diag}}\{n_1,n_0\}$ and the treatment--regressor block is
$\vT^{\top}\vW=(n_1\bar\vW^{(r,a)}_1,n_0\bar\vW^{(r,a)}_0)^{\top}$.

 We find the limit of $\hat{\bm{\Sigma}}_{\vW}^{(r,a)}$ by treating each piece. By the weak law
$n^{-1}\sum_i\vW^{(a)}_i(\vW^{(a)}_i)^{\top}\overset{P}{\to}E[\vW^{(a)}(\vW^{(a)})^{\top}]$.
By Lemma~\ref{lemma:wlln for car} with $f\equiv1$, $n_1/n$ and $n_0/n$ converge to
$\tfrac12$ in probability; and each arm mean satisfies $\bar\vW^{(r,a)}_1\overset{P}{\to}\vu^{(a)},\bar\vW^{(r,a)}_0\overset{P}{\to}\vu^{(a)}$, for every design $r$. Combining above, one derives
$$
\hat{\bm{\Sigma}}_{\vW}^{(r,a)}\overset{P}{\to}
E[\vW^{(a)}(\vW^{(a)})^{\top}]-\tfrac12\vu^{(a)}(\vu^{(a)})^{\top}
-\tfrac12\vu^{(a)}(\vu^{(a)})^{\top}
=\var(\vW^{(a)}),
$$
where $\var(\vW^{(a)})$ is the trailing principal submatrix of $\vOmega^{(a)}$ in Lemma \ref{lemma:block matrix}. Since matrix
inversion is continuous at any positive-definite matrix, the continuous mapping
theorem yields $(\hat{\bm{\Sigma}}_{\vW}^{(r,a)})^{-1}\overset{P}{\to}[\var(\vW^{(a)})]^{-1}$. It remains to identify $[\var(\vW^{(a)})]^{-1}$ and its inverse in each case which we show next.

For $a=\aZ$, $(\hat{\bm{\Sigma}}_{\vW}^{(r,\aZ)})^{-1}\overset{P}{\to}{\SigZd{d_Z}}^{-1}. $
For $a=\aS$, $\vW^{(\aS)}$ has the multinomial covariance matrix $\var(\vW^{(\aS)})={\operatorname{diag}}(\vq_{1:d_S})
-\vq_{1:d_S}\vq_{1:d_S}^{\top}$. Because $q_{\text{rem}}>0$, \cite[Eq (21)]{tanabe1992exact} gives
\begin{equation}
(\hat{\bm{\Sigma}}_{\vW}^{(r,\aS)})^{-1}\overset{P}\to{\operatorname{diag}}(\vq_{1:d_S}^{-1})+q_{\text{rem}}^{-1}{\bone}{\bone}^{\top}.
\end{equation}
Finally for $a=\aCM$, the stratum--covariate cross-covariance vanishes,
$\cov(\mathbbm{1}\{{S}=s\},{\widetilde{Z}_{k}})=q_s E[{\widetilde{Z}_{k}}\mid {S}=s]=0$, so $\var(\vW^{(\aCM)})$
is block diagonal. Hence,
\begin{equation}
(\hat{\bm{\Sigma}}_{\vW}^{(r,\aCM)})^{-1}\overset{P}\to\begin{pmatrix}{\operatorname{diag}}(\vq_{1:d_S}^{-1})+q_{\text{rem}}^{-1}{\bone}{\bone}^{\top} & {\vZero}\\ {\vZero} & {\SigZtd{d_Z}}^{-1}\end{pmatrix}. 
\end{equation}
Using $n_1-n_0=D_n$ and the partitioned
inverse above,
\[
{R^{(r,a)}_{n,1}
=v^{(a)}_1\sqrt{n}\Bigl(\frac{n}{n_1}-\frac{n}{n_0}\Bigr)
+\sqrt{n}
   \bigl(\bar\vW^{(r,a)}_1-\bar\vW^{(r,a)}_0\bigr)^{\top}(\hat{\bm{\Sigma}}_{\vW}^{(r,a)})^{-1}
   \{v^{(a)}_1\bigl(\bar\vW^{(r,a)}_1+\bar\vW^{(r,a)}_0\bigr)-\vw^{(a)}\}.}
\]
By Lemma~\ref{lemma:Rn1 converge} and the fact that $\vw^{(a)}=2v_1^{(a)}\vu^{(a)}$, we have $v_1^{(a)}[\bar\vW^{(r,a)}_1+\bar\vW^{(r,a)}_0]-\vw^{(a)}=o_P({\bone})$ and
$\sqrt{n}(\bar\vW^{(r,a)}_1-\bar\vW^{(r,a)}_0)
=2n^{-1/2}\sum_i(2I^{(r)}_i-1)(\vW^{(a)}_i-\vu^{(a)})+o_P(1) = O_P(1)$. Furthermore, $\sqrt{n}(\tfrac{n}{n_1}-\tfrac{n}{n_0})=-4D_n/\sqrt{n}+o_P(1)$. Recall the probability limit of $\hat{\bm{\Sigma}}_{\vW}^{(r,a)}$, results above yield the displayed expression. 
\end{proof}

We now present the proof of Theorem \ref{theorem:combined}. 

\begin{proof}[Proof of Theorem \ref{theorem:combined}]
    Recall $R_{n,2}^{(r,a)}=\bm L(\bm\Omega^{(a)})^{-1}\bm H_n^{(r,a)}$. Let $\vc=(c_1,c_2,\vc_{\vW}^\top)^\top$ be a constant vector, solving $\bm\Omega^{(a)}\bm c=\bm L^\top$ gives $c_1-c_2=4$ and $(c_1+c_2)q_{\text{rem}}/2=0$. Hence, $\vc = (2,-2,{\vZero}^\top)^\top$. This implies $\bm L(\bm\Omega^{(a)})^{-1}=(2,-2,{\vZero}^\top)$ for $a\in\{\aZ,\aS,\aCM\}$. Therefore,
$$R_{n,2}^{(r,a)}
 ={\frac{2}{\sqrt n}\sum_{i=1}^{n}(2I_i^{(r)}-1)e_i^{(a)}
 =\frac{2}{\sqrt n}\sum_{i=1}^{n}(2I_i^{(r)}-1)\bigl(e_i^{(a)}-2v_1^{(a)}\bigr)
  +\frac{4 v_1^{(a)}}{\sqrt n}D_n},$$
where the last step uses $\sum_{i=1}^{n}(2I_i^{(r)}-1)=D_n$. Per Lemma~\ref{lemma:Rn1 expansion}, 
\begin{equation}\label{eq:Rn1Rn2}
    {R_{n,1}^{(r,a)}+R_{n,2}^{(r,a)}
 =\frac{2}{\sqrt n}\sum_{i=1}^{n}(2I_i^{(r)}-1)\bigl(e_i^{(a)}-2v_1^{(a)}\bigr)+o_P(1).}
\end{equation}
Define the covariate-coefficient block of the pseudo-true $\vbeta^{(a)}$, i.e., the projection of $Y-\theta_1I_i^{(r)} - \theta_0(1-I_i)$ onto the column span of the working-model regressors, as
\begin{align*}
    {\bbetaW^{(\aZ)}} &= \vGamma= \SigZd{d_Z}^{-1}E[\vZ_{1:d_Z}g(\vZ)],\\
    {\bbetaW^{(\aS)}} &= \veta = \{E[g(\vZ)\mid {S}=s]:s\in [d_S]\}^\top-E[g(\vZ)\mid {S}>d_S],\\
    {\bbetaW^{(\aCM)}}&=((\bbetaW^{(\aS)})^\top, {\bgamt}^\top)^\top
\end{align*}
with ${\bgamt} = {\SigZtd{d_Z}}^{-1}E[{\vZt}_{1:d_Z}g(\vZ)]$. Then, substitute $e_i^{(a)}-2v_1^{(a)}
 =\{g(\bm Z_i)-E[g(\vZ)]\}-(\vW_i^{(a)}-\vu^{(a)})^\top{\bbetaW^{(a)}}+\varepsilon_i$. Lemma \ref{lemma:find var} evaluates the limiting variance of the leading term for each $r,a$, so that $R^{(r,a)}_{n,1}+R^{(r,a)}_{n,2}\rightsquigarrow
2[\sigma_\varepsilon^2+(\sigma^{(r,a)})^2]^{1/2}\xi$. Specifically, for $a\in\{\aZ,\aS,\aCM\},$ 
\begin{align}
(\sigma^{(\tCR,a)})^2 &= \var\bigl(g(\vZ)
   {-(\vW^{(a)})^\top{\bbetaW^{(a)}}\bigr)}, \label{eq:varCR}\\
(\sigma^{(\tDC,a)})^2 &= E\bigl\{\var[g(\vZ)
   {-(\vW^{(a)})^\top{\bbetaW^{(a)}}\mid {S}]}\bigr\}, \label{eq:varDC}\\
(\sigma^{(\tCC,\aZ)})^2 &= E[g^2(\vZ)]+2{\pi_{\Lambda}(\zeta_{g})}, \label{eq:varACC}\\
(\sigma^{(\tCC,\aS)})^2 &= E[g^2_{d_S}(\vZ)]+2{\pi_{\Lambda}(\zeta_{g_{d_S}})}, \label{eq:varsACC}
\end{align}
The explicit forms of $\pi_{\Lambda}(\cdot)$ are in Equations \eqref{eq:pi_Lambda g}-\eqref{eq:pi_Lambda gstar} and $g_{d_S}$ is defined before Equation \eqref{eq:pi_Lambda g}. Finally, if
$r=\tCC$ and $g$ is strongly balanced, then $\sigma^{(\tCC,\aZ)}\equiv 0$
\citep[Theorem 3.6]{ma2024new}. The theorem is established.
\end{proof}

\begin{lemma}\label{lemma:find var}
    The exact mathematical forms of $\sigma^{(r,a)}$ are derived in this lemma.
\end{lemma}
\begin{proof}
    Under each design we treat $a=\aZ$; the cases $a=\aS,\aCM$ are identical after changing the residual term $e_i^{(a)}-2v_1^{(a)}$. Under $\tCR$ the assignments $(2I^{(\tCR)}_i-1)$ are independent of $(\vZ_i,\varepsilon_i)$, so the central limit theorem gives the limit. We now consider different randomization designs.

\emph{DCAR.} Write $\tilde{h}^{(a)}:=h^{(a)}-2v_1^{(a)}$ for the centred residual function, so that $e^{(a)}-2v_1^{(a)}=\tilde{h}^{(a)}(\vZ)+\varepsilon$ with $E[\tilde{h}^{(a)}(\vZ)]=0$, and split $\tilde{h}^{(a)}$ into its within-stratum fluctuation and its stratum-level mean,
\begin{align*}
R^{(\tDC,\aZ)}_{n,1}+R^{(\tDC,\aZ)}_{n,2}
&=\frac{2}{\sqrt n}\sum_{i=1}^{n}(2I^{(\tDC)}_i-1)\bigl\{{\tilde{h}^{(a)}}(\vZ_i)-E[{\tilde{h}^{(a)}}(\vZ_i)\mid S_i]{+\varepsilon_i}\bigr\}\\
&\quad+\frac{2}{\sqrt n}\sum_{s\in[m]}D_n(\stind{s})E[{\tilde{h}^{(a)}}(\vZ)\mid {S}=s]+o_P(1).
\end{align*}
Conditioning on $\{S_1,\dots,S_n\}$, the summands of the first term are
mean-zero given the strata. The conditional central limit theorem therefore yields a normal
limit with variance $4\{\sigma_\varepsilon^2+E[\var({\tilde{h}^{(a)}}\mid {S})]\}$. The second term is $o_P(1)$ because of Condition \ref{assume: dgp}.

\emph{CCAR.} By Condition \ref{assume: dgp}, the linear adjustment is asymptotically
negligible for any fixed slope:
\[
\frac{1}{\sqrt n}\sum_{i=1}^{n}(2I^{(\tCC)}_i-1)\vGamma^{\top}\vZ_{i,1:d_Z}
=\sum_{j=1}^{d_Z}\gamma_j\frac{1}{\sqrt n}D_n(Z_j)=o_P(1),
\]
which holds regardless of the correlations among the covariates. It therefore suffices to prove, for any $h:\mathbb R^p\to
\mathbb R$ with $E[h^2(\vZ)]<\infty$, that
\[
\frac{1}{\sqrt n}\sum_{i=1}^{n}(2I^{(\tCC)}_i-1)[h(\vZ_i)+\varepsilon_i]
\rightsquigarrow\bigl[\sigma_\varepsilon^2+E[h^2(\vZ)]+2\pi_{\Lambda}(\zeta_h)\bigr]^{1/2}\xi;
\]
the claim then follows with $h=g-E[g]$ for $a=\aZ$ and $h=g_{d_S}$ for $a=\aS$.

Recall the martingale difference sequence $\Delta M_{h,i}$ defined in Lemma \ref{lemma:previous result1}(iv) ($r=\tCC$), by the telescoping sum and the Poisson equation,
    $$\sum_{i=1}^{n}(2{I_i^{(\tCC)}}-1)[h(\vZ_i)+\varepsilon_i]={M_{h,n}}+{\hat{\kappa}_h}(\Lambda_0)-\hat{\kappa}_h(\Lambda_n),$$
    where $M_{h,n} = \sum_{i=1}^{n}\Delta M_{h,i}$ is a martingale. Notice $\hat{\kappa}_h(\Lambda_0)-\hat{\kappa}_h(\Lambda_n)=o_P(\sqrt{n})$ by Lemma \ref{lemma:previous result1}, thus only the martingale term contributes to the final limiting distribution. By Martingale Central Limit Theorem \cite[Corollary 3.1]{hall2014martingale}, to show the convergence in distribution, it suffices to verify the following Lindeberg's condition and the convergence of predictive variance:
    \begin{align}
        &\sum_{i=1}^{n}E\left[({\Delta M_{h,i}}/\sqrt{n})^2\mathbbm{1}(|{\Delta M_{h,i}}/\sqrt{n}|>\epsilon)\mid \mathcal{F}_{i-1}\right]\overset{P}\to 0, \quad \text{ for all }\epsilon>0.\label{eq:Lindeberg condition}\\
        &\sum_{i=1}^{n}E\left[({\Delta M_{h,i}}/\sqrt{n})^2\mid \mathcal{F}_{i-1}\right] \overset{P}\to \sigma^2_{\varepsilon}+(\sigma^{(\tCC,\aZ)})^2.\label{eq:predictive var}
    \end{align}
    We first consider Condition \eqref{eq:Lindeberg condition}. Since $E({\Delta M_{h,i}})^2\leq E[h(\vZ_i)+\varepsilon_i]^2+E[\hat{\kappa}_h(\Lambda_i)]^2+E[\hat{\kappa}_h(\Lambda_{i-1})]^2<\infty$, for any $\epsilon>0$, we have
    $$\frac{1}{n}\sum_{i=1}^{n}E[{\Delta M_{h,i}}\mathbbm{1}(|{\Delta M_{h,i}}/\sqrt{n}|>\epsilon)]^2\leq \max_{i\in [n]}E[{\Delta M_{h,i}}\mathbbm{1}(|{\Delta M_{h,i}}/\sqrt{n}|>\epsilon)]^2\to 0.$$
    Hence, Condition \eqref{eq:Lindeberg condition} holds in view of the Markov's inequality. Next, we verify Condition \eqref{eq:predictive var}. Note the following decomposition 
    \begin{align*}
        &E[({\Delta M_{h,i}})^2\mid \mathcal{F}_{i-1}] \\&= E[h(\vZ)+\varepsilon]^2 + 2E\{(2I_i-1)[h(\vZ_i)+\varepsilon_i]\hat{\kappa}_h(\Lambda_{i-1})\mid \mathcal{F}_{i-1}\}\\&\quad+E[{\hat{\kappa}_h^2}(\Lambda_{i-1})\mid \mathcal{F}_{i-1}]-{\hat{\kappa}_h^2}(\Lambda_{i-1})\\&= E[h(\vZ)+\varepsilon]^2+[\zeta_h(\Lambda_{i-1})+\zeta_h(-\Lambda_{i-1})]+{\PL}{\hat{\kappa}_h^2}(\Lambda_{i-1})-{\hat{\kappa}_h^2}(\Lambda_{i-1})
    \end{align*}
    where $\zeta_h(\lambda) = E\{h(\vZ)\hat{\kappa}_h(\lambda+\psi(\vZ))[1+(1-2{\rho}){\operatorname{sgn}}({\langle \lambda,\psi(\vZ)\rangle})]\}$. Further, by the Ergodic theorem \cite[Theorem 1.10.2]{norris1998markov},
    $$\frac{1}{n}\sum_{i=1}^{n}E[({\Delta M_{h,i}})^2\mid \mathcal{F}_{i-1}]\overset{P}\to E[h(\vZ)+\varepsilon]^2+2{\pi_{\Lambda}(\zeta_h)}=:\sigma^2_{\varepsilon}+(\sigma_h^{(\tCC,\aZ)})^2.$$
    The proof is finished by recalling the forms of $R_{n,1}^{(\tCC,\aZ)}+R_{n,2}^{(\tCC,\aZ)}$ in \eqref{eq:Rn1Rn2} and define the corresponding $h$ function in view of Condition \ref{assume: dgp}.
\end{proof}

\begin{proof}[Proof of Theorem \ref{theorem:var limit}]
    Following \cite[Proof of Theorem 3.1]{ma2015testing} and using $\vL(\vOmega^{(a)})^{-1}=(2,-2,{\vZero}^\top)$ from the proof of Theorem \ref{theorem:combined}, we have
    $$\vL[(\vG^{(r,a)})^\top \vG^{(r,a)}]^{-1}\vL^\top = \frac{1}{n}\vL(\vOmega^{(a)})^{-1}\vL^\top+o_P\!\left(\frac{1}{n}\right) = \frac{4}{n}+o_P\!\left(\frac{1}{n}\right),$$
    where $\vOmega^{(a)}$ is defined in Lemma \ref{lemma:block matrix}. It remains to identify the probability limit of $(\hat{\sigma}_{\varepsilon}^{(r,a)})^2$. Write
    \begin{align*}
        (\hat{\sigma}_{\varepsilon}^{(r,a)})^2 &= \underbrace{\tfrac{1}{n}(\vY-\vG^{(r,a)}\vbeta^{(a)})^\top(\vY-\vG^{(r,a)}\vbeta^{(a)})}_{T_1}
        + \underbrace{\tfrac{1}{n}(\hat{\vbeta}^{(r,a)}_n-\vbeta^{(a)})^\top(\vG^{(r,a)})^\top \vG^{(r,a)}(\hat{\vbeta}^{(r,a)}_n-\vbeta^{(a)})}_{T_2}\\
        &\quad - 2\underbrace{\tfrac{1}{n}(\hat{\vbeta}^{(r,a)}_n-\vbeta^{(a)})^\top (\vG^{(r,a)})^\top(\vY-\vG^{(r,a)}\vbeta^{(a)})}_{T_3}.
    \end{align*}
    The weak law of large numbers gives, for every $r$,
    $$T_1 \overset{P}\to \sigma_\varepsilon^2 + E\big[(g(\vZ)-(\vW^{(a)})^\top\vbeta_{\vW}^{(a)})^2\big].$$
    For the remaining two terms, set $\bm\mu^{(a)}=2b^{(a)}(1,1,0,\ldots,0)^\top$. By Lemma \ref{lemma:inconsistency}, $\hat{\vbeta}^{(r,a)}_n-\vbeta^{(a)}=\bm\mu^{(a)}+o_P(\bm 1)$ for every $(r,a)$, and the same lemma gives the identity $\vv^{(a)}=\vOmega^{(a)}\bm\mu^{(a)}$. By Lemma \ref{lemma:block matrix}, $\tfrac1n(\vG^{(r,a)})^\top\vG^{(r,a)}\overset{P}\to\vOmega^{(a)}$ and $\tfrac1n(\vG^{(r,a)})^\top(\vY-\vG^{(r,a)}\vbeta^{(a)})\overset{P}\to\vv^{(a)}$. Hence,
    $$T_2 \overset{P}\to (\bm\mu^{(a)})^\top\vOmega^{(a)}\bm\mu^{(a)}, \qquad
      T_3 \overset{P}\to (\bm\mu^{(a)})^\top\vv^{(a)}=(\bm\mu^{(a)})^\top\vOmega^{(a)}\bm\mu^{(a)}.$$
    Since $\bm\mu^{(a)}$ is supported on its first two coordinates and the top $2\times2$ block of $\vOmega^{(a)}$ equals ${\operatorname{diag}}(1/2,1/2)$, this quadratic form is $(\bm\mu^{(a)})^\top\vOmega^{(a)}\bm\mu^{(a)}=4(b^{(a)})^2$. Therefore $T_2-2T_3\overset{P}\to-4(b^{(a)})^2$, and
    $$(\hat{\sigma}_{\varepsilon}^{(r,a)})^2 \overset{P}\to \sigma_\varepsilon^2 + E\big[(g(\vZ)-(\vW^{(a)})^\top\vbeta_{\vW}^{(a)})^2\big]-(2b^{(a)})^2.$$
    Because $E[g(\vZ)-(\vW^{(a)})^\top\vbeta_{\vW}^{(a)}]=2v_1^{(a)}=2b^{(a)}$ by the proof of Theorem \ref{theorem:combined}, the last two terms combine into $\var\!\big(g(\vZ)-(\vW^{(a)})^\top\vbeta_{\vW}^{(a)}\big)$, which equals $(\sigma^{(\tCR,a)})^2$ by Equation \eqref{eq:varCR}. Multiplying by $n\vL[(\vG^{(r,a)})^\top\vG^{(r,a)}]^{-1}\vL^\top\to4$ gives $4(\tau^{(\tCR,a)})^2$, as claimed.
    
\end{proof}

Now, we prove Proposition \ref{prop:opt precision} and Theorem \ref{theorem:compare tests}.

\begin{proof}[Proof of Proposition \ref{prop:opt precision}]
    For a correctly specified model, we can write
    \begin{equation}\label{eq:opt expansion}
        \sqrt{n}(\hat{\theta}_1-\hat{\theta}_0-\theta_1+\theta_0)=\frac{2}{\sqrt{n}}\sum_{i=1}^{n}(2I_i^{(r)}-1)\varepsilon_i+o_P(1)
    \end{equation}
    for any $r$. Similar CLT type arguments as in the proof of Lemma \ref{lemma:find var} implies (III) attains optimal precision.

    We now show (I)--(II). For (I), $r=\tDC,a=\aS$, since $g$ is strongly balanced, we have by Condition \ref{assume: dgp} that
    $$\frac{2}{\sqrt{n}}\sum_{i=1}^{n}(2I_i^{(\tDC)}-1)\left[g(\vZ_i)-\sum_{j=1}^{{d_S}}\eta_j\stind{j}(\vZ_i)\right]=o_P(1).$$
    Therefore, $R_{n,1}^{(\tDC,\aS)}+R_{n,2}^{(\tDC,\aS)}$ has the same distribution as Equation \eqref{eq:opt expansion} under $r=\tDC$, which implies (I) can attain optimal precision. (II) follows by analogous argument, using Condition \ref{assume: dgp} again. The proof is finished.
\end{proof}

\begin{proof}[Proof of Theorem \ref{theorem:compare tests}]
    The first result follows from Corollary \ref{cor: test under Rand} and \cite[Theorem 14.19]{vandevaart}. For the second result, recall from Equations \eqref{eq:varCR},\eqref{eq:varDC} that when $a=\aS$, 
    $$(\sigma^{(\tCR,\aS)})^2 = (\sigma^{(\tDC,\aS)})^2 + \var\{E[g(\vZ)-\sum_{j=1}^{{d_S}}\eta_j\mathbbm{1}\{{S}=j\}\mid {S}]\},$$
    so the inequality follows from the law of total variance and the statement is proved except the case with $d_S=m-1.$ Denote by $a_j= \eta_j-\eta_m$, we have
    \begin{align*}
        \var\{E[g(\vZ)-\sum_{j=1}^{m-1}a_j\mathbbm{1}\{{S}=j\}\mid S]\}&= \var\{E[g(\vZ)\mid {S}]-\sum_{j=1}^{m-1}\eta_j\mathbbm{1}\{{S}=j\}+\eta_m\sum_{j=1}^{m-1}\mathbbm{1}\{{S}=j\}\}\\&= \var[{\eta_m}\mathbbm{1}\{{S}=m\}+\eta_m\sum_{j=1}^{m-1}\mathbbm{1}\{{S}=j\}],
    \end{align*}
    where the random variable $\eta_m\mathbbm{1}\{{S}=m\}+\eta_m\sum_{j=1}^{m-1}\mathbbm{1}\{{S}=j\}=\eta_m$ almost surely. Hence, the equality holds when ${d_S=m-1}$.
\end{proof}

Finally, we prove Proposition \ref{prop:better model}, which gives quantification of when and how much the CM model is better than ANCOVA or SFE models. 

\begin{proof}[Proof of Proposition \ref{prop:better model}]
    Recall by Equations \eqref{eq:varDC}, for $a\in\{\aZ,\aS,\aCM\}$, 
    \begin{align*}
        &(\sigma^{(\tDC,\aS)})^2 = E[\var(g\mid S)],\\
        &(\sigma^{(\tDC,a)})^2 = Q(\mathbf{b}^{(a)}) \quad \text{for }a\in\{\aZ,\aCM\},
    \end{align*}
    where $Q(\mathbf{b}):=E\{\var[g-\mathbf{b}^\top{\vZt}_{1:d_Z}\mid S]\}$, $\mathbf{b}^{(\aZ)}=\vGamma, \mathbf{b}^{(\aCM)}={\bgamt}$.
    
    Expanding the conditional variance we have
    $$Q(\mathbf{b}) = E[\var(g\mid S)] - 2\mathbf{b}^\top E[\cov({\vZt}_{1:d_Z},g\mid S)]+\mathbf{b}^\top E[\var({\vZt}_{1:d_Z}\mid S)]\mathbf{b}.$$
    Notice that $E({\vZt}\mid S)=0$, we have that $E[\cov({\vZt}_{1:d_Z},g\mid S)] = E({\vZt}_{1:d_Z}g)$. Also, ${\SigZtd{d_Z}}{\bgamt} = E({\vZt}_{1:d_Z}g)$, it follows that
    $$Q(\mathbf{b}) =  E[\var(g\mid S)] - {\bgamt}^\top {\SigZtd{d_Z}}{\bgamt}+(\mathbf{b}-{\bgamt})^\top {\SigZtd{d_Z}}(\mathbf{b}-{\bgamt}).$$
    Taking $\mathbf{b} = \mathbf{b}^{(\aZ)}$ gives Equations \eqref{eq:gain-Z}. The minimum is achieved when $\mathbf{b}=\mathbf{b}^{(\aCM)}$, Equation \eqref{eq:gain-S} follows from recalling $(\sigma^{(\tDC,\aS)})^2$. The proof is finished.

\end{proof}

\subsection{Auxiliary Lemma}\label{app:aux lemma}
A novel law of large number result for CAR procedures is introduced, which is used multiple times in the technical proof.
\begin{lemma}\label{lemma:wlln for car}
    Let $f:\bR^{p}\to\bR$ be such that $E[f^2(\vZ)]<\infty$. Then, for any randomization design $r$,
    $$\frac{1}{n}\sum_{i=1}^{n}I_i^{(r)} f(\vZ_i) \overset{P}\to \frac{1}{2}E[f(\vZ)], \quad \frac{1}{n}\sum_{i=1}^{n}(1-I_i^{(r)})f(\vZ_i) \overset{P}\to \frac{1}{2}E[f(\vZ)].$$
\end{lemma}
\begin{proof}
    We prove the first statement, the proof for the second statement is similar. If $r = \tCR$, then the lemma follows directly by weak law of large numbers. If $r = \tDC,\tCC, \tMC$, then by \cite[Theorem 3.5]{ma2024new},
    $$\frac{1}{n}\sum_{i=1}^{n}(2I_i^{(r)}-1)f(\vZ_i) \overset{P}\to 0.$$
    Furthermore,
    $$
        \frac{1}{n}\sum_{i=1}^{n}I_i^{(r)} f(\vZ_i) = \frac{1}{n}\sum_{i=1}^{n}(2I_i^{(r)}-1)f(\vZ_i)+\frac{1}{n}\sum_{i=1}^{n}(1-I_i^{(r)})f(\vZ_i).
    $$
    This implies
    $$\frac{2}{n}\sum_{i=1}^{n}I_i^{(r)} f(\vZ_i) = E[f(\vZ)] +o_P(1)$$ 
    by the weak law of large numbers. 
\end{proof}

The following lemma finds the probability limits of the design matrices under different randomization and analysis procedures. 
\begin{lemma}\label{lemma:block matrix}
    Let $\vOmega^{(\aZ)},\vOmega^{(\aS)},\vOmega^{(\aCM)}$ be block diagonal matrices of the form 
     \begin{align*}
        &\vOmega^{(\aZ)} = 
        \begin{pmatrix}
            \frac{1}{2}  & 0&\vZero^\top\\ 0& \frac{1}{2}  & \vZero^\top\\
            \vZero &\vZero& {{\SigZd{d_Z}}}
        \end{pmatrix},\quad {\vOmega^{(\aS)} = \begin{pmatrix}
            \frac{1}{2}  & 0 &\frac{1}{2}\vq_{1:d_S}^\top\\ 0& \frac{1}{2} & \frac{1}{2}\vq_{1:d_S}^\top\\
           \frac{1}{2}\vq_{1:d_S} &\frac{1}{2}\vq_{1:d_S}& {\operatorname{diag}}(\vq_{1:d_S})
        \end{pmatrix}},\\
        &\vOmega^{(\aCM)} = \begin{pmatrix}
            \bm{\Omega}^{(\aS)}&\vZero^\top\\
            \vZero & {{\SigZtd{d_Z}}}
        \end{pmatrix}
     \end{align*}
 Then, $n^{-1}(\vG^{(r,a)})^\top\vG^{(r,a)} \overset{P}\to\vOmega^{(a)}, n[(\vG^{(r,a)})^\top\vG^{(r,a)}]^{-1} \overset{P}\to (\vOmega^{(a)})^{-1}$ for any fixed $(r,a)$.
\end{lemma}
\begin{proof}
    Throughout, convergence of random matrices means entry-wise convergence in probability. Since all dimensions are fixed, this is equivalent to convergence in any matrix norm.

    Recall $\vG_i^{(r,a)}=(I_i^{(r)},1-I_i^{(r)},(\vW_i^{(a)})^{\top})^{\top}$ and partition $$(\vG^{(r,a)})^{\top}\vG^{(r,a)}=\begin{pmatrix}\vG_{11}&\vG_{12}\\ \vG_{21}&\vG_{22}\end{pmatrix},$$
    where
    $$\vG_{11}=\begin{pmatrix}n_1&0\\0&n_0\end{pmatrix},\quad
\vG_{12}=\vG_{21}^{\top}
=\begin{pmatrix}\sum_{i=1}^{n}I_i^{(r)}(\vW_i^{(a)})^{\top}\\[2pt]
\sum_{i=1}^{n}(1-I_i^{(r)})(\vW_i^{(a)})^{\top}\end{pmatrix},\quad
\vG_{22}=\sum_{i=1}^{n}\vW_i^{(a)}(\vW_i^{(a)})^{\top}.
$$
For the blocks containing the treatment indicators we apply
Lemma~\ref{lemma:wlln for car} entrywise, for every design $r$ considered: with $f\equiv1$ it gives $n_1/n\overset{P}\to\tfrac12$ and
$n_0/n\overset{P}\to\tfrac12$, and with $f$ equal to a component of
$\vW^{(a)}$ it gives $n^{-1}\vG_{12}\overset{P}{\to}
{\bone}_2E[\vW^{(a)}]^{\top}/2$. The weak law of large numbers gives
$n^{-1}\vG_{22}\overset{P}{\to}E[\vW^{(a)}(\vW^{(a)})^{\top}]$, and it
remains to evaluate this matrix. For $a=\aZ$, $E[\vZ_{1:d_Z}\vZ_{1:d_Z}^{\top}]={\SigZd{d_Z}}$. For
$a=\aS$, so the block is
$\operatorname{diag}(\vq_{1:d_S})$. For $a=\aCM$, the first block is $\operatorname{diag}(\vq_{1:d_S})$, and $E[{\vZt}_{1:d_Z}{\vZt}_{1:d_Z}^{\top}]={\SigZtd{d_Z}}$; for the cross block, for each $s\in[d_S]$,
$$
E\big[\mathbbm{1}\{{S}=s\}{\vZt}_{1:d_Z}^{\top}\big]
=q_s E\big[{\vZt}_{1:d_Z}\mid {S}=s\big]^{\top}
=q_s \big(E[\vZ_{1:d_Z}\mid {S}=s]-E[\vZ_{1:d_Z}\mid {S}=s]\big)^{\top}
={\vZero}^{\top}.
$$
Collecting the blocks proves the
first claim for every $a$.

To prove the second statement, we show that the matrices are invertible. The matrix $\vOmega^{(\aZ)}$ is block diagonal with ${\SigZd{d_Z}}$ a principal submatrix of $\bm{\Sigma}_{\bm{Z}}$, so $\vOmega^{(\aZ)}\succ 0$. For $\vOmega^{(\aS)}$, the block $\operatorname{diag}(\vq_{1:d_S})$ is positive definite because $q_{{s}}>0$ for each ${s\in[d_S]}$, and the Schur complement of this block in $\vOmega^{(\aS)}$ equals $\tfrac12\mathbf{I}_2-\tfrac14(\textstyle\sum_{{s}\le d_S}q_{{s}}){\bone}_2{\bone}_2^{\top}$, whose eigenvalues are positive; thus $\vOmega^{(\aS)}\succ0$. Hence, $\vOmega^{(\aCM)}$ is block diagonal and positive definite. The smallest eigenvalue of
$n^{-1}(\vG^{(r,a)})^{\top}\vG^{(r,a)}$ converges in probability to
$\lambda_{\min}(\vOmega^{(a)})>0$, so $(\vG^{(r,a)})^{\top}\vG^{(r,a)}$ is
nonsingular with probability tending to one. Since matrix inversion is
continuous, with respect to the operator norm,
the continuous mapping theorem applied at $\vOmega^{(a)}$ yields
$n[(\vG^{(r,a)})^{\top}\vG^{(r,a)}]^{-1}\overset{P}{\to}(\vOmega^{(a)})^{-1}$. This finishes the proof.
\end{proof}

Lemma \ref{lemma:Rn1 converge} is an important lemma used to derive the asymptotic expansion of $R_{n,1}^{(r,a)}$ in the proof of Theorem \ref{theorem:combined}.

\begin{lemma}\label{lemma:Rn1 converge}
    Recall the definitions of $\vW_i^{(a)}, \vq$ in Appendix \ref{app:proof inference}, we have
    \begin{align*}
        \bar{\vW}_1^{(r,a)}+\bar{\vW}_0^{(r,a)} &= 2\vu^{(a)}+o_P({\bone}),
        \\\sqrt{n}(\bar{\vW}_1^{(r,a)}-\bar{\vW}_0^{(r,a)})&=
            \frac{2}{\sqrt{n}}\sum_{i=1}^{n}(2I_i^{(r)}-1)(\vW_i^{(a)}-\vu^{(a)})+o_P({\bone}),
    \end{align*}
    where $\vu^{(\aZ)} = (0,\ldots,0)^\top$ is an $d_Z$-dimensional vector, $\vu^{(\aS)} = \vq_{1:d_S}$ is a $d_S$-dimensional vector, $\vu^{(\aCM)} = (\vq_{1:d_S},{\vZero})^\top$ is a $(d_S+d_Z)$-dimensional vector.
\end{lemma}
\begin{proof}
    Note that
    $$\bar{\vW}_1^{(r,a)}+\bar{\vW}_0^{(r,a)} = \frac{n}{n_1}\frac{1}{n}\sum_{i=1}^{n}I_i^{(r)} \vW_i^{(a)} + \frac{n}{n_0}\frac{1}{n}\sum_{i=1}^{n}(1-I_i^{(r)}) \vW_i^{(a)}.$$
    The first claim follows since $n/n_a\overset{P}\to 2$ for $a\in\{0,1\}$ and using Lemma \ref{lemma:wlln for car}. Further, notice
    \begin{align*}
        \sqrt{n}(\bar{\vW}_1^{(r,a)}-\bar{\vW}_0^{(r,a)})&=\frac{n}{n_1}\sum_{i=1}^{n}(2I_i^{(r)}-1)\vW_i^{(a)}+\frac{n}{n_1}\sum_{i=1}^{n}(1-I_i^{(r)})\vW_i^{(a)}-\frac{n}{n_0}\sum_{i=1}^{n}(1-I_i^{(r)})\vW_i^{(a)}\\
        &=\frac{2}{\sqrt{n}}\sum_{i=1}^{n}(2I_i^{(r)}-1)\vW_i^{(a)}-\frac{n^2}{n_1n_0}\frac{D_n}{\sqrt{n}}\frac{1}{n}\sum_{i=1}^{n}I_i^{(r)}\vW_i^{(a)}+o_P({\bone}).
    \end{align*}
    The second statement follows by applying $n^2/(n_1n_0)\overset{P}\to 4$ and Lemma \ref{lemma:wlln for car} again. The proof is finished.
\end{proof}

The next lemma shows the inconsistency of the single treatment effect estimators $\hat{\theta}_1,\hat{\theta}_0$, it also proves the consistency of covariate effect and the contrast effect $\hat{\theta}_1-\hat{\theta}_0$. 

\begin{lemma}[Inconsistency of the OLS estimators]\label{lemma:inconsistency}
    Suppose Condition \ref{assume: dgp} holds. Then, for any $(r,a)$,
    $$\hat{\vbeta}^{(r,a)}_n - \vbeta^{(a)} = {(\vOmega^{(a)})^{-1}\vv^{(a)}+o_P({\bone})=2{b^{(a)}}(1,1,0,\ldots,0)^\top+o_P({\bone})},$$
    where $\vv^{(a)}$ is given in Lemma \ref{lemma:OLS decomposition},
    $$b^{(\aZ)} = \frac{1}{2}E[g(\vZ)],\quad b^{(\aS)}=b^{(\aCM)} = \frac{1}{2}E[g(\vZ)\mid {S}>d_S].$$
    In particular $\hat{\theta}_t-\theta_t\overset{P}\to 2b^{(a)}$, $\hat{\theta}_1-\hat{\theta}_0\overset{P}\to 0$.
\end{lemma}
\begin{proof}
    Fix $(r,a)$ and recall the definitions of $\vG^{(r,a)}, e^{(a)}$ from Lemma \ref{lemma:OLS decomposition}, the LS identity gives 
    $${\hat{\vbeta}_n^{(r,a)}-\vbeta^{(a)}= \left[n((\vG^{(r,a)})^\top\vG^{(r,a)})^{-1}\right]\frac{1}{n}\sum_{i=1}^{n}G_i^{(r,a)}e_i^{(a)}}.$$
    By Lemma \ref{lemma:block matrix}, the first factor converges to $(\vOmega^{(a)})^{-1}$ in probability. The second factor is $n^{-1/2}\vH_n^{(r,a)}$, which we have shown in Lemma \ref{lemma:OLS decomposition} that it converges to $\vv^{(a)}$ in probability. The first statement is proved. The treatment effect $\hat{\theta}_1-\hat{\theta}_0  = \vL\hat{\vbeta}$ is consistent because $\vL (1,1,0,\ldots,0)^\top = 0$. The proof is finished.
\end{proof}

Next, we state and prove Lemma \ref{lemma: compare TZ, TS under CR,D,C quadratic}, which is primarily used to compare the limiting variances of treatment effect estimation under CR and DCAR.

\begin{lemma}\label{lemma: compare TZ, TS under CR,D,C quadratic}
    Suppose Assumption \ref{assume:independent copy} holds, and recall $\vGamma={\SigZd{d_Z}}^{-1}E[\vZ_{1:d_Z}g(\vZ)], \eta_j= E[g(\vZ)\mid {S}=j]-E[g(\vZ)\mid {S}>d_S]$. Then,
    \begin{align}
    &(\sigma^{(\tCR,\aS)}_{d_S})^2 = (\sigma^{(\tCR,\aZ)}_{d_Z})^2+\var(\vGamma^\top\vZ_{1:d_Z})-\var(\sum_{j=1}^{d_S}\eta_j\stind{j}), \label{eq:cr relation}\\
    &(\sigma^{(\tDC,\aZ)}_{d_Z})^2=(\sigma^{(\tDC,\aS)}_{d_S})^2 + E\left[\var\left(\vGamma^\top\vZ_{1:d_Z} \mid {S}\right)\right]-2E\left\{\cov(g(\vZ),\vGamma^\top\vZ_{1:d_Z}\mid {S})\right\}. \label{eq:dcar relation}
    \end{align}
\end{lemma}

\begin{proof}[Proof of Lemma \ref{lemma: compare TZ, TS under CR,D,C quadratic}]
    We first prove Equation \eqref{eq:cr relation}. By the normal equations ${\SigZd{d_Z}}\vGamma = E[\vZ_{1:d_Z}g(\vZ)]$ and Assumption \ref{assume:independent copy},
    $$\cov\left(g(\vZ),\vGamma^\top\vZ_{1:d_Z}\right) = \vGamma^\top E[\vZ_{1:d_Z}g(\vZ)] = \vGamma^\top {\SigZd{d_Z}}\vGamma = \var\left(\vGamma^\top \vZ_{1:d_Z}\right).$$
    Analogously, note that $E[g(\vZ)]-\sum_{j\leq d_S}\eta_jq_j=E[g(\vZ)\mid {S}>d_S]$, whence, for every $k\leq d_S$,
    $$
    \cov\left(g-\sum_{j=1}^{d_S}\eta_j\stind{j},\stind{k}\right) = q_k\left\{E[g\mid {S}=k]-\eta_k-E[g\mid {S}>d_S]\right\}.$$
    So that, by Equation \eqref{eq:varCR}, $$\sigma^{(\tCR,a)}=\var[g(\vZ)]-\var\left((\vW^{(a)})^\top {\bbetaW^{(a)}}\right).$$
    Taking the difference of the above display for different $a$ yields Equation \eqref{eq:cr relation}. 

    We now prove Equation \eqref{eq:dcar relation}. Because $\sum_j \eta_j\stind{j}$ is $\sigma(S)-$measurable, $\var(g-\sum_j\eta_j\stind{j}\mid {S}) = \var(g\mid {S})$. Expanding Equation \eqref{eq:varDC} for different $a$, we have
    $$(\sigma^{(\tDC,\aZ)})^2=E[\var(g\mid {S})]+E[\var(\vGamma^\top \vZ_{1:d_Z}\mid {S})] - 2E\{\cov(g(\vZ),\vGamma^\top \vZ_{1:d_Z}\mid {S})\}.$$
    This is Equation \eqref{eq:dcar relation}. The proof is finished.
    
\end{proof}

\section{Bootstrap Adjusted tests}\label{app:bootstrap}

\subsection{Bootstrap consistency}
We use the following moment conditions.

\begin{assumption}\label{assume:bootstrap}
For each working model $a$:

\noindent(1) $E\|\vW^{(a)}\|^4<\infty$, $E|e^{(a)}|^4<\infty$, and $E\|\vW^{(a)}e^{(a)}\|^4<\infty$, where $e^{(a)}=Y-(\vG^{(r,a)})^\top\vbeta^{(a)}$ is the working-model residual.

\noindent(2) If $r=\tDC$, then, for every $q_0>0$, some $C_{q_0}<\infty$ satisfies
\[
\sup_{Q:\min_s Q(S=s)\geq q_0}\ \sup_{n\geq1}
E_Q\!\left\{\max_{s\leq m}|D_n(\chi_s)|^4\right\}\leq C_{q_0}.
\]
\end{assumption}

We treat the null and local alternatives together by writing the observed response as
$$
Y_i=\theta+g(\vZ_i)+\varepsilon_i+\frac{\delta}{\sqrt n}I_i,
\qquad \delta\in\bR.
$$
Thus, $\delta=0$ is the null and $(\theta_1,\theta_0)=(\theta+\delta/\sqrt n,\theta)$ under the local alternative. Throughout this appendix, $P_*$, $E_*$, and $\var_*$ are conditional on the full observed sample $\cD_n=\{(\vZ_i,I_i,Y_i):i\leq n\}$ and refer only to bootstrap randomness.

\begin{proof}[Proof of Theorem \ref{theorem:bootstrap consistency}]
By Lemma~\ref{lemma:conditional moments},
$$
n\var_*(\hat\theta_1^*-\hat\theta_0^*\mid\cA_n^{(a)})
=\var_*(X_n^*\mid\cA_n^{(a)})
\overset P\to4(\tau^{(r,a)})^2.
$$
Lemma~\ref{lemma:bootstrap variance consistency}, applied with $k=4$, shows that $n\hat v_{*,n,B}^{(r,a)}$ has the same limit. Lemma~\ref{lemma:empirical residual moments} ensures that this limit is unchanged under every fixed local alternative. Finally, Theorem~\ref{theorem:combined} gives
$$
\sqrt n(\hat\theta_1-\hat\theta_0)
\rightsquigarrow\mathcal N\left(\delta,4(\tau^{(r,a)})^2\right).
$$
The asserted studentized limit follows from Slutsky's theorem when $\tau^{(r,a)}>0$.
\end{proof}

Within model~\eqref{eq:true model}, the adjustment does not require knowing $g$ or correctly specifying the working regression. When $\tau^{(r,a)}>0$, it yields a valid test for CR and the covered DCAR procedures, while preserving the ARE comparisons in Theorem~\ref{theorem:compare tests}. The CCAR results reported in Section~\ref{sec:numerical} remain empirical.

\subsection{Auxiliary Lemmas}\label{app:bootstrap auxiliary}

Fix $a$ and suppress this superscript unless strictly needed. Let $\hat P_n$ be the empirical law of $(\vW_j,Y_j)$, $j\leq n$, and put
\begin{align*}
\bar{\vW}&=\hat P_n\vW, \qquad \bar Y=\hat P_nY,\\
\hat\Sigma_{\vW}&=\hat P_n\{(\vW-\bar{\vW})(\vW-\bar{\vW})^\top\},\\
\vb_n&=\hat\Sigma_{\vW}^{-1}
\hat P_n\{(\vW-\bar{\vW})(Y-\bar Y)\}.
\end{align*}
The empirical projection coefficient and residual (minimizing the LS error defined by $\hat{P}_n$) are
\[
\vbeta_n^\dag=
\bigl(\bar Y-\bar{\vW}^\top\vb_n,\bar Y-\bar{\vW}^\top\vb_n,\vb_n^\top\bigr)^\top,
\qquad
e_{n,j}=Y_j-\bar Y-(\vW_j-\bar{\vW})^\top\vb_n.
\]
For a bootstrap draw, set $\vW_i^*=\vW_{J_i^*}$,
$\vG_i^*=(I_i^*,1-I_i^*,(\vW_i^*)^\top)^\top$, and
\[
\vQ_n^*=\frac1n\sum_{i=1}^n\vG_i^*(\vG_i^*)^\top,\qquad
\vQ_n^0=
\begin{pmatrix}
1/2&0&\bar{\vW}^\top/2\\
0&1/2&\bar{\vW}^\top/2\\
\bar{\vW}/2&\bar{\vW}/2&\hat P_n(\vW\vW^\top)
\end{pmatrix}.
\]
The replicate is accepted on
$\cA_n=\{\lambda_{\min}(\vQ_n^*)\geq\epsilon_0\}$, where
$0<\epsilon_0<\lambda_{\min}(\vOmega^{(a)})$, $\hat{\vbeta}_n^*$ denotes its LS coefficient. If $\hat\Sigma_{\vW}$ is singular, define the inverses and residuals arbitrarily. The probability of this event tends to zero by the weak law of large numbers. Finally, set $X_n^*=\sqrt n\vL(\hat{\vbeta}_n^*-\vbeta_n^\dag)$ on $\cA_n$ and $X_n^*=0$ otherwise. We write $\rightsquigarrow_*$ for conditional weak convergence.

\begin{lemma}[Empirical projection identities]\label{lemma:empirical product projection}
If $\hat\Sigma_{\vW}$ is nonsingular, then
\[
\vL\vbeta_n^\dag=0,\qquad \hat P_ne_n=0,\qquad
\hat P_n(\vW e_n)=\vZero,\qquad
\vL(\vQ_n^0)^{-1}=(2,-2,0,\ldots,0).
\]
Consequently, $X_n^*=\sqrt n(\hat\theta_1^*-\hat\theta_0^*)$ on $\cA_n$.
\end{lemma}
\begin{proof}
The first identity is immediate, and the next two are the normal equations for the empirical LS projection. For the last, multiply $\vQ_n^0$ by $(2,-2,0,\ldots,0)^\top$ and use symmetry.
\end{proof}

The next lemma connects the empirical product-projection residuals with the limiting variances of Theorem \ref{theorem:combined} and provides the moment bounds used repeatedly below. Recall from Assumption \ref{assume:bootstrap} that $e^{(a)}=Y-(\vG^{(r,a)})^\top\vbeta^{(a)}$. Under the null it equals $g(\vZ)-(\vW^{(a)})^\top\bbetaW^{(a)}+\varepsilon$, free of $r$ and of the treatment indicator. Equations \eqref{eq:varCR}--\eqref{eq:varsACC} give
\begin{equation}\label{eq:residual variance identification}
\var(e^{(a)})=(\tau^{(\tCR,a)})^2,\qquad E\{\var(e^{(a)}\mid S)\}=(\tau^{(\tDC,a)})^2.
\end{equation}
Moreover, under the null, $\var(\vW^{(a)})\bbetaW^{(a)}=\cov(\vW^{(a)},Y)$, that is, $\bbetaW^{(a)}$ is the slope of the population projection of $Y$ onto $(1,(\vW^{(a)})^\top)$.

\begin{lemma}[Empirical residual moments]\label{lemma:empirical residual moments}
Suppose Assumption \ref{assume:bootstrap}(1) holds and let $\hat{P}_n$ be the empirical law of $\{(\vW_j^{(a)},Y_j):j\leq n\}$, where $Y_j$ follows the local alternative with fixed $\delta\in\bR$. Let $\hat{q}_s=\hat{P}_n\{S=s\}$ and let $\widehat{\var}_n(\cdot\mid S=s)$ denote the empirical variance within stratum $s$. Then,
\begin{align}
&\hat{P}_n e_n^2 \overset{P}\to (\tau^{(\tCR,a)})^2,\qquad \sum_{s=1}^{m}\hat{q}_s\widehat{\var}_n(e_n\mid S=s)\overset{P}\to (\tau^{(\tDC,a)})^2,\label{eq:residual second moments}\\
&\hat{P}_n|e_n|^4+\hat{P}_n\|\vW e_n\|^4 = O_P(1).\label{eq:residual fourth moments}
\end{align}
\end{lemma}

\begin{proof}
Consider first $\delta=0$ and write $\breve e=e^{(a)}-E(e^{(a)})$. Because $\bbetaW^{(a)}$ is the population projection slope,
$E\{(\vW-E\vW)\breve e\}=\vZero$. The empirical normal equations give
\begin{equation}\label{eq:slope expansion}
\vb_n-\bbetaW^{(a)}
=\hat\Sigma_{\vW}^{-1}\hat P_n\{(\vW-\bar{\vW})\breve e\},
\qquad
e_{n,j}=\breve e_j-\hat P_n\breve e
-(\vW_j-\bar{\vW})^\top(\vb_n-\bbetaW^{(a)}).
\end{equation}
Chebyshev's inequality yields
\[
\vb_n-\bbetaW^{(a)}=O_P(n^{-1/2}),
\qquad \hat P_n\breve e=O_P(n^{-1/2}).
\]
It follows from \eqref{eq:slope expansion} that
\[
\hat P_n(e_n-\breve e)^2
\leq2(\hat P_n\breve e)^2
+2\|\vb_n-\bbetaW^{(a)}\|^2\hat P_n\|\vW-\bar{\vW}\|^2
=O_P(n^{-1}).
\]
The weak law of large numbers and \eqref{eq:residual variance identification} now give both limits in \eqref{eq:residual second moments}; for the conditional variance, apply the same $L^2$ bound within each of the finitely many strata.

For \eqref{eq:residual fourth moments}, apply the $c_r$ inequality to \eqref{eq:slope expansion}. All resulting empirical fourth moments are $O_P(1)$ by Assumption~\ref{assume:bootstrap}(1), except possibly the term
\[
\|\vb_n-\bbetaW^{(a)}\|^4\hat P_n\|\vW\|^8
\leq O_P(n^{-2})\max_{j\leq n}\|\vW_j\|^4\hat P_n\|\vW\|^4=o_P(1),
\]
where $n^{-1}\max_{j\leq n}\|\vW_j\|^4\overset P\to0$. This proves the fourth-moment bound.

It remains only to transfer these conclusions to a fixed local alternative. Let $\vb_n(\delta)$ and $e_{n,j}(\delta)$ denote the empirical projection slope and residual when the treatment coefficient is $\delta/\sqrt n$. Since the response changes pointwise by at most $|\delta|/\sqrt n$, linearity of LS projection and the Cauchy--Schwarz inequality give
\[
\|\vb_n(\delta)-\vb_n(0)\|=O_P(n^{-1/2}),\qquad
\hat P_n\{e_n(\delta)-e_n(0)\}^2=O_P(n^{-1}).
\]
The same expansion, together with $n^{-1}\max_{j\leq n}\|\vW_j\|^4=o_P(1)$, yields
\[
\hat P_n|e_n(\delta)-e_n(0)|^4+
\hat P_n\|\vW\{e_n(\delta)-e_n(0)\}\|^4=o_P(1).
\]
The $L^2$ triangle inequality, applied also within each of the finitely many strata, and the $c_r$ inequality now transfer \eqref{eq:residual second moments}--\eqref{eq:residual fourth moments} from $\delta=0$ to every fixed $\delta$.
\end{proof}

\begin{lemma}\label{lemma:signed sums}
    Let $(S_{n,i},V_{n,i}), i\leq n$ be i.i.d. from some distribution $Q_n$ and let the assignments $I_1,\ldots,I_n$ be generated sequentially by CR or by a DCAR procedure applied to $(S_{n,i})_{i\leq n}$. For DCAR, assume $\inf_{n,s}Q_n(S=s)\geq q_0>0$ and Assumption~\ref{assume:bootstrap}(2). Take $\bar{V}_{n,s}=Q_n(V_n\mid S=s)$ and $U_{n,i} = V_{n,i}-\bar{V}_{n,S_{n,i}}$. Then, under CR, $n^{-1/2}\sum_i (2I_i-1)V_{n,i}\rightsquigarrow \mathcal{N}(0,v)$ whenever $Q_nV_n^2\to v$ and $Q_n\{V_n^2\mathbbm{1}(|V_n|>\epsilon\sqrt{n})\}\to0$ for every $\epsilon>0$. Moreover, for a positive constant $C$,
    \begin{equation}\label{eq:cr fourth}
        E_{Q_n}\left|\frac{1}{\sqrt{n}}\sum_{i=1}^{n}(2I_i-1)V_{n,i}\right|^4\leq C\left\{(Q_nV_n^2)^2+n^{-1}Q_nV_n^4\right\}.
    \end{equation}
    For DCAR, $\sum_{i=1}^{n}(2I_i-1)V_{n,i}=\sum_{i=1}^{n}(2I_i-1)U_{n,i}+\sum_{s=1}^{m}D_n(\chi_s)\bar{V}_{n,s}$, the summands $(2I_i-1)U_{n,i}$ form a martingale difference sequence with
    \begin{equation}\label{eq:mds fourth}
        E_{Q_n}\left|\frac{1}{\sqrt{n}}\sum_{i=1}^{n}(2I_i-1)U_{n,i}\right|^4\leq C\left\{(Q_nU_n^2)^2+n^{-1}Q_nU_n^4\right\},
    \end{equation}
    and $$\frac{1}{\sqrt{n}}\sum_{i=1}^{n}(2I_i-1)U_{n,i}\rightsquigarrow \mathcal{N}(0,v)$$
    provided that $v_n =\sum_{s=1}^{m}Q_n(S=s)\var_{Q_n}(V_n\mid S=s)\to v$ and $Q_n\{U_n^2\mathbbm{1}(|U_n|>\epsilon\sqrt{n})\}\to0$ for every $\epsilon>0$. For every $n$, $$E_{Q_n}\left|\frac{1}{\sqrt{n}}\sum_{s=1}^{m}D_n(\chi_s)\bar{V}_{n,s}\right|^4 \leq \frac{m^4C_{q_0}}{n^2}\max_{s\leq m}|\bar{V}_{n,s}|^4\leq \frac{m^4C_{q_0}}{{q_0} n^2}Q_nV_n^4.$$
\end{lemma}

\begin{proof}
    Under CR, the summands $(2I_i-1)V_{n,i}$ are independent, centered conditional on $Q_n$ and have variance $Q_nV_n^2$. The convergence holds by the Lindeberg--Feller Central Limit Theorem for triangular arrays, whose Lindeberg condition is assumed; Equation \eqref{eq:cr fourth} is the fourth-moment formula for sums of independent centered variables, $E|\sum_iX_i|^4=\sum_iEX_i^4+3\sum_{i\neq j}EX_i^2EX_j^2$.

    For DCAR, the decomposition is algebraic. Let $\cF_{n,i-1}$ be the $\sigma$-field generated by $\{(S_{n,j},V_{n,j},I_j):j\leq i-1\}$ and by the design randomizers used through step $i-1$. The pair $(S_{n,i},V_{n,i})$ is independent of $\cF_{n,i-1}$. Iterated conditioning therefore gives
    \begin{align*}
        E[(2I_i-1)U_{n,i}\mid \cF_{n,i-1}] &= \sum_{s}Q_n(S=s)E[2I_i-1\mid \cF_{n,i-1},S_{n,i}=s]E[U_{n,i}\mid S_{n,i}=s]=0,\\
        E[(2I_i-1)^2U_{n,i}^2\mid \cF_{n,i-1}]&=Q_nU_n^2,
    \end{align*}
    since $E[U_{n,i}\mid S_{n,i}=s]=0$ by construction. Hence $\{(2I_i-1)U_{n,i}\}$ is a martingale difference sequence whose normalized predictable variation equals $v_n$ exactly and whose conditional Lindeberg quantity is $Q_n\{U_n^2\mathbbm{1}(|U_n|>\epsilon\sqrt{n})\}\to0$. The martingale central limit theorem \cite[Corollary 3.1]{hall2014martingale} gives the stated normal limit, and Rosenthal's inequality for martingales \cite[Theorem 2.12]{hall2014martingale} gives \eqref{eq:mds fourth}. Finally, H\"older's inequality in the form $|\sum_{s\leq m}D_n(\chi_s)\bar{V}_{n,s}|^4\leq m^3\sum_{s\leq m}|D_n(\chi_s)|^4|\bar{V}_{n,s}|^4$ and Assumption \ref{assume:bootstrap}(2) indicates
     $$E_{Q_n}\left|\frac{1}{\sqrt{n}}\sum_{s=1}^{m}D_n(\chi_s)\bar{V}_{n,s}\right|^4\leq \frac{m^4}{n^2}\max_s|\bar{V}_{n,s}|^4E_{Q_n}\max_s|D_n(\chi_s)|^4 \leq \frac{m^4{C_{q_0}}}{n^2}\max_s|\bar{V}_{n,s}|^4;$$
     Jensen's inequality bounds $\max_s|\bar{V}_{n,s}|^4\leq\max_sQ_n(V_n^4\mid S=s)\leq Q_nV_n^4/{q_0}$.
\end{proof}

Lemma \ref{lemma:signed sums} is important because it proves convergence of summation under bootstrap distribution for DCAR, typically with martingale decomposition. The result for CCAR, however, does not hold using the same technique. To show under CCAR, one generally requires the imbalance Markov Chain's invariant distribution under bootstrap distribution to converge to the true invariant distribution, which is out-of-scope for this discussion. This is also related to the bootstrap conjecture in \cite[Remark 4.4]{ma2024new}.

\begin{lemma}[Bootstrap matrix convergence]\label{lemma:bootstrap matrix}
    Recall the working covariate vector $\vW_j$ on the empirical bootstrap population and recall $e_{n,j}$ to be its empirical product projection residual. For CR or DCAR satisfying Assumption \ref{assume:bootstrap}(2), conditionally on $\cD_n$,
    $$\vQ_n^*-\vQ_n^0=o_{P_*}(1), \quad E_*\left\|\frac{1}{\sqrt{n}}\sum_{i=1}^{n}\vG_i^* e_{n,J_i^*}\right\|^4=O_P(1),\quad \vQ_n^0\overset{P}\to \vOmega^{(a)}\succ 0.$$
\end{lemma}
\begin{proof}
    The lower-right block of $\vQ_n^*$, similar to Lemma \ref{lemma:block matrix}, is the bootstrap average of $\vW\vW^\top$ and therefore converges conditionally to $\hat{P}_n\vW\vW^\top$. For a treatment-covariate cross block, use
    $$\frac{1}{n}\sum_i I_i^*\vW_{J_i^*} = \frac{1}{2n}\sum_i \vW_{J_i^*}+\frac{1}{2n}\sum_i(2I_i^*-1)\vW_{J_i^*}.$$
    The first term converges to $\hat{P}_n \vW/2$ by conditional Chebyshev's inequality, since its conditional variance is bounded by $n^{-1}\hat{P}_n\|\vW\|^2=O_P(n^{-1})$. Under CR, the second term is $O_{P_*}(n^{-1/2})$ by \eqref{eq:cr fourth} applied conditionally with $Q_n=\hat{P}_n$ and using Lemma \ref{lemma:signed sums} with $V_n$ as a component of $\vW$. Under DCAR, work on the event $\min_s\hat P_n(S=s)\geq\tfrac12\min_sP(S=s)$, whose probability tends to one. On this event, apply Lemma \ref{lemma:signed sums} componentwise, conditionally on $\cD_n$ with $Q_n=\hat P_n$.The martingale term is $O_{P_*}(n^{-1/2})$ by \eqref{eq:mds fourth} and the imbalance term is $o_{P_*}(1)$ by Lemma \ref{lemma:signed sums} and Markov's inequality. Taking $V_n\equiv1$, so that $U_{n,i}=0$, we get $n_1^*/n\overset{P_*}\to 1/2$. This proves the first statement. For the second statement, note that $\vG_i^*e_{n,J_i^*}$ contains sums of $\vW_{J_i^*}e_{n,J_i^*}, e_{n,J_i^*}, I_i^*e_{n,J_i^*}, (1-I_i^*)e_{n,J_i^*}$. By Lemma \ref{lemma:empirical product projection}, the unweighted sums of $e_{n,J_i^*}$ and $\vW_{J_i^*}e_{n,J_i^*}$ are conditionally i.i.d and centered, with normalized fourth moments bounded in probability by the independent-sum fourth-moment formula in \eqref{eq:cr fourth} together with \eqref{eq:residual second moments}--\eqref{eq:residual fourth moments}. Use the decomposition
    \begin{align*}
        &I_i^*e_{n,J_i^*} = \frac{1}{2}e_{n,J_i^*}+\frac{1}{2}(2I_i^*-1)e_{n,J_i^*},\\
        &(1-I_i^*)e_{n,J_i^*} = \frac{1}{2}e_{n,J_i^*}-\frac{1}{2}(2I_i^*-1)e_{n,J_i^*}.
    \end{align*}
    The second statement is proved by Lemma \ref{lemma:signed sums} by substituting $V_{n,i}=e_{n,J_i^*}$. All terms are $O_P(1)$ by \eqref{eq:residual second moments}--\eqref{eq:residual fourth moments}. The third statement follows from the Weak Law of Large Numbers applied to $\bar{\vW}$ and $\hat{P}_n\vW\vW^\top$. The limit $\vOmega^{(a)}$ is positive definite by Lemma \ref{lemma:block matrix}. The proof is finished.
\end{proof}

\begin{lemma}\label{lemma:conditional OLS limit}
    Conditionally on the data, $P_*\{\cA_n^c\}\overset{P}\to 0$ and $X_n^*$ converges weakly in probability to $\mathcal{N}\left(0,4(\tau^{(r,a)})^2\right).$
\end{lemma}

\begin{proof}
    For every resampled observation, $Y_{J_i^*} = (\vG_i^*)^\top\vbeta_n^\dag+e_{n,J_i^*}$ by definitions of $\vG_i^*,\vbeta_n^\dag, e_{n,j}$. By the normal equations for LS, valid on $\cA_n$, and Lemma \ref{lemma:bootstrap matrix},
    \begin{align*}
        \sqrt{n}\vL(\hat{\vbeta}_n^*-\vbeta_n^\dag) &= \vL (\vQ_n^0)^{-1}\frac{1}{\sqrt{n}}\sum_{i=1}^{n}\vG_i^* e_{n,J_i^*} + \vL\{(\vQ_n^*)^{-1}-(\vQ_n^0)^{-1}\}\frac{1}{\sqrt{n}}\sum_{i=1}^{n}\vG_i^* e_{n,J_i^*}\\
        &= \frac{2}{\sqrt{n}}\sum_{i=1}^{n}(2I_i^*-1)e_{n,J_i^*} + o_{P_*}(1)\quad \text{on }\cA_n,
    \end{align*}
    where the leading term uses $\vL(\vQ_n^0)^{-1}=(2,-2,0,\ldots,0)$ from Lemma \ref{lemma:empirical product projection}, and the remainder is $o_{P_*}(1)$ because, on $\cA_n$, $\|(\vQ_n^*)^{-1}-(\vQ_n^0)^{-1}\|_{\text{op}}\leq \epsilon_0^{-1}\|\vQ_n^*-\vQ_n^0\|_{\text{op}}\|(\vQ_n^0)^{-1}\|_{\text{op}}=o_{P_*}(1)$ while $n^{-1/2}\sum_{i}\vG_i^*e_{n,J_i^*}=O_{P_*}(1)$, both by Lemma \ref{lemma:bootstrap matrix}.
    Let $M_n^* = 2n^{-1/2}\sum_i(2I_i^*-1)e_{n,J_i^*}$. It follows by Lemma \ref{lemma:signed sums} applied conditionally with $Q_n=\hat{P}_n$, $V_n=e_n$, together with Lemma \ref{lemma:empirical residual moments}, that $E_*|M_n^*|^4$$= O_P(1)$. By Weyl's inequality, $|\lambda_{\min}(\vQ_n^*)-\lambda_{\min}(\vOmega^{(a)})|\leq \|\vQ_n^*-\vOmega^{(a)}\|_{\text{op}}$. Further, since $\epsilon_0<\lambda_{\min}(\vOmega^{(a)})$, we have $\cA_n^c\subseteq \{\|\vQ_n^*-\vOmega^{(a)}\|_{\text{op}}>C_0\}$ with $C_0:=\lambda_{\min}(\vOmega^{(a)})-\epsilon_0>0$. The triangle inequality gives
    \begin{align*}
        P_*\{\cA_n^c\}&\leq P_*\{\|\vQ_n^*-\vOmega^{(a)}\|_{\text{op}}>C_0\}\\
        &\leq P_*\{\|\vQ_n^*-\vQ_n^0\|_{\text{op}}>C_0/2\} + \mathbbm{1}\{\|\vQ_n^0-\vOmega^{(a)}\|_{\text{op}}>C_0/2\}\\
        &\overset{P}\to 0,
    \end{align*}
    where the last convergence holds because $\vQ_n^*-\vQ_n^0=o_{P_*}(1)$ and $\vQ_n^0-\vOmega^{(a)}=o_P(1)$. Hence, the conditional H\"older inequality gives
    $$E_*\left[|{M_n^*}|^2\mathbbm{1}\{\cA_n^c\}\right]\leq \{E_*|{M_n^*}|^4\}^{1/2}\left[P_*\{\cA_n^c\}\right]^{1/2}=o_P(1).$$
    Combined with the expansion on $\cA_n$, this proves $X_n^*=M_n^*+o_{P_*}(1)$ unconditionally, so it suffices to derive the conditional weak limit of $M_n^*$.
    Under CR and given $\cD_n$, $(2I_i^*-1)$ are i.i.d. independent of the resampled indices, and the relevant conditional variance is $\hat{P}_ne_n^2$. Under DCAR, we decompose
    $$\frac{1}{\sqrt{n}}\sum_i (2I_i^*-1)e_{n,J_i^*} = \frac{1}{\sqrt{n}}\sum_i (2I_i^*-1)\{e_{n,J_i^*}-E_{\hat{P}_n}[e_{n,J_i^*}\mid S_{J_i^*}]\}+ \frac{1}{\sqrt{n}}\sum_{s\in[m]} D_n^*(\chi_s)E_{\hat{P}_n}[e_{n}\mid S=s].$$
    The second term is asymptotically negligible by conditional Markov's inequality and Lemma \ref{lemma:signed sums}. The first term is the martingale part of Lemma \ref{lemma:signed sums}, whose conditional variance is $v_n(\hat{P}_n)=\sum_{s}\hat{q}_s\widehat{\var}_n(e_n\mid S=s)$.

    It remains to let the law $Q_n=\hat{P}_n$ be random. By Lemma \ref{lemma:empirical residual moments},
    $$\hat{P}_ne_n^2\overset{P}\to(\tau^{(\tCR,a)})^2,\qquad v_n(\hat{P}_n)\overset{P}\to(\tau^{(\tDC,a)})^2,\qquad \frac{\hat{P}_n|e_n|^4}{\epsilon^2 n}\overset{P}\to 0\ \ \text{for every }\epsilon>0.$$
    Note that $\tfrac{\hat{P}_n|e_n|^4}{\epsilon^2 n}$ dominates the Lindeberg quantities of Lemma \ref{lemma:signed sums}. Given an arbitrary subsequence, extract a further subsequence along which these convergences hold almost surely. Along it, for almost every realization of the data, the deterministic hypotheses of Lemma \ref{lemma:signed sums} are satisfied by $Q_n=\hat{P}_n$ with $V_n=e_n$, so the conditional law of $n^{-1/2}\sum_i(2I_i^*-1)e_{n,J_i^*}$ converges weakly to $\mathcal{N}(0,(\tau^{(r,a)})^2)$ almost surely. The subsequence criterion for convergence in probability yields
    $${M_n^*}\rightsquigarrow_* \mathcal{N}(0,4(\tau^{(r,a)})^2)\quad\text{in probability},$$
    and therefore $X_n^*=M_n^*+o_{P_*}(1)$ has the same conditional weak limit. The lemma is proved.
\end{proof}

\begin{lemma}[Conditional moments]\label{lemma:conditional moments}
    Recall $X_n^*=\sqrt{n}\vL(\hat{\vbeta}_n^*-\vbeta_n^\dag)\mathbbm{1}\{\cA_n\}$. The following holds:
    \begin{align}
        &X_n^*\mid \mathcal{A}_n \rightsquigarrow_* \mathcal{N}\left(0,4(\tau^{(r,a)})^2\right),\label{eq:bootstrap conditional law}\\
        & E_*(|X_n^*|^4\mid \mathcal{A}_n) = O_P(1),\label{eq:bootstrap conditional 4th moment}\\
        &\var_*(X_n^*\mid \mathcal{A}_n)\overset{P}\to 4(\tau^{(r,a)})^2. \label{eq:bootstrap conditional variance}
    \end{align}
\end{lemma}

\begin{proof}
    By Lemma \ref{lemma:conditional OLS limit}, $P_*(\cA_n^c)\overset{P}\to 0$. For any bounded $f$,
    $$\left|E_*\{f(X_n^*)\mid \cA_n\}-E_*f(X_n^*)\right|\leq \frac{2\|f\|_\infty P_*(\cA_n^c)}{P_*(\cA_n)} = o_P(1). $$
    Applying this bound to bounded Lipschitz $f$, the conditional law of $X_n^*$ given $\cA_n$ has the same weak limit in probability as the unconditional law obtained in Lemma \ref{lemma:conditional OLS limit} by Prokhorov's Lemma. On the accepted observations, the exact normal-equation and $\|(\vQ_n^*)^{-1}\|_{\text{op}}\leq \epsilon_0^{-1}$ on $\cA_n$ give $$|X_n^*|\leq \|\vL\|\epsilon_0^{-1}\left\|\frac{1}{\sqrt{n}}\sum_{i=1}^{n}\vG_i^* e_{n,J_i^*}\right\|,$$
    therefore,
    $$E_*(|X_n^*|^4\mid \cA_n) \leq \frac{\|\vL\|^4 \epsilon_0^{-4}}{P_*(\cA_n)}E_*\left\|\frac{1}{\sqrt{n}}\sum_{i=1}^{n}\vG_i^* e_{n,J_i^*}\right\|^4 = O_P(1),$$
    using Lemma \ref{lemma:bootstrap matrix} and $1/P_*(\cA_n)=O_P(1)$. The second statement is proved. To show variance convergence, it suffices to show uniform integrability of the second moment. For $M>0$,
    $$E_*\{(X_n^*)^2\mathbbm{1}\{|X_n^*|>M\}\mid \cA_n\}\leq M^{-2}E_*\{|X_n^*|^4\mid \cA_n\}.$$
    For fixed $M$, the maps $x\to x^2\wedge M^2$ and $x\to(-M)\vee(x\wedge M)$ are bounded and continuous, so \eqref{eq:bootstrap conditional law} gives
    $$E_*\{(X_n^*)^2\wedge M^2\mid \cA_n\}\overset{P}\to E(\zeta^2\wedge M^2),\qquad E_*\{(-M)\vee(X_n^*\wedge M)\mid \cA_n\}\overset{P}\to 0,$$
    where $\zeta\sim\mathcal{N}(0,4(\tau^{(r,a)})^2)$. Letting first $n\to\infty$ and then $M\to\infty$ yields
    $$E_*(X_n^*\mid \cA_n)\overset{P}\to 0, \quad E_*\{(X_n^*)^2\mid \cA_n\}\overset{P}\to 4(\tau^{(r,a)})^2,$$
    so that $\var_*(X_n^*\mid \cA_n)=E_*\{(X_n^*)^2\mid \cA_n\}-\{E_*(X_n^*\mid \cA_n)\}^2\overset{P}\to 4(\tau^{(r,a)})^2$, and Equations \eqref{eq:bootstrap conditional law}-\eqref{eq:bootstrap conditional variance} hold.
\end{proof}

Lemma \ref{lemma:bootstrap variance consistency} shows that the sample variance for bootstrap observations converge in probability to the true variance under the bootstrap resampled distribution. 

\begin{lemma}[Bootstrap consistency]\label{lemma:bootstrap variance consistency}
    Suppose for some fixed $k>2$, $P_*(\mathcal{A}_n) \overset{P}\to 1$, $E_*(|X_n^*|^k\mid \mathcal{A}_n) = O_P(1)$. Then, as $\min\{n,B\}\to\infty$, $$n\{\hat{v}_{*,n,B}-\var_*(\hat{\theta}_1^*-\hat{\theta}_0^*\mid \mathcal{A}_n)\}\overset{P}\to 0.$$
\end{lemma}

\begin{proof}
    Conditionally on $\cD_n$, the accepted observations $K_B$ follow Binomial$(B,P_*(\cA_n))$. This implies $K_B/B\overset{P}\to 1$. Write $\mu_n = E_*(X_n^*\mid \cA_n), m_{2,n} = E_*[(X^*_n)^2\mid \cA_n]$. Then, $m_{2,n} = O_P(1)$ by Assumption. Given $K_B = K$, conditional Chebyshev's inequality yields
    $$E_*\left[\left|\frac{1}{K}\sum_{{j}\leq K}X^*_{n,{j}}-\mu_n\right|^2\mid K_B=K\right]\leq \frac{1}{K}m_{2,n}.$$
    For second moment, briefly denote $\Xi_{n,j} = (X^*_{n,j})^2$ and let $r_0=\min(k/2,2)$, we have
    $$
        E_*\left[\left|\frac{1}{K}\sum_{{j}\leq K}({\Xi_{n,j}}-m_{2,n})\right|^{r_0}\mid K_{B}=K\right]\leq C_{r_0}K^{1-r_0}E_*\{|\Xi_{n,1}-m_{2,n}|^{r_0}\mid \cA_n\} = O_P(K^{1-r_0}),
    $$
    where the inequality is due to the Conditional von Bahr-Esseen inequality. The above two displays show that, as $B\to\infty$ (i.e., $K_B\to\infty$), the first and second empirical moments of the accepted $X_n^*$'s converge to $\mu_n, m_{2,n}$, respectively. Since $X_n^*=\sqrt{n}(\hat{\theta}_1^*-\hat{\theta}_0^*)$ on $\cA_n$, the quantity $n\hat{v}_{*,n,B}$ is the sample variance of the $K_B$ accepted $X_n^*$'s and $n\var_*(\hat{\theta}_1^*-\hat{\theta}_0^*\mid \cA_n)=m_{2,n}-\mu_n^2$. The proof is finished.
\end{proof}

\subsection{Verification of Assumption \ref{assume:bootstrap}(2)}\label{app:verification}
Let $\bm{\Psi}_s$ be the specific value of the discrete feature map in stratum $s$, put $\bm{\Psi}=(\bm{\Psi}_1,\ldots,\bm{\Psi}_m)$ which will be a vector of $0'$s with only one $1$ entry. Define
$$\vM={\bm{\Psi}^\top\bm{\Psi}}.$$
For the vector of within-stratum imbalances
$\vD_n=(D_n(\chi_1),\ldots,D_n(\chi_m))^\top$, the feature imbalance is
$\Lambda_n={\bm{\Psi}} \vD_n$ and its quadratic criterion is
$\|\Lambda_n\|^2=\vD_n^\top \vM \vD_n$.

Lemma \ref{lem:full-rank-dcar} justifies that Assumption \ref{assume:bootstrap}(2) holds for a large class of designs that balances within-stratum imbalances. The proof strategy follows from a Foster-Lyapunov drift argument for Markov chains.

\begin{lemma}
\label{lem:full-rank-dcar}
Suppose $\vM$ is positive definite and the allocation probability
$\rho\in(1/2,1]$ is fixed.  On the arrival of a unit in stratum $s$, assign
the sign $\varsigma=(2I-1)\in\{-1,1\}$ that produces the smaller value of
$\|\Lambda_n+\varsigma\bm{\Psi}_s\|^2$ with probability $\rho$.  Then, for every ${q_0}>0$, there are constants ${t_{q_0}}>0$ and
${C_{q_0}}<\infty$ such that
\begin{equation}
 \sup_{Q:\min_s Q(S=s)\ge{q_0}}\ \sup_{n\ge0}
 E_Q\exp\!\left\{{t_{q_0}}
 (\vD_n^\top \vM \vD_n)^{1/2}\right\}\le {C_{q_0}}.                         \label{eq:exponential-imbalance}
\end{equation}
Consequently, Assumption~\ref{assume:bootstrap}(2) holds.
\end{lemma}

\begin{proof}
Write $R(\vd)=(\vd^\top \vM\vd)^{1/2}$,
$x_s=(\vM\vd)_s$, $\gamma=2\rho-1>0$,
$\lambda_0=\lambda_{\min}(\vM)$, and
$K_0=\max_s\vM_{ss}$. Denote $\ve_s$ the $s$th unit vector, since
$$
 {\|\bm{\Psi}(\vd+\ve_s)\|^2-\|\bm{\Psi}(\vd-\ve_s)\|^2=4x_s},
$$
the allocation rule gives
$$
 E({\varsigma}\mid \vD_n=\vd,S_{n+1}=s)=-\gamma\operatorname{sgn}(x_s).
$$  
Note that $\sqrt{u+v}-\sqrt{n}\leq v/(2\sqrt{u})$. For $R(\vd)>0$, 
$R(\vd+\varsigma\ve_s)^2=R(\vd)^2+2\varsigma x_s+\vM_{ss}$ by definition. Take $u=R(\vd)^2, v=2\varsigma x_s+\vM_{ss}$, substitute into the inequality and take expectation on both sides, we get 
\begin{equation*}
 E\{R(\vD_{n+1})-R(\vd)\mid \vd,S_{n+1}=s\}
 \le\frac{-2\gamma|x_s|+\vM_{ss}}{2R(\vd)}.
\end{equation*}
Because $q_s=Q(S=s)\ge{q_0}$, then
\begin{align*}
 \sum_sq_s|x_s|
 &\ge{q_0}\|\vM\vd\|_1
 \ge{q_0}\|\vM\vd\|_2
 \ge{q_0}\sqrt{\lambda_0}R(\vd).
\end{align*}
It follows that
\begin{equation}
 E\{R(\vD_{n+1})-R(\vd)\mid \vD_n=\vd\}
 \le-\gamma{q_0}\sqrt{\lambda_0}+\frac{K_0}{2R(\vd)}.                  \label{eq:radial-drift}
\end{equation}
Thus, the conditional drift is at most
$-\gamma{q_0}\sqrt{\lambda_0}/2$ outside a fixed ball.

The reverse triangle inequality gives
\begin{equation}\label{eq:delta Rd}
 |R(\vd+{\varsigma}\ve_s)-R(\vd)|\le\sqrt{\vM_{ss}}\le \sqrt{K_0}.
\end{equation}
Choose ${t_{q_0}}>0$ sufficiently small. Note that $\exp(u)\leq 1+u+u^2\exp(|u|)/2$. Take $u={t_{q_0}} [R(\vD_{n+1})-R(\vd)]$, together with \eqref{eq:radial-drift} yield,
$$\exp\left\{{t_{q_0}} [R(\vD_{n+1})-R(\vd)]\right\}\leq 1+{t_{q_0}} [R(\vD_{n+1})-R(\vd)]+\frac{t^2K_0}{2}\exp({t_{q_0}} \sqrt{K_0}).$$
Outside the ball $\{\vd:R(\vd)\leq K_0/(\gamma{q_0}\sqrt{\lambda_0})\}$, the following holds by choosing ${t_{q_0}}$ sufficiently small:
$$
 E\left[
  \exp\{{t_{q_0}}(R(\vD_{n+1})-R(\vd))\}\mid \vD_n=\vd
 \right]\le 1-{t_{q_0}}\frac{\gamma{q_0}\sqrt{\lambda_0}}{4}.
$$
Take $V(\vd)=\exp\{{t_{q_0}} R(\vd)\}$, observe that outside the ball,
$$
 E\{V(\vD_{n+1})\mid \vD_n=\vd\}
 \le{\lambda_{q_0}} V(\vd).
$$
Inside the ball $\{\vd:R(\vd)\leq K_0/(\gamma{q_0}\sqrt{\lambda_0})\}$, because of Equation \eqref{eq:delta Rd},
$$V(\vD_{n+1})\leq \exp\left\{{t_{q_0}}\left(\frac{K_0}{\gamma{q_0}\sqrt{\lambda_0}}+\sqrt{K_0}\right)\right\}.$$
Combining the two scenarios,
$$E\{V(\vD_{n+1})\mid \vD_n=\vd\}\leq {\lambda_{q_0}} V(\vd)+\exp\left\{{t_{q_0}}\left(\frac{K_0}{\gamma{q_0}\sqrt{\lambda_0}}+\sqrt{K_0}\right)\right\}\mathbbm{1}\left\{R(\vd)\leq \frac{K_0}{\gamma{q_0}\sqrt{\lambda_0}}\right\}.$$
Take expectation on both sides, then iterate from $\vD_0=\vZero$ with geometric sums, we prove the first claim \eqref{eq:exponential-imbalance}.

Finally, because $R(\vD_n)^2 \geq \lambda_0\|\vD_n\|^2_2$, so 
$$\max_s|\vD_n(\chi_s)|\leq \|\vD_n\|_2\leq \frac{R(\vD_n)}{\sqrt{\lambda_0}}.$$
Take $4$th power on both sides, Assumption \ref{assume:bootstrap}(2) holds. The proof is finished. 
\end{proof}

\begin{prop}[Verification for commonly used designs]
\label{prop:named-designs}
Suppose the design strata are the analysis strata.  Assumption
\ref{assume:bootstrap} holds for:
(1) STR-PB with a fixed finite block size;
(2) STR-BC with $\rho \in(1/2,1]$;
(3) HH with $\omega_{s}>0$ for each $s\in[m]$.
\end{prop}

\begin{proof}
For STR-PB, each completed block has zero imbalance, and the imbalance in the single incomplete block is bounded by the maximum block size.  The required bound therefore holds by design.

For STR-BC, take $\bm{\Psi}_s=\ve_s$ the $s$th basis vector, so that $\vM=\bI_m$, the $m$-by-$m$ identity matrix. Lemma \ref{lem:full-rank-dcar} applies.

For HH method, let $u_{j\ell}\in\{0,1\}^m$ indicate which joint strata have
level $\ell$ of margin $j$. Recall Example \ref{example: HuHu}. On the state $\vD_n$, its quadratic
criterion has matrix
\begin{equation*}
 \vM=\omega_0{\bone}{\bone}^\top
   +\sum_{j,\ell}{\omega_{j\ell}}u_{j\ell}u_{j\ell}^\top
   +\operatorname{diag}(\omega_{1},\ldots,
                         \omega_{m}).                  
\end{equation*}
If $\omega_{s}>0$ for each stratum $s$, then
$M\succeq\min_s\omega_s I_m$ is positive definite. Lemma~\ref{lem:full-rank-dcar}
applies and thus the proof is finished. 
\end{proof}

\section{Additional numerical studies}\label{app:additional sim}

\subsection{Numerical studies for variance inflation}
Recall the discussion on CCAR variance inflation in Example \ref{ex:gaussian-inflation} in Section \ref{sec:instability}. Although it is clear that CCAR with scalar, non-constant $\psi$ may result in inflated variability, the impact on corresponding hypothesis test is not directly assessed. Here, we present both theoretically and numerically the inflation exists.

\begin{examplenonumber}[Example \ref{ex:gaussian-inflation} continued]
    Consider the data generating model $Z_i\overset{d}=\xi$, $$Y_i=g(Z_i)+\varepsilon_i, \quad g(z) = C \sqrt{\frac{\pi}{2}}\operatorname{sgn}(z)-z$$ 
    where constant $C\in\{1,2\}, \varepsilon_i\overset{d}=\xi$ with no treatment effect. The treatment rule uses scalar CCAR with $\psi(Z)=Z$. The theoretical, unadjusted test size is therefore
    $$2\Phi\left[-\Phi^{-1}(0.975)\sqrt{\frac{1+C^2(\frac{\pi}{2}-1)}{1+C^2\frac{\pi}{2}(\frac{\pi}{2}-1)}}\right].$$
    We also conduct Monte-Carlo simulations with different $n$, repeated $10^4$ times each. The results are listed in Table \ref{tab:gaussian-inflation}. There are significant size inflations both theoretically and numerically.
    \begin{table}[t!]
        \centering
        \begin{tabular}{l|c|ccc}
        \toprule
             & Theoretical size & $n=500$ & $n=2000$& $n=8000$\\\midrule
         $C=1$  & 7.45 & 7.44 & 7.58 & 7.34\\
         $C=2$ & 9.73 & 10.22 & 9.90 & 9.64\\\bottomrule
        \end{tabular}
        \caption{The theoretical and simulated type I error of $T_n$ under Example \ref{ex:gaussian-inflation}. The units are $10^{-2}$.}
        \label{tab:gaussian-inflation}
    \end{table}
\end{examplenonumber}

The variance inflation phenomenon under non-scalar $\psi$ does not admit closed form evaluation. Hence, in this section, we complement the discussion by providing another empirical example showing CCAR need not out perform CR. The consistency of the Monte-Carlo variance estimator is justified in Appendix \ref{app:monte carlo}. 

\begin{example}\label{example:instability}
    Consider $({Z_{i,1},Z_{i,2}})_{i=1}^{n}$ following $\text{Gumbel}(0,2),\text{Laplace}(-1,1)$ which are independent. Suppose the two variables are discretized at $-1$ and $1$, resulting in $9$ strata. After the assignments, we compute the within-stratum imbalances $D_n^{(r)}(\stind{s}) = \sum_{i=1}^{n}(2I_i-1)\mathbbm{1}\{S_i=s\}$ for $s\in[9]$. Since the explicit form of ${\pi_{\Lambda}(\zeta_{\stind{s}})}$ is hard to obtain for CCAR, we propose to consistently estimate its limiting variance by Monte-Carlo simulation, see Appendix \ref{app:monte carlo}. In Table \ref{tab:instability}, we report the theoretical standard deviations for CR and Monte-Carlo standard deviations for CCAR of the first $6$ strata. It is shown that at least the variability of the $5$th stratum under CCAR may \textit{significantly} exceed that under CR in this scenario.

\begin{table}[t!]
    \centering
    \begin{tabular}{l|cccccc}
    \toprule
    Imbalance &$D_n^{(r)}(\stind{1})$ &$D_n^{(r)}(\stind{2})$&$D_n^{(r)}(\stind{3})$&$D_n^{(r)}(\stind{4})$&$D_n^{(r)}(\stind{5})$&$D_n^{(r)}(\stind{6})$ \\
    \midrule
    {$\sigma^{(\tCR)}_{\stind{s}}$}& 0.31&0.28&0.11&0.42&0.39&0.15 \\
    $\hat{\sigma}^{(\text{COV})}_{\stind{s}}(\text{se})$ &0.31 (0.005)&0.35 (0.005)& 0.11 (0.002)&0.48 (0.008)&\textbf{0.63} (0.009)&0.17 (0.002)\\
    \bottomrule
    \end{tabular}
    \caption{(Estimated) Standard deviations of the within-strata imbalance (and its Monte-Carlo standard error). For CCAR, we consider COV method with weights $(1,2,1)$. The experiment is repeated $2000$ times with sample size $n=2000$.}
    \label{tab:instability}
    \end{table}
\end{example}

\subsection{Numerical studies for empirical power curves}
Simulation results on the adjusted empirical power of the Wald-test under different cases (Case 1,2) studied in Section \ref{sec:numerical} are provided in Figures \ref{fig:case 12}, \ref{fig:case 36}. 

\begin{figure}[t!]
    \centering
    
    \begin{subfigure}[b]{\textwidth}
        \centering
        \includegraphics[width=\linewidth]{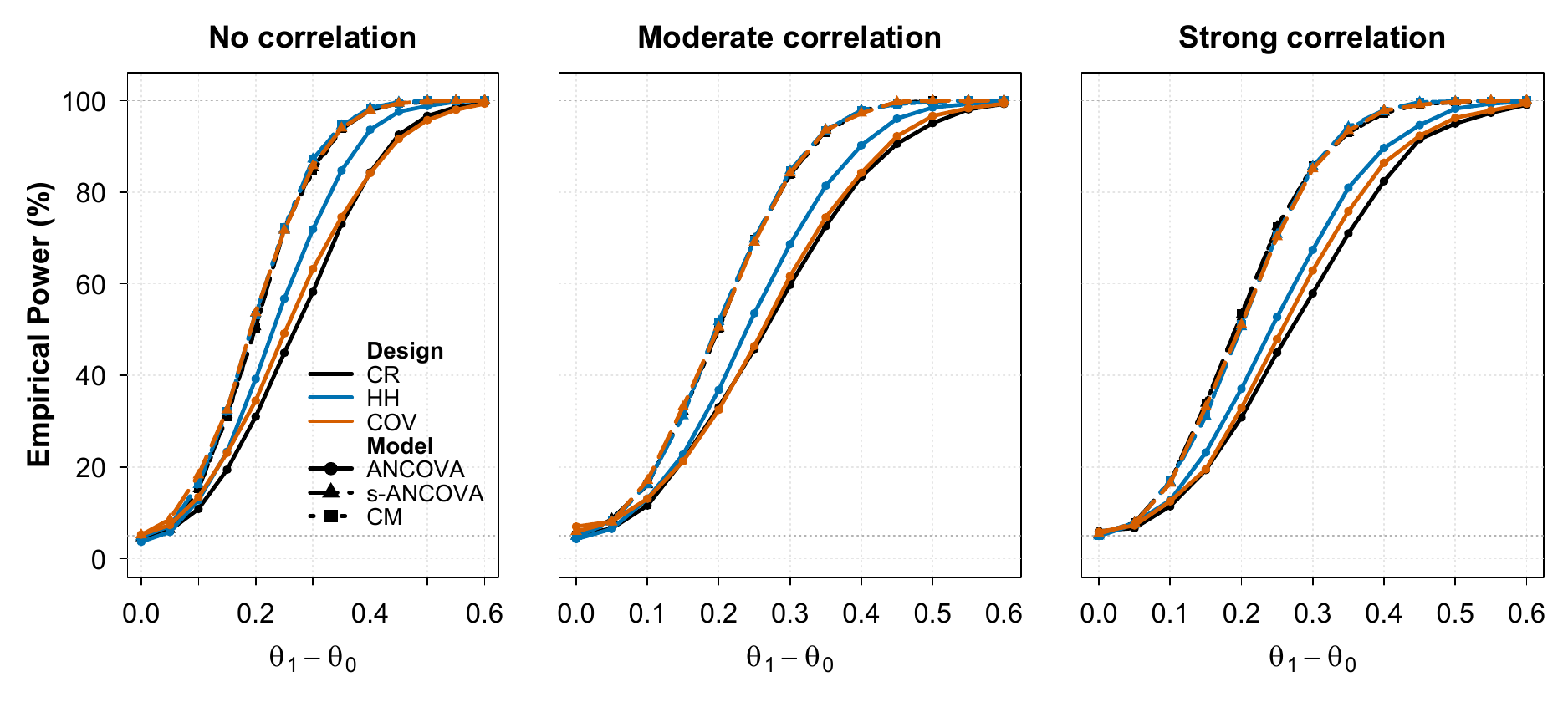}
        
    \end{subfigure}
    
    \bigskip
    
    \begin{subfigure}[b]{\textwidth}
        \centering
        \includegraphics[width=\linewidth]{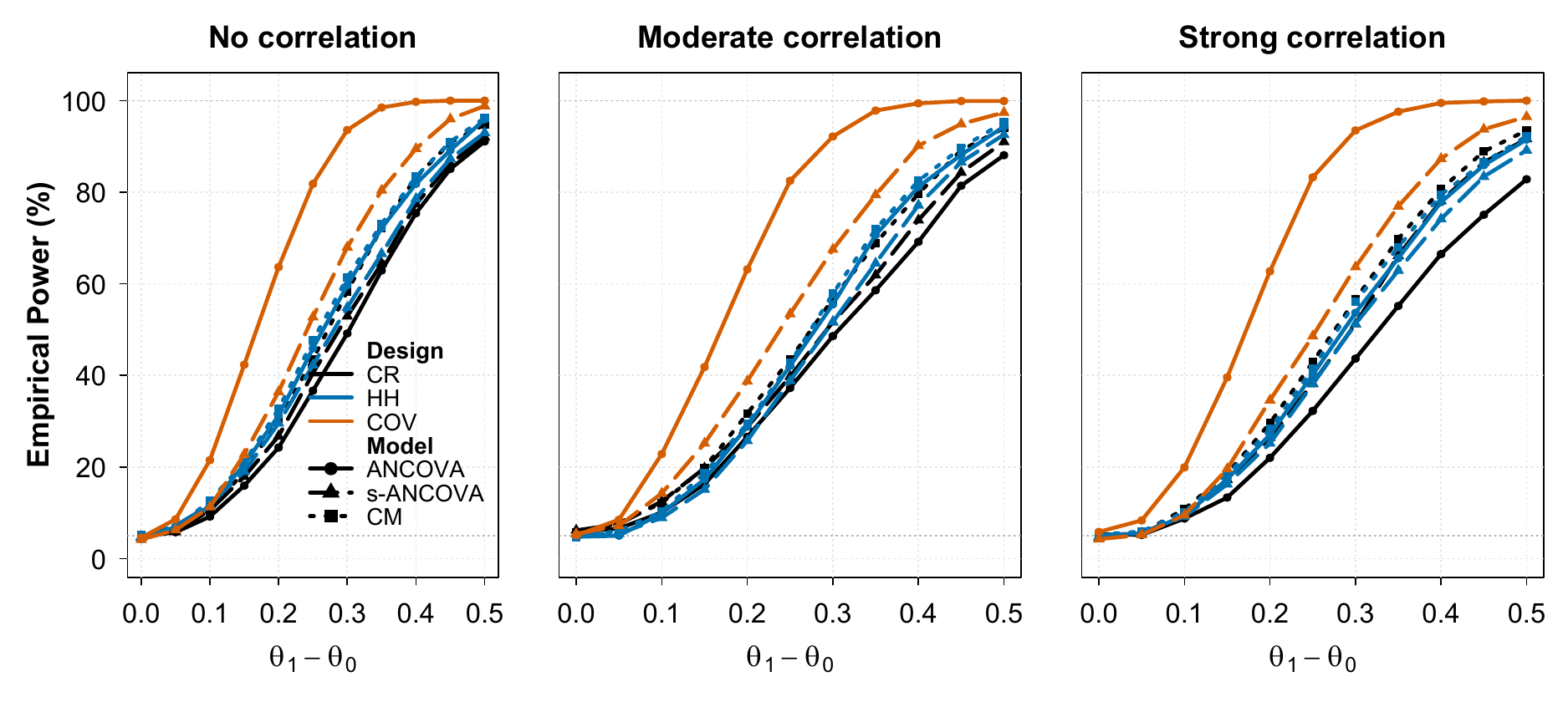}
        
    \end{subfigure}

    \caption{Bootstrap adjusted empirical power curves for \textbf{Case} 1 (Top) and \textbf{Case} 2 (Bottom). }
    \label{fig:case 12}
\end{figure}

\begin{figure}[t!]
    \centering
    
    \begin{subfigure}[b]{\textwidth}
        \centering
        \includegraphics[width=\linewidth]{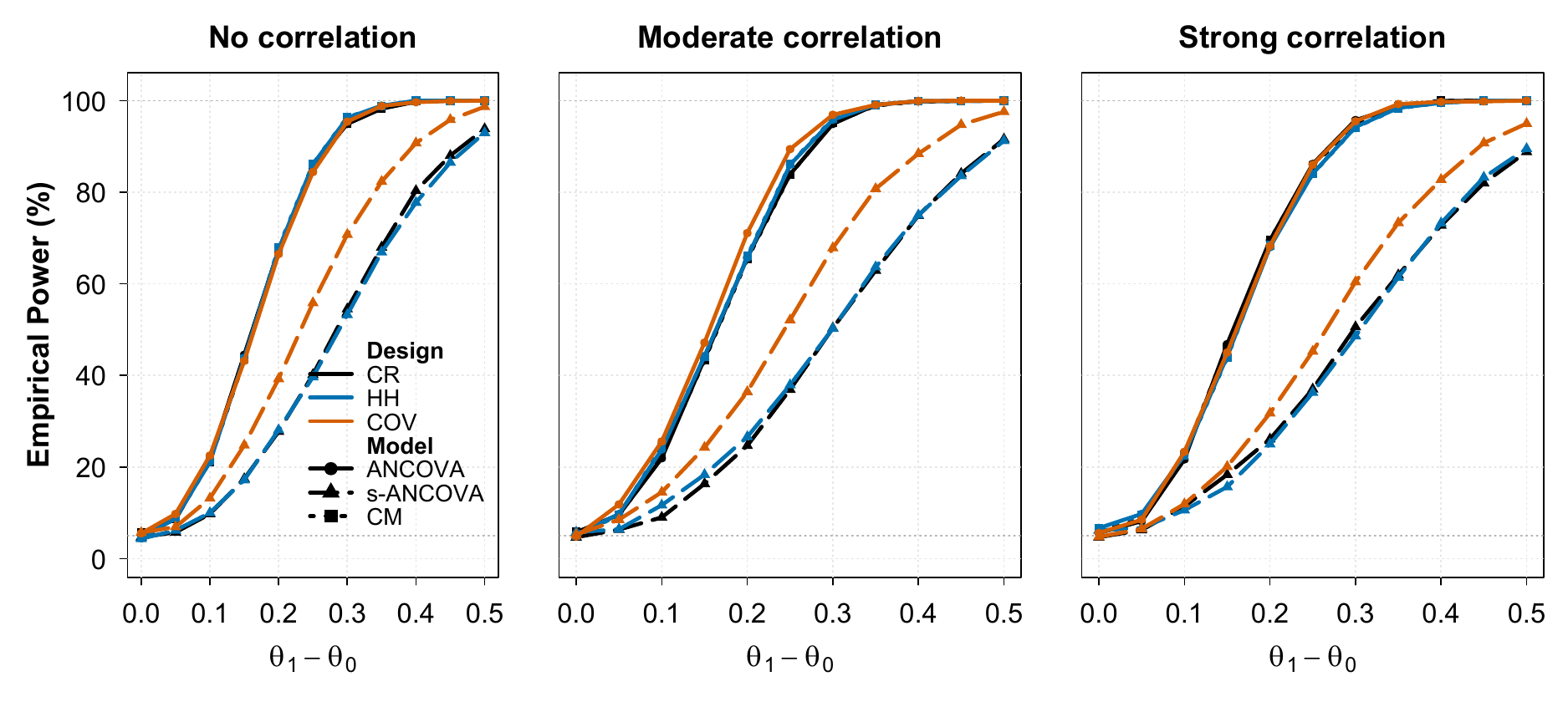}
        
    \end{subfigure}
    
    \bigskip
    
    \begin{subfigure}[b]{\textwidth}
        \centering
        \includegraphics[width=\linewidth]{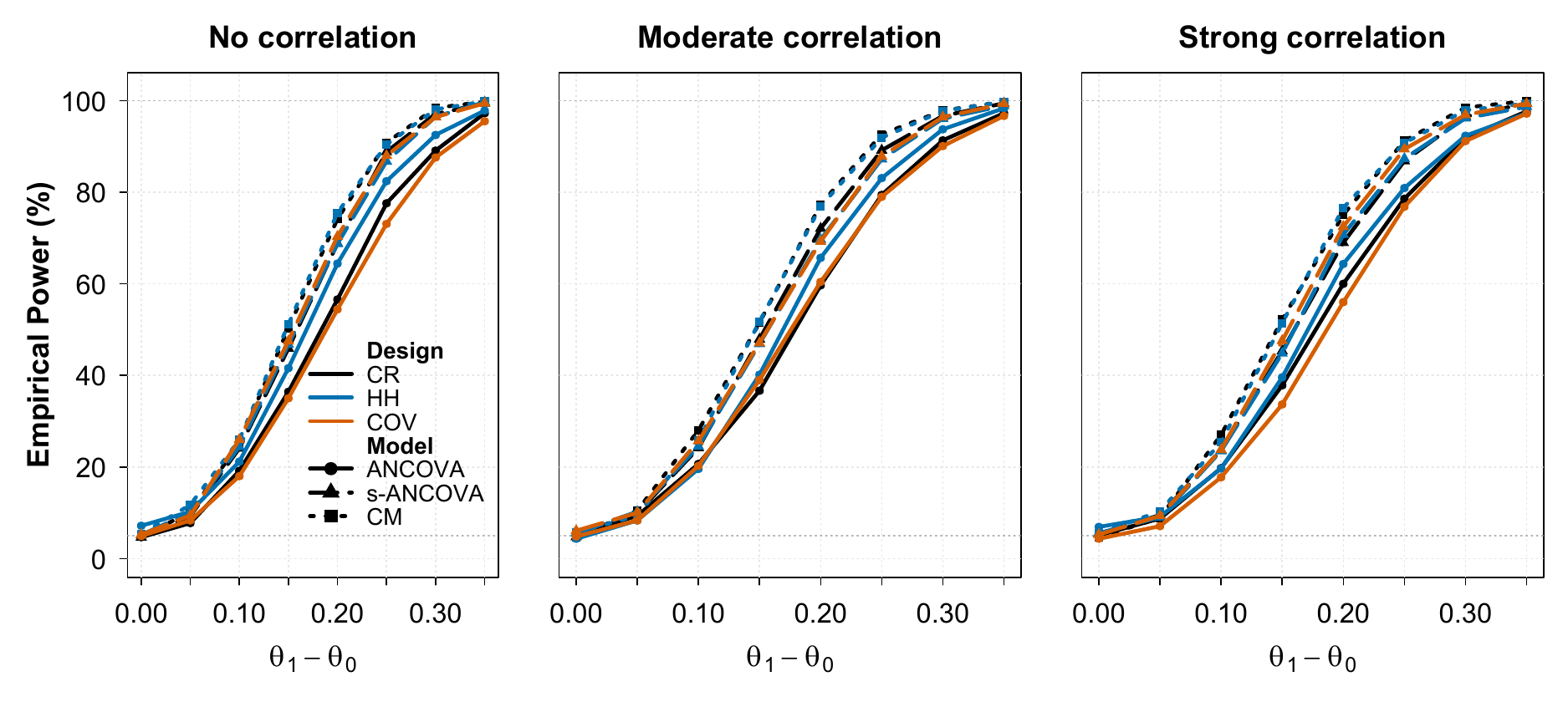}
        
    \end{subfigure}

    \caption{Bootstrap adjusted empirical power curves for \textbf{Case} 5 (Top) and \textbf{Case} 6 (Bottom). }
    \label{fig:case 36}
\end{figure}

In addition to Cases $3,4$ presented in the main paper, Figure \ref{fig:case 45 fewer} considers the scenario where only a small to moderate sample of subjects is recruited and the analysis may not support the fitting of all variables (equivalently, $d_Z<p, d_S<m-1$). The major conclusions are similar. Importantly, when not all variables are taken into account, DCAR achieves higher statistical power than CR consistently. This observation again corroborates Theorem \ref{theorem:compare tests}.

\begin{figure}[t!]
    \centering
    
    \begin{subfigure}[b]{\textwidth}
        \centering
        \includegraphics[width=\linewidth]{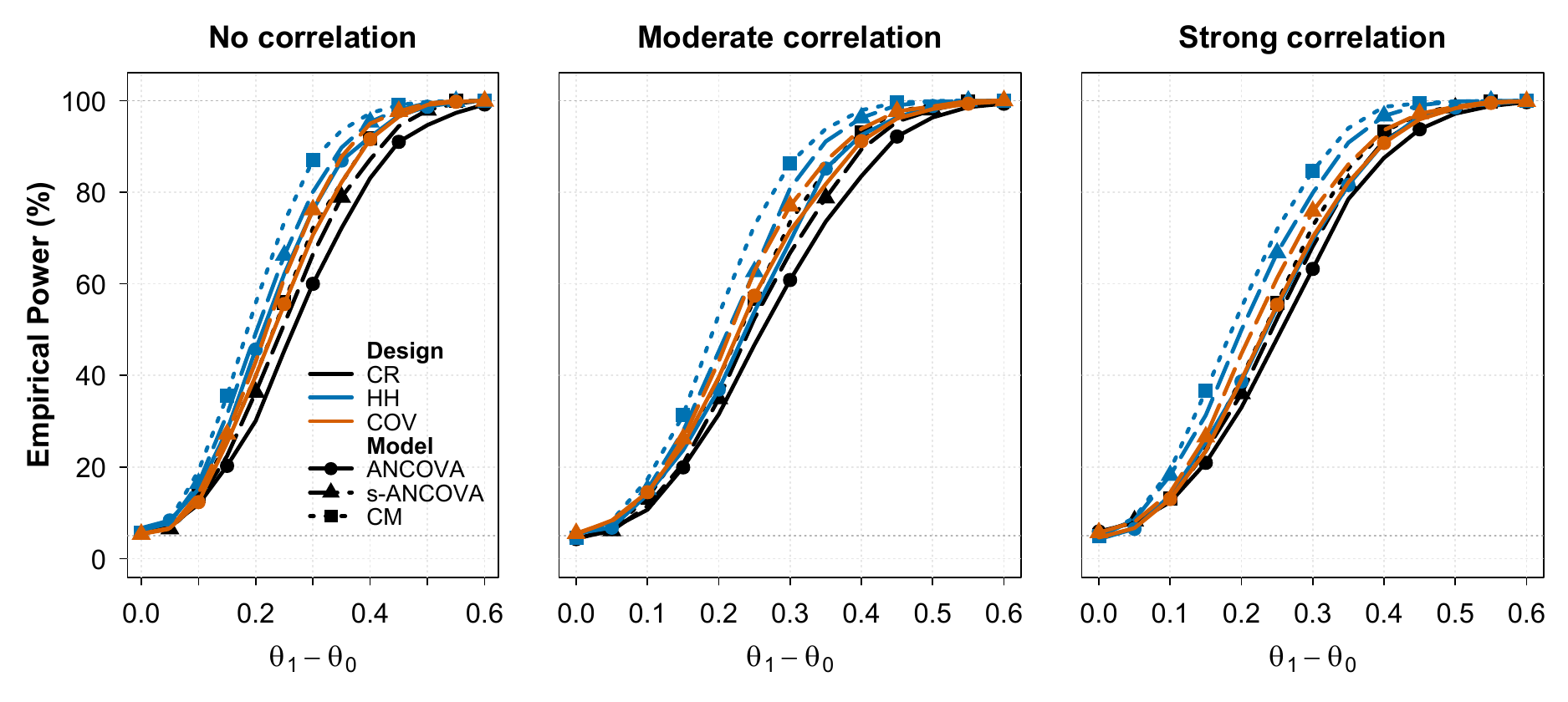}
    \end{subfigure}
    
    \bigskip
    
    \begin{subfigure}[b]{\textwidth}
        \centering
        \includegraphics[width=\linewidth]{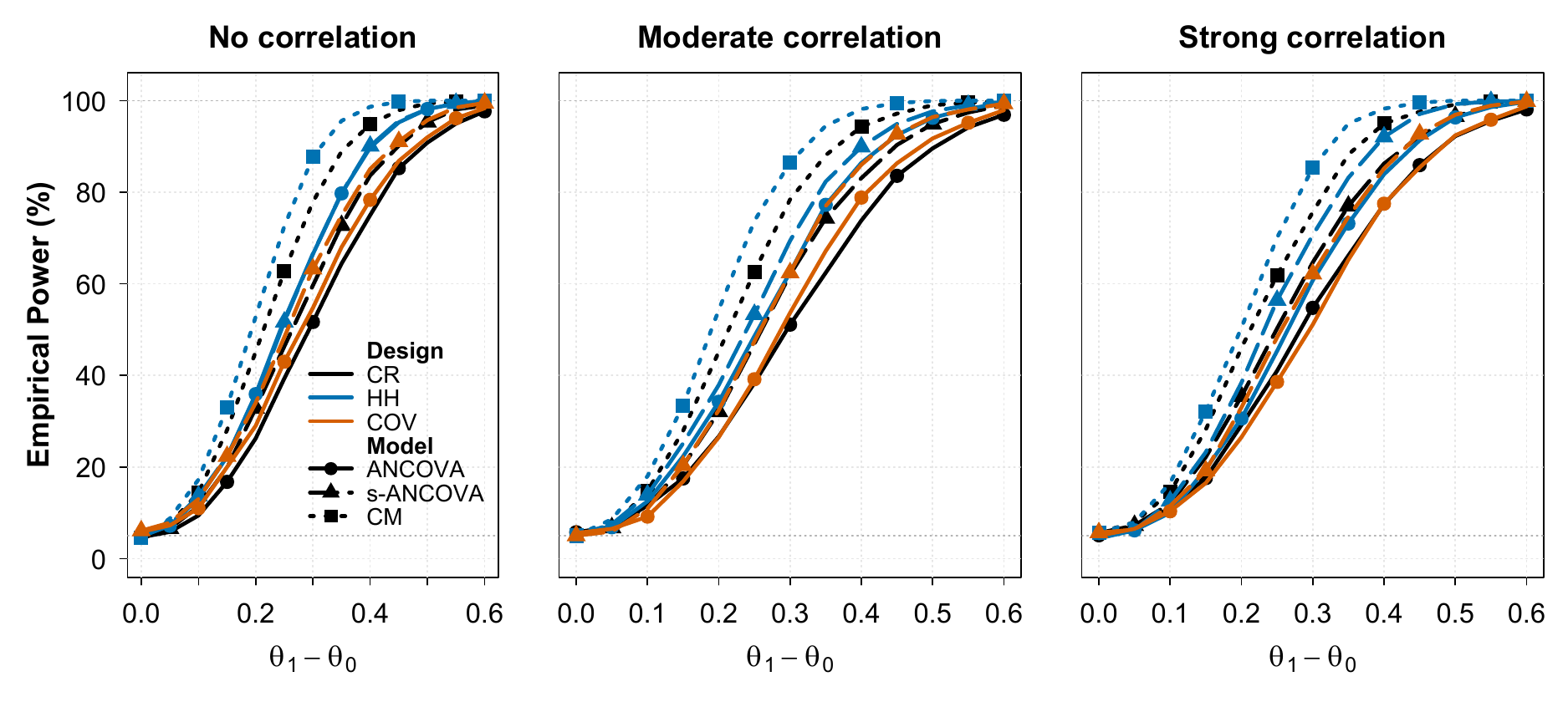}
        
    \end{subfigure}

    \caption{Bootstrap adjusted empirical power curves for \textbf{Case} 3 (Top) and \textbf{Case} 4 (Bottom) with \textbf{fewer covariates} fitted ($d_Z=1, d_S=3$ for Case 3; $d_Z=d_S=1$ for Case 4). }
    \label{fig:case 45 fewer}
\end{figure}

Next, we study two additional non-additive, non-linear Case $5,6$ to supplement discussions in Section \ref{sec:numerical}. 

\noindent\textbf{Case 5.} (Monotone smooth model, $n=600$) Consider three continuous covariates where the first is normal $\mathcal{N}(0,1)$, the second is $\text{Exp}(1)-1$ and the last is Uniform$(-3/2,3/2)$. The variables are discretized at their marginal medians $0, \log(2)-1,0$, respectively, so that $S:\mathbb{R}^3\to\{1,\ldots,8\}$. The response-covariate relationship $g(\vz) = 1.4z_1 + 0.8\log\{1+\exp[2(z_2+1)]\}+0.9z_3$. For ANCOVA model, $d_Z=3$; for SFE model, $d_S=7$; for CM model, $(d_Z,d_S) = (3,7)$.

\noindent \textbf{Case 6.} (Non-linear non-additive, $n=800$) Consider two continuous covariates $({Z_{i,1}},{Z_{i,2}})$ following a two dimensional multivariate normal distribution $\mathcal{N}_2(\vZero,\vB_{\rho_Z})$. Each ${Z_{i,j}}$ is discretized at $1/3,2/3$ quantiles so that $S:\mathbb{R}^2\to\{1,\ldots,9\}$. Recall that the marginal bin function for $Z_{i,2}$ is $\bin_2:\bR\to \{1,2,3\}$. The response-covariate relationship is $g(\vz) = 1.5[\bin_2(z_2)-2] + 1.5\text{tanh}(z_1)+1.5[\bin_2(z_2)-2]\text{tanh}(z_1z_2)$.

The power table for Case $6$ is provided in Table \ref{fig:case 6}. In this non-additive non-linear case, $a=\aS$ generally performs better than $a=\aZ$, while $a=\aCM$ outperforms both for fixed $d_Z=p,d_S=m-1$. In terms of designs, DCAR and CR are the better choices, achieving the same statistical power because all strata are fitted. The rest conclusions are similar to Case $4$.

\begin{figure}
    \centering
    \includegraphics[width=\linewidth]{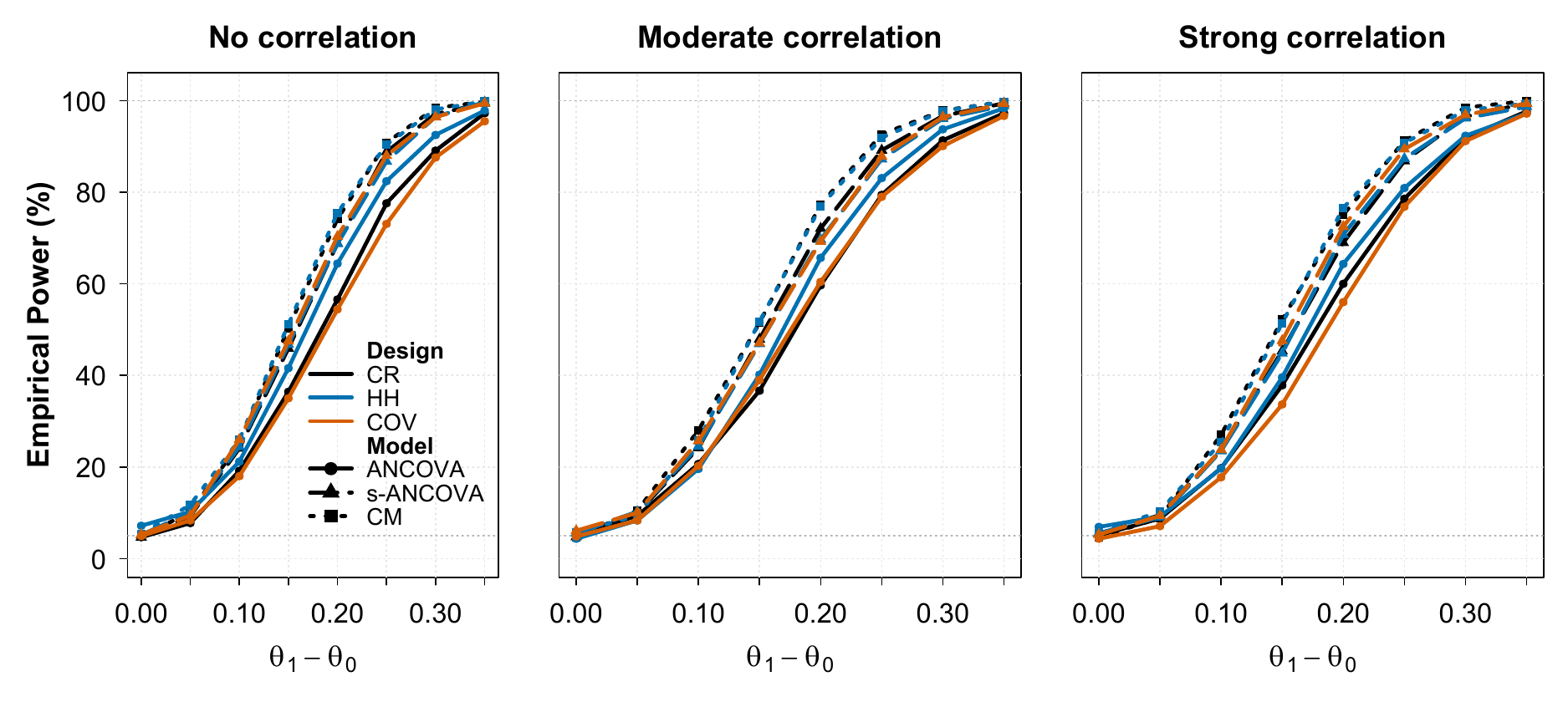}
    \caption{Bootstrap adjusted empirical power curves for \textbf{Case} 6.}
    \label{fig:case 6}
\end{figure}

\subsection{Empirical gains from the CM model}\label{app:empirical gain from CM}
In order to quantify the gain from using CM model, we simulate $n=1000$ using the assignments via the STR-BC method with $\rho=0.8$ and consider two continuous covariates $({Z_{i,1}},{Z_{i,2}})$ following a two dimensional multivariate normal distribution $\mathcal{N}_2(\vZero,\vB_{\rho_Z})$. Both variables are discretized marginally at $-0.6,0.6$, so that $S:\mathbb{R}^2\to\{1,\ldots,9\}$. The response-covariate relationship $g(\vz) = 1.5 [\bin_2(z_2)-1] + c\{\text{tanh}(z_1)+[\bin_2(z_2)-1]\text{tanh}(z_1z_2)\}$. By construction $\vGamma,\bgamt$ are functions of $c$. Figure \ref{fig:gain from CM} exhibits the variance reduction comparing $a=\aCM$ with $a=\aZ$ or $a=\aS$. The gain is amplified by increasing correlation in CM vs ANCOVA because $\vGamma-\bgamt$ is sensitive to $\rho_Z$; the gain for CM vs SFE is $0$ at $\bgamt=\vZero$ and increases gradually as the magnitude of $\bgamt$ increases via $c$.

\begin{figure}
    \centering
    \includegraphics[width=0.9\linewidth]{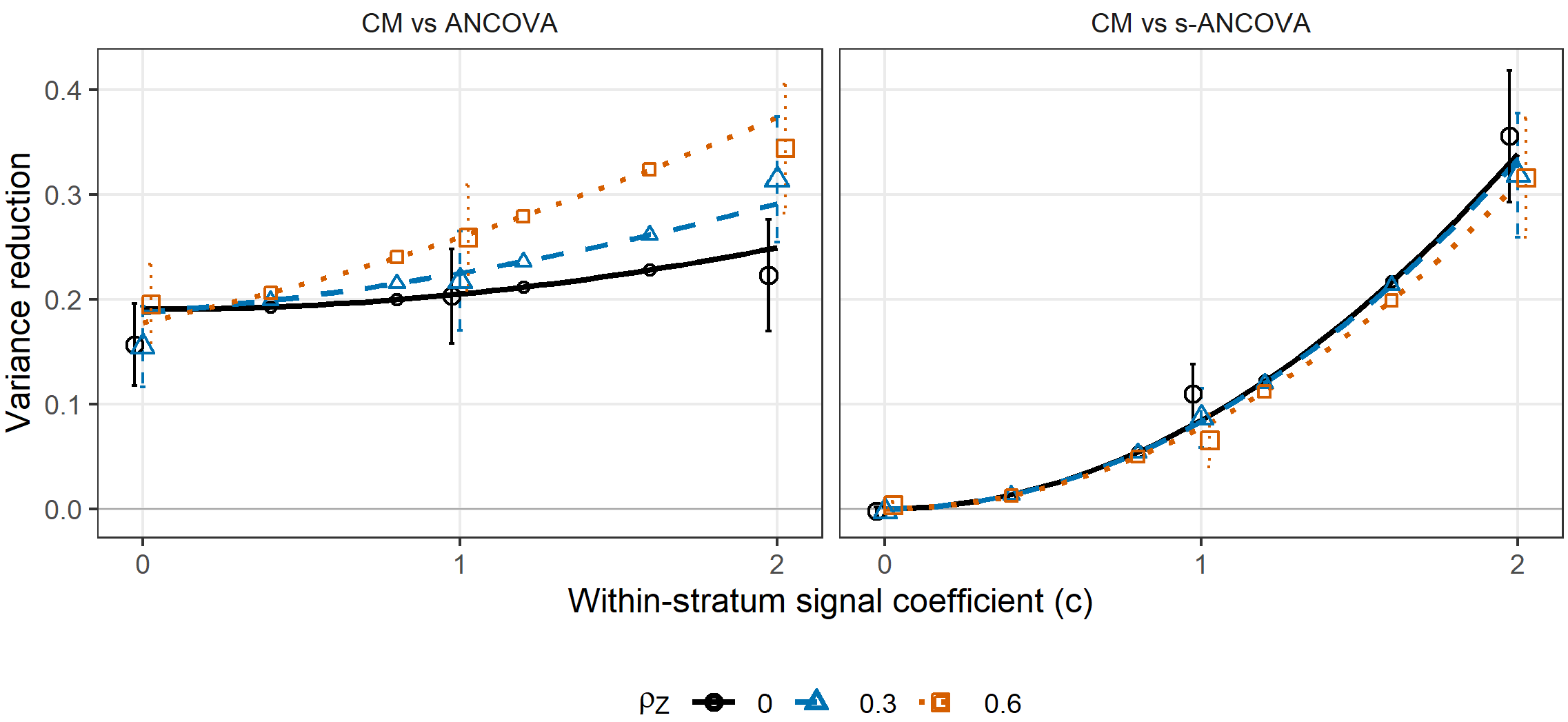}
    \caption{The markers represent variance reduction $(\sigma^{({\text{STR-BC}},a)})^2-(\sigma^{({\text{STR-BC}},\aCM)})^2$ as a function of correlation $\rho_Z$ and $\vGamma(c), \bgamt(c)$, approximated by $10^7$ observations and the vertical bars are $\pm 1.96$ Monte-Carlo Standard Error. The connected lines with markers are theoretical values per Proposition \ref{prop:better model}.}
    \label{fig:gain from CM}
\end{figure}

\subsection{Consistent Monte-Carlo variance estimator}\label{app:monte carlo}
We recall the setup of Example \ref{example:instability}. In this section, we state the Monte-Carlo variance estimator and prove its consistency. For the $k$th repetition, $({Z^{(k)}_{i,1}})_{i=1}^{n},({Z^{(k)}_{i,2}})_{i=1}^{n}$ are generated from $\text{Gumbel}(0,2)$ and $\text{Laplace}(-1,1)$. The assignment vector $(I_i^{(k),(\text{COV})})_{i=1}^{n}$ are then obtained via the COV$(1,2,1)$ procedure in Example \ref{example: COV} and 
$D_n^{(k),(\text{COV})}({\stind{s}}) = \sum_{i=1}^{n}(2I_i^{(k),(\text{COV})}-1)\mathbbm{1}\{S_i^{(k)}=s\},$
are calculated. After $K$ repetitions, the variance estimator is given by $$(\hat{\sigma}_{{\stind{s}}}^{(\text{COV})})^2 = \frac{1}{K-1}\sum_{k=1}^{K}\left(\frac{D_n^{(k),(\text{COV})}({\stind{s}})}{\sqrt{n}}-\hat{\mu}_n^{K,(\text{COV})}\right)^2,$$
where $\hat{\mu}_n^{K,(\text{COV})}= n^{-1/2}K^{-1}\sum_{k=1}^{K}D_n^{(k),(\text{COV})}({\stind{s}})$ is the empirical mean. It is shown in Lemma \ref{lemma:boot} that $(\hat{\sigma}_{{\stind{s}}}^{(\text{COV})})^2=(\sigma_{{\stind{s}}}^{(\text{COV})})^2+o_P(1)$ as $\min(n,K)\to\infty$. 

\begin{lemma}\label{lemma:boot}
    Consider the Monte-Carlo simulation setup in Subsection \ref{example:instability}. As $\min(n,K)\to\infty$, we have $(\hat{\sigma}_{{\stind{s}}}^{(\text{COV})})^2=(\sigma_{{\stind{s}}}^{(\text{COV})})^2+o_P(1)$ for each $s\in[m]$.
\end{lemma}
\begin{proof}[Proof of Lemma \ref{lemma:boot}]
    Fix stratum $s\in[\prod_jm_j]$, note the following decomposition
    $$(\hat{\sigma}_{{\stind{s}}}^{(\text{COV})})^2-(\sigma_{{\stind{s}}}^{(\text{COV})})^2= (\hat{\sigma}_{{\stind{s}}}^{(\text{COV})})^2- \var[D_n^*({\stind{s}})/\sqrt{n}]+ \var[D_n^*({\stind{s}})/\sqrt{n}]-(\sigma_{{\stind{s}}}^{(\text{COV})})^2,$$
    where $D_n^{*}({\stind{s}})$ is the generic random variable having the same distribution as $D_n^{(k)}({\stind{s}})$ for each $k\in[K]$. 
    
    Suppose for the original data $\{\vZ_i,I_i:i\in[n]\}$ has the same distribution as the generic vector $(\vZ,I)$. Since we generate $\vZ_i^{(k)}, I_i^{(k)}$ from $(\vZ,I)$ directly, $n^{-1/2}D_n^*({\stind{s}})\rightsquigarrow \Theta_s$ with $\Theta_s$ following $\mathcal{N}(0,(\sigma_{{\stind{s}}}^{(\text{COV})})^2)$ as defined in Theorem \ref{theorem: design property}. The first difference is $o_P(1)$ due to weak law of large numbers and the second difference is $o(1)$ as $n\to\infty$. The statement is proved.
\end{proof}

\begin{remark}
    The Monte-Carlo approach considered is similar to the bootstrap approach in \cite{zhao2024estimation}, where they estimate the limiting variance of the within-stratum imbalances under the Minimization method because the limiting variance does not have closed form as well. In \cite{zhao2024estimation}, $\vZ_i$ is discrete, the distribution of $\vZ_i$ is unknown and is estimated by the empirical CDF first. Then for each repetition, $\vZ_i^{(k)}$ is generated via the empirical CDF. In our case, since the distribution for $\vZ$ is assumed to be known, this is a standard Monte-Carlo approach and the proofs will be greatly simplified.
\end{remark}

\section{Further discussions}
\subsection{Discussion on MixCAR}\label{app:discussions on MixCAR}
In practice, it is often the case that $\vZ$ is mixed. In this section, we consider $\vZ = (\vZ_D^\top,\vZ_C^\top)^\top= ({Z_1,\ldots, Z_{p_d},Z_{p_d+1},\ldots, Z_{p}})$, that is, the first $p_d$ entries of $\vZ$ are discrete and the last $p-p_d$ entries are continuous random variables. There are subtle differences in the theoretical results under MixCAR. First, the discretization function is a two-step mapping 
$$S:{\bigtimes_{j=1}^{p_d}}[m_j]\times \bR^{p-p_d}\to {\bigtimes_{j=1}^{p}}[m_j]\to [m].$$
The last $p-p_d$ entries are discretized and merged into the first $p_d$ variables. On the other hand, let
$${S^{\text{dis}}}:{\bigtimes_{j=1}^{p_d}}[m_j]\to [m]$$
be the function that maps $p_d$ margins into strata for the discrete variables only. The following modified assumption (condition) are assumed for MixCAR designs:

\begin{assumption}\label{assume:independent copy2}
    The covariates $\vZ_i$ are independent copies of the generic vector $\vZ$ such that $E(\vZ_C)=\vZero$, and ${\bm{\Sigma}_{\widetilde{Z}_C}} = E\{\var(\vZ_C\mid {S})\}$ is positive definite.
\end{assumption}
\begin{condition}\label{cond:dgp MixCAR}
The overall imbalance $D_n^{(\tMC)} = o_P(\sqrt{n})$. Additionally, $\stind{k}$ and $f(Z_{p_d+1},\ldots, Z_p) = (Z_{p_d+1},\ldots, Z_p)^\top$ will be strongly balanced for any $k\in[m]$.
\end{condition}
A typical example of MixCAR is given below:
\begin{example}[MIX procedure \cite{liu2025}]\label{example: mixed}
    Consider a procedure which uses a mixed covariate map $\psi(\vZ_i) = (\sqrt{\omega_0},\sqrt{\omega_1}\stind{1}(\vZ_i),...,\sqrt{\omega_m}\stind{m}(\vZ_i), \sqrt{\omega_{1}^*}\vZ_i^\top, \sqrt{\omega_{2}^*}{\operatorname{vec}}(\vZ_i\vZ_i^\top)^\top)^\top,$
    with $\omega_0,\omega_1,\ldots,\omega_m,\omega_1^*,\omega_2^*\geq 0$ as weights. The procedure is an MixCAR.
\end{example}
MIX design \ref{example: mixed} satisfy Condition \ref{cond:dgp MixCAR} if $\omega_k>0$ for $k=1,\ldots,m$. Under Assumption \ref{assume:independent copy2} Theorem \ref{theorem: design property} still holds with $({\sigma_f^{(\tMC)}})^2=E[f^2(\vZ)]+2{\pi_{\Lambda^{(\tMC)}}(\zeta_f)}$. The proof is identical to the CCAR case in Theorem \ref{theorem: design property}. For statistical inference, \eqref{eq:true model} is assumed. However, due to the mixed nature of $\vZ_i$, the treatment effect is estimated via the following working model
$$E(Y_i\mid I_i,\vZ_i) = \vG_{i,d_a}^{(\tMC,a)}\vbeta_{d_a}^{(a)}$$
whose fitted covariates and parameters are given by
\begin{align*}
    &\vG_{i,d_Z}^{(\tMC,\aZ)} = (I_i^{(\tMC)},1-I_i^{(\tMC)},\mathbbm{1}\{S^{\text{dis}}_{i}=1\},\ldots, \mathbbm{1}\{S^{\text{dis}}_{i}=d_{Z1}\},{Z_{i,p_d+1},\ldots,Z_{i,p_d+d_{Z2}}}),\\
    &\vG_{i,d_S}^{(\tMC,\aS)} = (I_i^{(\tMC)},1-I_i^{(\tMC)},\mathbbm{1}\{S_{i}=1\},\ldots, \mathbbm{1}\{S_{i}=d_{S}\}),
\end{align*}
where $d_{Z1}\leq m-1, d_{Z2}\leq p-p_d, d_S\leq m-1, d_Z=d_{Z1}+d_{Z2}$. The ANCOVA model under MixCAR fits covariates used in the randomization design unchanged (with both discrete and continuous covariates), while the SFE Model fits all \textbf{discretized} covariates. The construction of the test statistic is analogous to \eqref{eq:test statistic}. Corollary \ref{cor: test under Rand} still follows with 
$(\tau^{(\tMC,a)})^2 = \sigma^2_{\varepsilon}+(\sigma^{(\tMC,a)})^2$,
\begin{align*}
    &(\sigma^{(\tMC,\aZ)})^2 = E[g^2(\vZ)]+2{\pi_{\Lambda^{(\tMC)}}(\zeta_g)},\\
    &(\sigma^{(\tMC,\aS)})^2 = E[g_{d_S}^2(\vZ)]+2{\pi_{\Lambda^{(\tMC)}}(\zeta_{g_{d_S}})}.
\end{align*}
where $\zeta_h$ is defined in the proof of Lemma \ref{lemma:find var}, $g,g_{d_S}$ are defined similarly as in the proof of Theorem \ref{theorem:combined}, $\pi_{\Lambda^{(\tMC)}}$ is the invariant distribution induced by the chain $(\Lambda^{(\tMC)}_i)_{i\geq 1}$. The proof strategy is similar to the CCAR case. The bootstrap adjustment is also identical to Section \ref{app:bootstrap}. We note that as long as continuous covariates are used in the design (e.g., CCAR, MixCAR), the variance inflation issue in Subsection \ref{example:instability} may exist and the design may give rise to test inflation. Therefore, the covariates are suggested to be discretized fully during the design for robustness.

\begin{remark}\label{remark:remark on MixCAR partially known}
    One special case is when $g$ is partially known. Then, according to discussion in \ref{subsec: other functions}, the known components should be strongly balanced in the design, while the unknown components should be suitably discretized. Then, in the analysis, a combined model \eqref{eq:combined working model} should be implemented for the covariates related to the unknown components. Since the discussion will be repetitive, we do not explicitly derive the results here.
\end{remark}

\subsection{Discussion on Condition \ref{assume: dgp}}\label{app:discuss on dgp}

    Our theoretical results relies on Condition \ref{assume: dgp} and can apply to any design that satisfies it. A commonly used case is STR-PB which does not follow assignment \textbf{Step 1.-4.} but is a DCAR that satisfies Condition \ref{assume: dgp}. On the contrary, PS follows \textbf{Step 1.-4.} but does not satisfy Condition \ref{assume: dgp} because the within-stratum imbalances converge slower than $o_P(\sqrt{n})$ \cite[Theorem 1]{zhao2024estimation}. CCARs that balance covariate means or Mahalanobis distance only (with $\psi_1\equiv 0$) \cite[Example (2.1),(2.2)]{ma2024new} need not satisfy Condition \ref{assume: dgp}, as the overall imbalance converges slower than $o_P(\sqrt{n})$. In the literature, this can be overcome via the pairwise sequential approach \cite{qin2022adaptive} or simply by setting $\psi_1>0$.

\subsection{Discussion on additional covariates}\label{sec:additional covariates}
In practice, the true model might also depend on covariates not used in randomization. Mathematically speaking, the true model might be
$$Y_i = \theta_1I_i+\theta_0(1-I_i)+g(\vZ_i)+h(\vX_i)+\varepsilon_i$$
where $\vX_i\in \bR^{p_x}$ are additional covariates not used in randomization and $h:\bR^{p_x}\to \bR$. Usually, $h$ is assumed to be a linear function of $\vX$, i.e., suppose
$$Y_i = \theta_1I_i+\theta_0(1-I_i)+g(\vZ_i)+\sum_{j=1}^{p_x}\iota_j X_{i,j}+\varepsilon_i.$$
The additional variables introduce variability to the estimation of treatment effect. Specifically, if no $\vX_i$ are fitted in the working models \eqref{eq:working model 1}, then the limiting variance in Theorem \ref{theorem:combined} for the new model will be $(\tilde{\sigma}^{(r,a)})^2+\sigma^2_{\varepsilon}+E[\var(\sum_j \iota_j X_j\mid \vZ)]$ where $\tilde{\sigma}^{(r,a)}$ are the same as $\sigma^{(r,a)}$ in Equations \eqref{eq:varCR}-\eqref{eq:varsACC} in Appendix \ref{app:proof inference} with $\tilde{g}(\vZ) = g(\vZ) + E(\sum_j \iota_j X_j\mid \vZ)$ in places of $g(\vZ)$. On the other hand, if all $\vX_i$'s are added to working models \eqref{eq:working model 1}, then no additional variance will be introduced. This case is a generalization of the results in \cite{ma2015testing}, in which they only consider linear relationships in $\vX,\vZ$ and a subset of DCAR. It is also possible to consider $h$ to be non-linear, in any case, the inclusion and exclusion of $\vX$ will not affect our recommendations on the discretization strategy of $\vZ$. Since such extension is relatively straightforward following our discussion, we do not go deeper in this paper for simplicity. We also want to mention that additional covariates $\vX$ is a crucial building block for efficient estimation under model-robust approach. The utilization of $\vX$ is mainly for more efficient estimation and has readily been discussed in a series of works \citep{ye2021inference,ye2023toward,bannick2023general,tu2024unified,ma2022regression,liu2023lasso}. Detailed discussion on model-robust approach is in Section \ref{sec:model robust}. 

\subsection{Discussion on the real-data}\label{app:counter}

Table \ref{tab:realdata_trial6} contains an entry that appears paradoxical.
Under the COV design, the stratum-adjusted test based on all twelve strata is
\emph{less} powerful ($88.6\%$) than the same test with the adiponectin stratum omitted ($94.7\%$). Standard practice suggests that adding prognostic baseline covariates to a working model reduces the residual
variance, and the added stratum should not make the precision worse. The purpose of this appendix is to explain why the principle fails once the CAR design uses the
covariates, and to exhibit a minimal example in which every quantity can be computed exactly.

The question can be formulated as follows. Fix a randomization rule $r$ and two nested working models. Both models yield consistent estimators of
$\theta_1-\theta_0$, and the larger model has the smaller population residual. Then, is the larger model guaranteed to be asymptotically no less efficient? 
Recall from the proofs of our inference results in Appendix \ref{app:proof inference} that, for a design $r$ and a
working model $a$ with population residual $e^{(a)}$, the estimator
admits the expansion
$$
\sqrt{n}\bigl(\hat\theta_1-\hat\theta_0-(\theta_1-\theta_0)\bigr)
=\frac{2}{\sqrt n}\sum_{i=1}^n (2I_i-1)\bigl\{e^{(a)}_i-2\nu_1^{(a)}\bigr\}
+o_P(1).
$$
By Theorem \ref{theorem: design property}, the general imbalance measure $n^{-1/2}D_n^{(r)}(f)$ converge in distribution with limiting variance $(\sigma^{(r)}_f)^2$. Define the design-induced norm
$$
\|f\|_{(r)}^2=(\sigma^{(r)}_f)^2.
$$
Then, Theorem \ref{theorem:combined} shows that the asymptotic variance of the estimator is essentially $$4\{\sigma_{\varepsilon}^2+\|e^{(a)}\|_{(r)}^2\}.$$ In contrast, the consensus of `no worse models when more covariates are fitted' is defined via $$R^2= 1-\frac{\|e^{(a)}\|^2}{\text{Total sum of squares}}$$ which is a monotone function of $\|e^{(a)}\|^2$. The observation originates from difference between these two norms:
\begin{center}
\emph{model information orders by $\|e^{(a)}\|$} $\overset{?}=$
\emph{efficiency orders by $\|e^{(a)}\|_{(r)}$.}
\end{center}
Under CR there is no distinction between the two, since $I_i$ is i.i.d. centered and $\|f\|_{(\tCR)}=\|f\|$ for every
centered $f$. Under CAR the seminorm is anisotropic, in the sense that directions that are strongly balanced are removed, while directions `orthogonal' to the balanced features are not balanced, see discussions in Section \ref{sec:instability}. Therefore, adding more strata (covariates) \textbf{reduces the $L^2$ norm} of the residual, but not the design induced semi-norm of the residual. The latter determines the estimator's precision. 

We provide a theoretical explanation for the observation with the following covariate-mean balancing CCAR procedure as a concrete example.
\begin{condition}\label{cond:app ideal}
The design has feature map $\psi(z)=(1,z)^{\top}$ and satisfies, for equal
allocation: (i) $n^{-1/2}D_n^{(r)}(f)\overset{P}\to 0$ for $f\in\{1,z\}$;
(ii) $n^{-1/2}D_n^{(r)}(f)\rightsquigarrow N(0,E[f^2(Z)]+2\pi_{\Lambda}(\zeta_f))$ for every fixed $f$ with $E[f^2(Z)]<\infty$ not in the span of $\psi$.
\end{condition}

Part (i) is strong balance of the features; part (ii) is essentially satisfied by using Theorem \ref{theorem: design property}.

\begin{example}
 Let $Z_i\sim N(0,1)$ and
$Y_i=\theta_1I_i+\theta_0(1-I_i)+\beta Z_i+\varepsilon_i$ with
$\varepsilon_i$ independent of $Z_i$, $E[\varepsilon_i]=0$ and
$\mathrm{Var}(\varepsilon_i)=\sigma_{\varepsilon}^2$. For $m\ge2$ let
$S_i\in[m]$ be the equiprobable discretization of $Z_i$, with cut points
$\Phi(c_s)=s/m$. Consider the following two working models fitted by least squares:
(1) the unadjusted model $a=\varnothing$, regressing $Y$ on $(I,1-I)$, (2) $a=\aS$, additionally fitting $\stind{1},\ldots,\stind{m-1}$. By some calculations, it can be shown that
$$
v_m=E\bigl[\mathrm{Var}(Z\mid S)\bigr]
=1-m\sum_{s=1}^{m}\bigl\{\phi(c_{s-1})-\phi(c_s)\bigr\}^2,
$$
where $\phi(\cdot)$ is the density of standard normal distribution. The asymptotic variances of $\sqrt n\{\hat\theta_1-\hat\theta_0-(\theta_1-\theta_0)\}$ are
$$
\begin{array}{l@{\qquad}l@{\qquad}l}
 & a=\varnothing & a=\aS\\[2pt]
\tCR: & 4(\sigma_{\varepsilon}^2+\|e^{(\phi)}\|^2) = 4\bigl(\sigma_{\varepsilon}^2+\beta^2\bigr)
      &4(\sigma_{\varepsilon}^2+\|e^{(\aS)}\|^2)=  4\bigl(\sigma_{\varepsilon}^2+\beta^2 v_m\bigr)\\[2pt]
\text{CCAR}: & 4(\sigma_{\varepsilon}^2+\|e^{(\phi)}\|^2_{(\tCC)}) = 4\sigma_{\varepsilon}^2
      & 4(\sigma_{\varepsilon}^2+\|e^{(\aS)}\|^2_{(\tCC)})=4\bigl(\sigma_{\varepsilon}^2+\beta^2 v_m+\pi_{\Lambda}(\zeta_{e})\bigr),
\end{array}
$$
where $\zeta_e$ is defined as in Equation \eqref{eq:pi_Lambda g} with $e = \beta(Z-E[Z\mid S])$ being the residual of the working model. In particular, if
$$4\beta^2 v_m+\pi_{\Lambda}(\zeta_e)>0,$$
then adjustment by strata hurts precision under CCAR. It has been shown in, e.g., Section \ref{sec:instability} Example \ref{ex:gaussian-inflation} that $\pi_{\Lambda}(\zeta_e)$ can be positive to make the above inequality to hold. The observation is corroborated with a simulation study in Table \ref{tab:paradox-simulation}.

\begin{table}[t]
\centering
\caption{Paradox simulation with $n=400$, two equiprobable strata, and $5000$ repetitions. Power is oracle variance-calibrated at $\tau=0.36$. CCAR used: covariate mean balancing CAR; DCAR used: stratified biased-coin method.}
\label{tab:paradox-simulation}
\begin{tabular}{llrrrr}
\toprule
Design & Analysis & $\|e\|^2$ & $\|e\|_{(r)}^2$ & $n\operatorname{Var}(\hat\theta)$ & Power (\%) \\
\midrule
CR & Unadjusted & 1.000 & 1.009 & 7.829  &73.3 \\
CR & Strata & 0.364 & 0.361 & 5.451 & 87.2 \\
\addlinespace
CCAR & Unadjusted & 1.000 & 0.008 & 4.113 & \textbf{94.4} \\
CCAR & Strata & 0.364 & 0.260 & 5.119 & \textbf{88.8} \\
\addlinespace
DCAR & Unadjusted & 1.000 & 0.365 & 5.418   & 87.6 \\
DCAR & Strata & 0.364 & 0.360 & 5.394 & 87.7 \\
\bottomrule
\end{tabular}
\end{table}
\end{example}

\section{Extension to multi-arm trials}
\label{app:multiarm}

This appendix gives a rigorous equal-allocation extension for CR, DCAR, CCAR procedures satisfying the balance condition below. In multi-arm setup, CR refers to the randomization procedure such that the assignment vectors $\bT_i$ are i.i.d., independent of $\{(\vZ_i,\varepsilon_i):i\geq1\}$, and satisfy $P(\bT_i=\be_k)=1/A$, $A\geq 3$. 

\subsection{Multi-arm setup}
\label{app:multiarm-setup}
Let the number of arms $A\geq2$ be fixed as $n\to\infty$, let $\bT_i=(T_i^{(1)},\ldots,T_i^{(A)})^\top$ be the one-hot assignment vector, and put $\vc=A^{-1}\bone$. Write $\be_k$ for the $k$th standard basis vector and
\begin{equation}
\label{eq:ma-projection}
\bP_A=\bI_A-A^{-1}\bone\bone^\top.
\end{equation}
For a scalar $f$, define
$$\vD_n^{(r)}(f)=\sum_{i=1}^n(\bT_i-\vc)f(\vZ_i).$$
If $\psi$ is $p'$-dimensional, the full feature imbalance is the $A\times p'$ matrix with $k$th row $\sum_i(T_i^{(k)}-A^{-1})\psi(\vZ_i)^\top$. When $A=2$ and $\bell=(1,-1)^\top$, $\bell^\top\vD_n^{(r)}(f)$ is exactly the binary imbalance in \eqref{eq:general imbalance}.

\begin{condition}[Multi-arm balance]
\label{cond:ma-balance}
The strata have fixed probabilities $q_s>0$. Let $\mathcal F_{i-1}$ denote the assignment history and define $\bpi_i=P(\bT_i=\be_k\mid\mathcal F_{i-1},S_i)$. Conditional on $(\mathcal F_{i-1},S_i)$, the assignment draw is independent of the remainder of $\vZ_i$ and of $\varepsilon_i$; the equal allocation target is imposed through the following balance condition, not by requiring $\bpi_i=\vc$ pointwise. For every $s\in[m]$, $\vD_n^{(\tDC)}(\chi_s)=o_P(\sqrt n).$ For CCAR, $\vD_n^{(\tCC)}(f) = o_P(\sqrt{n})$, where $f(\vz) = z_j$.
\end{condition}

\subsection{Multi-arm covariate imbalance}
\label{app:multiarm-imbalance}

\begin{lemma}
\label{lem:ma-pi}
For a probability vector $\bpi=(\pi^{(1)},\ldots,\pi^{(A)})^\top$, define
\[
\bPi(\bpi)=\sum_{k=1}^A\pi^{(k)}(\be_k-\vc)(\be_k-\vc)^\top.
\]
Then
\begin{equation}
\label{eq:ma-pi}
\bPi(\bpi)=\operatorname{diag}(\bpi)-A^{-1}(\bpi\bone^\top+\bone\bpi^\top)+A^{-2}\bone\bone^\top,
\end{equation}
so $\bPi$ is affine and $\bPi(\vc)=A^{-1}\bP_A$.
\end{lemma}

\begin{proof}
Expand $(\be_k-\vc)(\be_k-\vc)^\top$, multiply by $\pi^{(k)}$, and sum over $k$. Substitution of $\bpi=\vc$ gives the final identity.
\end{proof}

\begin{prop}[Multi-arm covariate imbalance]
\label{prop:ma-imbalance}
Suppose $E\{f^2(\vZ)\}<\infty$. Under CR,
\[
n^{-1/2}\vD_n^{(\tCR)}(f)\rightsquigarrow
\mathcal N_A\!\left(\vZero,A^{-1}E\{f^2(\vZ)\}\bP_A\right).
\]
Under a DCAR satisfying Condition \ref{cond:ma-balance},
\[
n^{-1/2}\vD_n^{(\tDC)}(f)\rightsquigarrow
\mathcal N_A\!\left(\vZero,A^{-1}E\{\var(f(\vZ)\mid S)\}\bP_A\right).
\]
Under a CCAR satisfying Condition \ref{cond:ma-balance},
\[
n^{-1/2}\vD_n^{(\tDC)}(f)\rightsquigarrow
\mathcal N_A\!\left(\vZero,A^{-1}(\sigma_f^{(\tCC)})^2\bP_A\right),
\]
where $\sigma_f^{(\tCC)}$ may not have a closed form.
\end{prop}

\begin{proof}
Under independent uniform randomization the summands are i.i.d., centered, and have covariance $A^{-1}E(f^2)\bP_A$ by Lemma \ref{lem:ma-pi}; the multivariate CLT applies.

For DCAR, put $U_i=f(\vZ_i)-E(f\mid S_i)$. The stratum decomposition gives
\[
\vD_n^{(\tDC)}(f)=\sum_{i=1}^n(\bT_i-\vc)U_i+
\sum_{s=1}^m\vD_n^{(\tDC)}(\chi_s)E(f\mid S=s),
\]
and the second term is $o_P(\sqrt n)$. Let $\mathcal H_{2i-1}$ contain the past and the newly revealed $S_i$, and let $\mathcal H_{2i}$ additionally contain the randomization draw and the remainder of $\vZ_i$. With $\bpi_i=P(\bT_i=\be_k\mid\mathcal H_{2i-1})_{k\leq A}$, the first sum is a martingale and its predictable covariance increment is
\[
E\{[(\bT_i-\vc)U_i]^{\otimes2}\mid\mathcal H_{2i-1}\}
=\bPi(\bpi_i)\var(f\mid S_i).
\]
For every $s,k$, $\sum_{i:S_i=s}(T_i^{(k)}-\pi_i^{(k)})=O_P(\sqrt n)$. Together with $N_{s,k}-N_s/A=o_P(\sqrt n)$, $N_s/n\to q_s$, and the affineness of $\bPi$, this yields
\[
\frac1n\sum_{i=1}^n\bPi(\bpi_i)\var(f\mid S_i)
\overset P\to A^{-1}E\{\var(f\mid S)\}\bP_A.
\]
Finally, $\|\bT_i-\vc\|\leq1$, and the conditional Lindeberg remainder is bounded in expectation by $E[U^2\mathbbm1\{|U|>\epsilon\sqrt n\}]\to0$. The martingale CLT and Cram\'er--Wold device complete the proof for DCAR. For CCAR, the results can be found directly in \cite[Theorem 4.3]{zhang2026asymptoticpropertiesmultitreatmentcovariate}.
\end{proof}

\subsection{Inference for multi-arm CR and DCAR}
\label{app:multiarm-inference}
Let
\[
Y_i=\sum_{k=1}^A\theta_kT_i^{(k)}+g(\vZ_i)+\varepsilon_i,
\qquad
\vG_i^{(r,a)}=(\bT_i^\top,\bW_i^{(a)\top})^\top,
\]
where $a\in\{\aZ,\aS,\aCM\}$, $E\{g^2(\vZ)\}<\infty$, $\var(\bW^{(a)})\succ0$, and the errors obey the independence, mean-zero, and $0<\sigma_\varepsilon^2<\infty$ conditions in \eqref{eq:true model}. The population projection includes an intercept; explicitly,
\[
\bbetaW^{(a)}=\var(\bW^{(a)})^{-1}\cov\{\bW^{(a)},g(\vZ)+\varepsilon\}.
\]
Define
\[
e_i^{(a)}=g(\vZ_i)-(\bW_i^{(a)})^\top\bbetaW^{(a)}+\varepsilon_i,
\qquad \mu^{(a)}=E(e_i^{(a)}).
\]
Then $E[\bW^{(a)}\{e^{(a)}-\mu^{(a)}\}]=\vZero$.
Fix a nonzero contrast $\bell\in\bR^A$ with $\bone^\top\bell=0$ and set $\vL=(\bell^\top,\vZero^\top)$. Pairwise comparison of arms $k$ and $j$ corresponds to $\bell=\be_k-\be_j$.

\begin{lemma}[Design matrix limit]
\label{lem:ma-omega}
Under independent uniform randomization, or under a DCAR satisfying Condition \ref{cond:ma-balance},
\[
n^{-1}(\vG^{(r,a)})^\top\vG^{(r,a)}\overset P\to
\vOmega^{(a)}=
\begin{pmatrix}
A^{-1}\bI_A&A^{-1}\bone\vu^{(a)\top}\\
A^{-1}\vu^{(a)}\bone^\top&E\{\bW^{(a)}\bW^{(a)\top}\}
\end{pmatrix},
\qquad \vu^{(a)}=E\bW^{(a)},
\]
and $\vOmega^{(a)}\succ0$.
\end{lemma}

\begin{proof}
The treatment block converges because $n_k/n\to1/A$. For each component of $\bW$, Proposition \ref{prop:ma-imbalance} and the ordinary law of large numbers imply $n^{-1}\sum_iT_i^{(k)}\bW_i\to A^{-1}E\bW$; the lower-right block follows from the law of large numbers. To show positive definiteness, for $x\in\bR^A$ and $y$ conformable with $\bW$,
\[
(x^\top,y^\top)\vOmega^{(a)}(x^\top,y^\top)^\top
=A^{-1}\|x+\bone(\vu^{(a)\top}y)\|^2+y^\top\var(\bW^{(a)})y,
\]
which is positive for every nonzero $(x,y)$.
\end{proof}

\begin{lemma}\label{lem:ma-contrast}
For every $\bell$,
\[
\vL(\vOmega^{(a)})^{-1}=(A\bell^\top,\vZero^\top),
\qquad
\vL(\vOmega^{(a)})^{-1}\vL^\top=A\|\bell\|^2.
\]
\end{lemma}

\begin{proof}
Direct multiplication shows that $(A\bell^\top,\vZero^\top)\vOmega^{(a)}=\vL$, because $\bone^\top\bell=0$. The second identity follows by multiplying by $\vL^\top$.
\end{proof}

\begin{lemma}[Multi-arm residual imbalance]\label{lem:ma-residual}
For $r=\tCR$ or DCAR, CCAR satisfying Condition \ref{cond:ma-balance},
\[
\frac1{\sqrt n}\sum_{i=1}^n(\bT_i-\vc)\{e_i^{(a)}-\mu^{(a)}\}
\rightsquigarrow
\begin{cases}
\mathcal N_A\!\left(\vZero,A^{-1}\var(e^{(a)})\bP_A\right),&r=\tCR,\\
\mathcal N_A\!\left(\vZero,A^{-1}E\{\var(e^{(a)}\mid S)\}\bP_A\right),&r=\tDC.\\
\mathcal N_A\!\left(\vZero,A^{-1}(\sigma^{(\tCC)}_{e^{(a)}})^2\bP_A\right),&r=\tCC.
\end{cases}
\]
\end{lemma}

\begin{proof}
Under independent uniform randomization, the summands are i.i.d.; Lemma \ref{lem:ma-pi} gives the stated covariance and the multivariate CLT applies. For DCAR, put $h(\vZ)=g(\vZ)-\bW^{(a)\top}\bbetaW^{(a)}$ and
\[
U_i=h(\vZ_i)-E\{h(\vZ)\mid S_i\}+\varepsilon_i.
\]
Then $e_i^{(a)}-\mu^{(a)}=U_i+E\{h(\vZ)\mid S_i\}-E h(\vZ)$. The contribution of the second term is a finite linear combination of $\vD_n^{(\tDC)}(\chi_s)$ and is therefore $o_P(\sqrt n)$. With the two-stage filtration used in Proposition \ref{prop:ma-imbalance}, $(\bT_i-\vc)U_i$ is a martingale difference and its predictable covariance is
\[
\bPi(\bpi_i)\left[\var\{h(\vZ)\mid S_i\}+\sigma_\varepsilon^2\right].
\]
The averaging argument in that proposition gives the limit $A^{-1}E\{\var(e^{(a)}\mid S)\}\bP_A$. Conditional Lindeberg follows from square integrability of $U_i$ by truncation. The martingale CLT finishes the proof for DCAR. For CCAR, the results follows from an application of \cite[Theorem 4.3]{zhang2026asymptoticpropertiesmultitreatmentcovariate} and Proposition \ref{prop:ma-imbalance}. The proof is finshed.
\end{proof}

\begin{lemma}
\label{lem:ma-linear}
Under the conditions of Lemma \ref{lem:ma-omega},
\[
\sqrt n\left\{\vL\widehat{\vbeta}^{(r,a)}_n-\sum_{k=1}^A\ell_k\theta_k\right\}
=\frac{A}{\sqrt n}\sum_{i=1}^n\bell^\top(\bT_i-\vc)
\{e_i^{(a)}-\mu^{(a)}\}+o_P(1).
\]
\end{lemma}

\begin{proof}
The arm-specific OLS normal equations give $\widehat\theta_k=\bar Y_k-\overline{\bW}_k^\top\widehat{\bm\beta}_W$. Therefore the zero-sum contrast error equals
\[
\sum_{k=1}^A\ell_k\bar e_k-
\left\{\sum_{k=1}^A\ell_k\overline{\bW}_k\right\}^\top
(\widehat{\bm\beta}_W-\bbetaW^{(a)}).
\]
The sample moment limits in Lemma \ref{lem:ma-omega}, Lemma \ref{lem:ma-residual}, the ordinary law of large numbers, and the population normal equation $E[\bW^{(a)}\{e^{(a)}-\mu^{(a)}\}]=\vZero$ give $\widehat{\bm\beta}_W\overset{P}\to\bbetaW^{(a)}$, while Proposition \ref{prop:ma-imbalance} gives $\overline{\bW}_k-\vu^{(a)}=O_P(n^{-1/2})$. Hence the second term is $o_P(n^{-1/2})$. Since $\sum_k\ell_k=0$, $\sum_k\ell_k\bar e_k=\sum_k\ell_k(\bar e_k-\mu^{(a)})$. Moreover, $n_k=n/A+O_P(\sqrt n)$ under independent uniform randomization and $n_k=n/A+o_P(\sqrt n)$ under DCAR. Lemma \ref{lem:ma-residual} and the identity
$\sum_iT_i^{(k)}(e_i^{(a)}-\mu^{(a)})=A^{-1}\sum_i(e_i^{(a)}-\mu^{(a)})+[\sum_i(\bT_i-\vc)(e_i^{(a)}-\mu^{(a)})]_k$
show that every arm-centered residual sum is $O_P(\sqrt n)$. Replacing $n_k^{-1}$ by $A/n$ therefore costs $o_P(n^{-1/2})$ and gives the stated expansion.
\end{proof}

Recall the CR/DCAR/CCAR variances in Theorem \ref{theorem:combined}.

\begin{prop}\label{prop:ma-estimation}
For $r\in\{\tCR,\tDC,\tCC\}$ under the corresponding conditions above,
\[
\sqrt n\left\{\vL\widehat{\vbeta}^{(r,a)}_n-\sum_{k=1}^A\ell_k\theta_k\right\}
\rightsquigarrow
\mathcal N\!\left(0,A\|\bell\|^2(\tau^{(r,a)})^2\right).
\]
\end{prop}

\begin{proof}
Apply Lemma \ref{lem:ma-residual} in Lemma \ref{lem:ma-linear}. Because $\bell^\top\bP_A\bell=\|\bell\|^2$, the covariance factor is $A^2\cdot A^{-1}\|\bell\|^2(\tau^{(r,a)})^2$.
\end{proof}

Let $d_a$ be the dimension of the design matrix, let the homoscedastic residual estimator use divisor $n-A-d_a$, and define the usual Wald statistic $T_{n,\bell}^{(r,a)}$ with contrast $\vL$. The proof of Lemma \ref{lem:ma-linear} gives $\widehat{\bm\beta}_W\overset{P}\to\bbetaW^{(a)}$ and $\widehat\theta_k\overset{P}\to\theta_k+\mu^{(a)}$. Consequently,
\[
\frac{\text{Residual sum of squares}}n=\frac1n\sum_{i=1}^n\{e_i^{(a)}-\mu^{(a)}\}^2+o_P(1)
\overset P\to\var(e^{(a)})=(\tau^{(\tCR,a)})^2.
\]
Together with Lemmas \ref{lem:ma-omega} and \ref{lem:ma-contrast}, this gives
\[
n\widehat\sigma_e^2\vL(\vG^\top\vG)^{-1}\vL^\top
\overset P\to A\|\bell\|^2(\tau^{(\tCR,a)})^2.
\]
Because $\sigma_\varepsilon^2>0$, $\tau^{(\tCR,a)}>0$. Consequently, under the null, $T_{n,\bell}^{(r,a)}\rightsquigarrow(\tau^{(r,a)}/\tau^{(\tCR,a)})\xi$. Under the local alternative $\sum_k\ell_k\theta_k=\delta/\sqrt n$, its limit is
\[
\frac{\delta}{\sqrt A\|\bell\|\tau^{(\tCR,a)}}+
\frac{\tau^{(r,a)}}{\tau^{(\tCR,a)}}\xi.
\]
Since $E\{\var(e^{(a)}\mid S)\}\leq\var(e^{(a)})$, the DCAR Wald test is weakly conservative, with equality exactly when $E(e^{(a)}\mid S)$ is almost surely constant. For fixed $A$ and nonzero $\bell$, and $r,r'\in\{\tCR,\tDC, \tCC\}$ with positive limiting variances, tests calibrated to the same asymptotic size have ARE $(\tau^{(r',a')})^2/(\tau^{(r,a)})^2$.

\end{document}